\def\paperversion{final}
\newif\ifsinglecolumn\singlecolumnfalse
\newif\ifwidemargins\widemarginsfalse
\newif\ifwarning\warningfalse
\newif\ifshowcomments\showcommentsfalse
\newif\ifblinded\blindedfalse
\newif\ifreport\reportfalse
\newif\ifcopyrightspace\copyrightspacefalse
\newif\ifacknowledgments\acknowledgmentsfalse
\newif\ifshowpagenumbers\showpagenumberstrue
\newif\iffinalformat\finalformatfalse
\newif\ifweb\webfalse

\ifx\anonymoussubmission\undefined
\else
  \if\anonymoussubmission 1
    \blindedtrue
  \else
    \blindedfalse
  \fi
\fi
\IfFileExists{.blinded}{\blindedtrue}

\ifx\paperversion\xxxxundefined
\PackageError{paperversions}{*** No valid document version was specified.
Macro paperversions must be be defined as one of (markup, draft,
local, submission, final, web, tr, trdraft)}
\fi

\def\xxversion{\csname xx\paperversion\endcsname}
\newif\ifsawversion\sawversionfalse

\ifcase\xxversion\relax
    \widemarginstrue
    \singlecolumntrue
    \warningtrue
    \sawversiontrue
\or %draft
    \warningtrue
    \showcommentstrue
    \sawversiontrue
\or %local
    \warningtrue
    \showcommentsfalse
    \blindedfalse
    \sawversiontrue
    \acknowledgmentstrue
\or %submission
    \sawversiontrue
\or %final
    \blindedfalse
    \sawversiontrue
    \copyrightspacetrue
    \acknowledgmentstrue
    \showpagenumbersfalse
    \finalformattrue
\or %tr
    \singlecolumntrue
    \blindedfalse
    \sawversiontrue
    \reporttrue
    \acknowledgmentstrue
    \webtrue
\or %trdraft
    \blindedfalse
    \singlecolumntrue
    \showcommentstrue
    \sawversiontrue
    \reporttrue
\or %finaldraft
    \blindedfalse
    \showcommentstrue
    \copyrightspacetrue
    \acknowledgmentstrue
    \sawversiontrue
    \finalformattrue
\or %web
    \blindedfalse
    \sawversiontrue
    \copyrightspacetrue
    \acknowledgmentstrue
    \finalformattrue
    \webtrue
\or %blindtr
    \singlecolumntrue
    \blindedtrue
    \acknowledgmentsfalse
    \sawversiontrue
    \reporttrue
    \webtrue
\fi

\ifsawversion
\else
\fi

\let\xxversion=\undefined

\pdfoutput=1
\def\anonymoussubmission{1}
\newif\ifarxiv\arxivtrue

\documentclass[14pt]{article}
\usepackage{arxiv}
\usepackage[T1]{fontenc}
\usepackage[numbers]{natbib}
\usepackage[usenames,dvipsnames]{xcolor}
\usepackage{graphicx,eso-pic}
\usepackage{amsmath, amssymb, stmaryrd, utf8math, mathabx}
\usepackage{inconsolata}
\usepackage{ifthen}
\usepackage{ttquot,tikz}

\usepackage{thmtools}
\usepackage{thm-restate}
\usepackage{hyperref}
\usepackage{cleveref}
    \crefname{figure}{Figure}{Figures}
    \crefname{section}{Section}{Sections}
    \crefformat{footnote}{#2\footnotemark[#1]#3}
\usepackage{appendix}

\hypersetup{
	colorlinks,
	urlcolor=black,
	citecolor=black,
	filecolor=black,
	linkcolor=black
}

\ifshowcomments\relax\else
    \PassOptionsToPackage{disabled}{authcomments}
\fi
\usepackage{aplcomments}
\usepackage{paralist}
\usepackage{xspace}
\usepackage{xparse}
\usepackage{subcaption}
\usepackage{code}
\usepackage{booktabs}
\usepackage{pgfplots}
\pgfplotsset{compat=1.3}
\usepackage{dcolumn}
\usepackage{dashrule}
\usepackage{multirow}
\usepackage{colortbl}
\usepackage[normalem]{ulem}

\usepackage{listings}
\lstdefinestyle{small}{
	basicstyle=\footnotesize
}

\usetikzlibrary{calc}
\usetikzlibrary{backgrounds}
\usetikzlibrary{positioning}
\usetikzlibrary{decorations}
\usetikzlibrary{decorations.pathreplacing}
\usetikzlibrary{decorations.pathmorphing}
\usetikzlibrary{shapes}
\usetikzlibrary{arrows}
\usetikzlibrary{arrows.meta}
\usetikzlibrary{fit}

\usepackage{suffix}
\usepackage{mathpartir}
\usepackage{amsmath,amsfonts,amssymb,amsthm,cases}
\usepackage{mathtools}

\usepackage{ttquot,macros}
\usepackage{etoolbox}

\newcommand\ourtitle{A Calculus for Flow-Limited Authorization}

\theoremstyle{plain}
\newtheorem{theorem}{Theorem}
\newtheorem{proposition}{Proposition}
\newtheorem{lemma}{Lemma}
\newtheorem{corollary}{Corollary}

\theoremstyle{definition}
\newtheorem{definition}{Definition}

\theoremstyle{remark}

\usepackage{enumitem}

\usepackage{authblk}
\renewcommand{\headeright}{Technical Report}

\begin{document}

\author[a]{Owen Arden}
\author[b]{Anitha Gollamudi}
\author[c]{Ethan Cecchetti}
\author[b]{Stephen Chong}
\author[c]{Andrew C. Myers}
\affil[a]{UC Santa Cruz, Santa Cruz, CA, USA \authorcr \hspace{.05in} \textit{Email:} \url{owen@soe.ucsc.edu}}
\affil[b]{Harvard University, Cambridge, MA, USA \authorcr \hspace{.05in} \textit{Email:} \url{agollamudi@g.harvard.edu}, \url{chong@seas.harvard.edu}}
\affil[c]{Cornell University, Ithaca, NY, USA \authorcr \hspace{.05in} \textit{Email:} \url{ethan@cs.cornell.edu}, \url{andru@cs.cornell.edu}}

\title{\ourtitle}
\maketitle

\begin{abstract}
Real-world applications routinely make authorization decisions
based on dynamic computation.  Reasoning about dynamically computed
authority is challenging. Integrity of the system might be compromised
if attackers can improperly influence the authorizing computation.
Confidentiality can also be compromised by authorization, since
authorization decisions are often based on sensitive
data such as membership lists and passwords.  Previous formal models for
authorization do not fully address the security implications of
permitting trust relationships to change, which limits their ability
to reason about authority that derives from dynamic computation.
Our goal is an approach to constructing dynamic authorization mechanisms that do not
violate confidentiality or integrity.

The Flow-Limited Authorization Calculus (FLAC) is
a simple, expressive model for reasoning about dynamic
authorization as well as an information flow control language for 
securely implementing various authorization mechanisms. 
FLAC combines the insights of two previous models: it extends the Dependency
Core Calculus with features made possible by the Flow-Limited Authorization Model.
FLAC provides strong end-to-end information security guarantees even
for programs that incorporate and implement rich dynamic authorization 
mechanisms. These guarantees include noninterference and robust
declassification, which prevent attackers from influencing information
disclosures in unauthorized ways.
We prove these security properties formally for all FLAC programs and explore the
expressiveness of FLAC with several examples.

\end{abstract}

\section{Introduction}
Authorization mechanisms are critical components in all
distributed systems.  The policies enforced by these mechanisms constrain what
computation may be safely executed, and therefore an expressive policy language is
important. Expressive mechanisms for authorization have been an active
research area. A variety of approaches have been developed, including authorization
logics~\cite{labw91,abadi06,nal}, often implemented with cryptographic mechanisms~\cite{SPKI-SDSIb,macaroons};
role-based access control (RBAC)~\cite{rbac}; and
trust management~\cite{li2002design,shtz06,RTI}.

However, the security guarantees of authorization mechanisms are usually
analyzed using formal models that abstract away the computation and
communication performed by the system. Developers must take great care to
faithfully preserve the (often implicit)
assumptions of the model, not only when implementing
authorization mechanisms, but also when employing them.  Simplifying
abstractions can help 
extract formal security guarantees, but
abstractions can also obscure the challenges of implementing and using an abstraction securely.
This disconnect between abstraction and implementation 
can lead to vulnerabilities and covert channels that allow attackers to 
leak or corrupt information.

A common blind spot in many authorization models is confidentiality. Most models
cannot express authorization policies that are confidential or are based on
confidential data.  Real systems, however, use confidential data for
authorization all the time: users on social networks receive access to photos
based on friend lists, frequent fliers receive tickets based on credit card
purchase histories, and doctors exchange patient data while keeping
doctor--patient relationships confidential.  While many models can ensure, for
instance, that only friends are permitted to access a photo, few can say
anything about the secondary goal of preserving the confidentiality of the
friend list. Such authorization schemes may fundamentally require _some_ information to
be disclosed, but failing to detect these disclosures can lead to unintentional
leaks.

Authorization without integrity is meaningless, so formal authorization models are typically
better at enforcing integrity.
However, many formal models
make unreasonable or unintuitive assumptions about integrity.  For instance, in
many models (e.g., \cite{labw91}, \cite{abadi06}, \cite{li2002design})
authorization policies either do not change or change only when modified by a
trusted administrator.  This is a reasonable assumption in centralized systems
where such an administrator will always exist, but in decentralized systems,
there may be no single entity that is trusted by all other entities.

Even in centralized systems, administrators must be careful when performing
updates based on partially trusted information, since malicious users may 
try to confuse or mislead the administrator into carrying out an attack on their behalf.
Unfortunately, existing authorization models offer little help to administrators that
need to reason about how attackers may have influenced security-critical update
operations.

Developers need a better programming model for implementing expressive dynamic
authorization mechanisms. Errors that undermine the security of these mechanisms
are common~\cite{cwe25}, so we want to be able to verify their security.
We argue that information flow control (IFC) is a lightweight, useful tool for
building secure authorization mechanisms
since it offers compositional, end-to-end security guarantees.  However,
applying IFC to authorization mechanisms in a meaningful
way requires building on a theory that integrates authority and information security.
In this work, we show how to embed
such a theory into a programming model, so that dynamic authorization
mechanisms—as well as the programs that employ them—can be statically
verified.  

Approaching the verification of dynamic authorization mechanisms from this
perspective is attractive for two reasons.  First,
it gives a model for building
secure authorization mechanisms by construction rather than verifying them after
the fact. This model offers programmers insight into the subtle
interaction between information flow and authorization, and helps programmers address
problems early, during the design process.  Second, it addresses a core
weakness lurking at the heart of existing language-based security schemes:
that the underlying policies may change in a way that breaks security.  By statically verifying 
the information security of dynamic
authorization mechanisms, we expand the real-world scenarios in which language-based 
information flow control is useful and strengthen its security guarantees.

We demonstrate that such an embedding is possible by presenting a core language for
authorization and information flow control, called the Flow-Limited
Authorization Calculus (FLAC).  FLAC is a functional language for designing
and verifying decentralized authorization protocols. FLAC is inspired by the
Polymorphic Dependency Core Calculus~\cite{abadi06} (DCC).\footnote{DCC was first presented in~\cite{ccd99}. 
We use the abbreviation DCC to refer
to the extension to polymorphic types in~\cite{abadi06}.}  Abadi develops DCC as an
authorization logic,
but DCC is limited to static trust relationships defined
externally to DCC programs by a lattice of principals. FLAC supports dynamic
authorization by building on the Flow-Limited Authorization Model
(FLAM)~\cite{flam}, which unifies reasoning about authority, confidentiality,
and integrity.
Furthermore, FLAC is a language for information flow control.  It uses FLAM's
principal model and FLAM's logical reasoning rules to define an operational model
and type system for authorization computations that preserve information
security.
George~\cite{mdg-thesis} also extends DCC's monadic approach to model information flow
in authorization mechanisms, with more explicit modeling of a distributed execution environment.

The types in a FLAC program can be considered
propositions~\cite{wadler2015propositions} in an authorization logic, and the
programs can be considered proofs that the proposition holds.  Well-typed FLAC programs are both
proofs of secure information flow and proofs of authorization, 
ensuring the confidentiality and integrity of not only data, but also authorization policies.

FLAC is useful from a logical perspective, but also serves as a core
programming model for real language implementations.  Since FLAC programs can
dynamically authorize computation and flows of information, FLAC applies to more
realistic settings than previous authorization logics. Thus FLAC offers more than 
a language and type system for proving propositions—FLAC programs do useful computation.

This paper makes the following contributions.
\begin{itemize}
  \item We define FLAC, a language, type system, and semantics for dynamic authorization
   mechanisms with strong information security:
  \begin{itemize}
  \item Programs in low-integrity contexts exhibit _noninterference_, 
    ensuring attackers cannot leak or corrupt information, and cannot subvert authorization mechanisms.
  \item Programs in higher-integrity contexts exhibit _robust declassification_, 
     ensuring attackers cannot influence authorized disclosures of information.
  \end{itemize}
  \item We present two authorization mechanisms implemented in FLAC, commitment
schemes and bearer credentials, and demonstrate that FLAC ensures
the programs that use these mechanisms preserve the desired
confidentiality and integrity properties.
\end{itemize}

We have organized our discussion of FLAC as follows.  
Section~\ref{sec:ex} introduces commitment schemes and bearer credentials, two examples of 
dynamic authorization mechanisms we use to explore the features of FLAC.
Section~\ref{sec:flam} reviews the FLAM principal lattice~\cite{flam}, and
Section~\ref{sec:flac} defines the FLAC language and type system. 
FLAC implementations of the dynamic authorization examples are
presented in Section~\ref{sec:exre}, and their properties are examined.
Section~\ref{sec:pftheory} explores aspects of FLAC's proof theory,
and Section~\ref{sec:props} discusses semantic security guarantees
of FLAC programs, including noninterference and robust declassification.
We explore related work in Section~\ref{sec:relwork} and conclude in Section~\ref{sec:concl}.

Many of the contributions of this paper were previously published by
Arden and Myers~\cite{flac}.  This article expands upon and
strengthens the formal results, as well as corrects several technical
errors in that work. This includes a more detailed treatment of how FLAC
constrains delegations and a more general noninterference theorem 
on a new formal result regarding the compartmentalization of delegations.
Most of the semantic security proofs in Section~\ref{sec:props}
are either new to this work or
were redeveloped from scratch to account for changes made to the 
semantics and type system.

Furthermore, subsequent work based on FLAC by
Cecchetti {\it et al.}~\cite{nmifc} and Gollamudi {\it et
al.}~\cite{dflate} exposed alternate design decisions that improve the
connection between FLAC and cryptographic implementations of its
protection abstractions.  
We feel some of these changes are objectively better
than the original abstractions, and so incorporate them into the core 
FLAC formalism so that future work based on FLAC may benefit.
Significant departures from the original formalization are footnoted, 
but minor changes and corrections are included without comment.

\section{Dynamic authorization mechanisms} 
\label{sec:ex}
Dynamic authorization is challenging to
implement correctly since authority, confidentiality, and integrity
interact in subtle ways. FLAC helps programmers
securely implement both authorization mechanisms and the code that uses them.
FLAC types support the definition of compositional security
abstractions, and vulnerabilities in the implementations of these abstractions are
caught statically. Further, the guarantees offered by FLAC simplify
reasoning about the security properties of these abstractions.

We illustrate the usefulness and expressive power of FLAC using two
important security mechanisms: commitment schemes and bearer
credentials. We show in Section~\ref{sec:exre} that these
mechanisms can be implemented using FLAC, and that their security goals
are easily verified in the context of FLAC.

\subsection{Commitment schemes} 
\label{sec:ex1}
 
A commitment scheme~\cite{commitment} allows one party to give another party a
``commitment'' to a secret value without revealing the value.  The committing
party may later reveal the secret in a way that convinces the receiver that the
revealed value is the value originally committed.

Commitment schemes provide three essential operations: "commit",
"reveal", and "open".\footnote{This commitment scheme example
differs from the one presented in Arden and Myers~\cite{flac}, which is not 
compatible with changes to the type system presented here.  
Furthermore, we feel the API presented here is a better representation of the operations present in
cryptographic commitment schemes.}
Suppose $p$ wants to commit to a value to principal $q$. First, $p$
applies "commit" to the value and provides the result to $q$ without revealing  
the committed value to $q$.  When $p$ wishes to reveal the value, it gives $q$ a "reveal" operation
$q$ can use to open the previously sent commitment.  
Then $q$ uses "reveal" to "open" the committed value, finally revealing it.  
In a cryptographic implementation,
the "reveal" operation might correspond to the secret $p$ used to encrypt the commitment, and the "open" 
operation might correspond to $q$ using that secret to decrypt the commitment.

A commitment scheme must have several properties in order to be secure.  First,
$q$ should not be able to open a value that $p$ has not committed to,
since this could allow $q$ to manipulate $p$ to "open" a value it had not
committed to.  Second, $q$ should not be able to learn a secret of $p$ that has not been
committed to or revealed by $p$.  Third, $p$ should not be able to modify the committed value 
after it is received by $q$.  Specifically, the "reveal" operation should only reveal the 
committed value and not modify it in any way.

One might wonder why a programmer would bother to create high-level
_implementations_ of operations like "commit", "reveal", and "open".  Why not
simply treat these as primitive operations and give them type signatures so that
programs using them can be type-checked with respect to those signatures?  The
answer is that an error in a type signature could lead to a serious
vulnerability. Therefore, we want more assurance that the type signatures are
correct.  Modeling such operations in FLAC is often easy and ensures that
the type signature is consistent with a set of assumptions about existing trust
relationships and the information flow context the operations are used within.
These FLAC-based implementations serve as language-based specifications of the security
properties required by implementations that use cryptography or trusted third parties.

\subsection{Bearer credentials with caveats} 
\label{sec:ex2} A bearer credential is a capability
that grants authority to any entity that possesses it.  Many authorization
mechanisms used in distributed systems 
employ bearer credentials in some form.  Browser cookies that store
session tokens are one example: after a website authenticates a user's
identity, it gives the user a token to use in subsequent interactions.
Since it is infeasible for attackers to guess the token, the website grants the
authority of the user to any requests that include the token.

Bearer credentials create an information security conundrum for authorization
mechanisms. Though they efficiently control access to restricted
resources, they create vulnerabilities and introduce
covert channels when used incorrectly. For example, suppose Alice shares a remotely hosted photo with
her friends by giving them a credential to access the photo.  Giving a
friend such a credential doesn't disclose their friendship, but each friend
that accesses the photo implicitly discloses the friendship to the hosting
service.  Such covert channels are pervasive, both in classic distributed
authorization mechanisms like SPKI/SDSI~\cite{SPKI-SDSIb}, as well as in
more recent ones like Macaroons~\cite{macaroons}.  

Bearer credentials can also lead to vulnerabilities if they are leaked.  If an
attacker obtains a credential, it can exploit the authority of the credential.
Thus, to limit the authority of a credential, approaches like SPKI/SDSI and Macaroons
provide _constrained delegation_ in which a newly issued credential
attenuates the authority of an existing one by adding _caveats_.  Caveats require
additional properties to hold for the bearer to be granted authority.
Session tokens, for example, might have a caveat that restricts the source IP address
or encodes an expiration time.  As pointed out by Birgisson et
al.~\cite{macaroons}, caveats themselves can introduce covert channels 
if the properties reveal sensitive information.

FLAC is an effective framework for reasoning about bearer credentials with caveats since
it captures the flow of credentials in programs as well as the sensitivity
of the information the credentials and caveats derive from.  We can reason about credentials
and the programs that use them in FLAC with an approach similar to that used for
commitment schemes.
That we can do so in a straightforward way is somewhat
remarkable: prior formalizations of credential mechanisms (e.g., \cite{macaroons,hk00spki,SecPAL}) 
usually do not consider confidentiality nor provide 
end-to-end guarantees about credential propagation.  

\section{The FLAM Principal Lattice}

\label{sec:flam} Like many models, FLAM uses _principals_ to represent the
authority of all entities relevant to a system. However, FLAM's principals and
their algebraic properties are richer than in most models, so we briefly review
the FLAM principal model and notation. Further details are found in
the earlier paper~\cite{flam}.

Primitive principals such as "Alice", "Bob", etc., are represented as elements
$n$ of a (potentially infinite) set of names $\mathcal N$.\footnote{Using
$\mathcal N$ as the set of all names is convenient in our formal calculus, but a
general-purpose language based on FLAC may wish to dynamically allocate names at
runtime.  Since knowing or using a principal's name holds no special privilege
in FLAC, this presents no fundamental difficulties. To use dynamically allocated 
principals in type signatures, however, the language's type system 
should support types in which principal names may be existentially quantified.}
In addition to these names, FLAM uses $⊤$ to
represent a universally trusted primitive principal and $⊥$ to represent a universally
untrusted primitive principal. The combined authority of two principals, $p$ and $q$, is
represented by the authority conjunction $p ∧ q$, whereas the authority of either $p$ or
$q$ is the disjunction $p ∨ q$.

Unlike principals in other models,
FLAM principals also represent information flow
policies.  The confidentiality of principal $p$ is represented by the principal
$p^{→}$, called $p$'s confidentiality projection.  It denotes the authority necessary to _learn_ anything
$p$ can learn.
The integrity of principal $p$ is represented by
$p^{←}$, called $p$'s integrity projection. It denotes the authority to _influence_ anything $p$ can influence. 

These projections, conjunctions, and disjunctions allow us to construct the set of all principals
from any set of names $\mathcal N$.
The set $\mathcal N \cup \{⊤, ⊥\}$ under the
syntax operators\footnote{FLAM defines an additional set of operators called _ownership
projections_, which we omit here to simplify our presentation.} $∧, ∨, ←, →$
forms a set $\mathcal L$.  The equivalence classes of this set under the 
_acts-for relation_\footnote{FLAM's acts-for relation 
  is inspired by the acts-for relation defined by the Decentralized Label Model~\cite{ml-ssp-98} (DLM).
  In the DLM, the principal with the most authority is referred to as $⊤$, and the upper bound of the authority  
  of two principals is written $∧$, and the lower bound with $∨$. Unfortunately, this departs from the usual
  notational conventions for lattices, where $∧$ is used for lower bounds (meets), and $∨$ is used for upper bounds.
  We have stubbornly held on to the DLM-based acts-for notation, but this sometimes creates confusion for those more familiar with the standard lattice notation.
} $≽$ defined in Figure~\ref{fig:static}\footnote{The original FLAC formalization
  presented a set of static rules derived from FLAM's principal algebra, but still relied on
  FLAM's algebraic identities for completeness. Here we present a complete set of static
  rules and have proven their equivalence to FLAM's principal algebra (without ownership projections) in Coq~\cite{SAFproof}.
 } %
form a complete distributive lattice.
 Using these rules we can derive the equivalences from FLAM's principal algebra.  
 Two principals are considered equivalent if they act for each other:
 $$
      p ≡ q \triangleq \stafjudge{\L}{p}{q} \text{ and } \stafjudge{\L}{q}{p}
 $$
We have proven in Coq~\cite{SAFproof} that this definition is equivalent to the
algebraic definition (minus ownership) used in the FLAM Coq formalization~\cite{flamtr}.
We write operators $←,→$ with higher precedence than $∧,∨$; for instance,
$p∧q^{←}$ is the same as writing $p∧(q^{←})$. 
Projections distribute over $∧$ and $∨$ 
(rules \ruleref{ProjDistConj} and \ruleref{ProjDistDisj})
so, for example, $(p∧q)^{←} ≡ (p^{←}∧q^{←})$.

All authority may be represented as some combination of confidentiality and integrity. 
Any principal $p$ is equivalent to, via rules \ruleref{ProjR} and \ruleref{ConjBasis},
the conjunction of its confidentiality and integrity authority: $p^{→} ∧ p^{←}$. In fact, any
principal can be normalized~\cite{flam} to $q^{→} ∧ r^{←}$ for some $q$ and $r$.
For example, $\Alice^{→} ∧ \Bob$ is equivalent to $(\Alice ∧ \Bob)^{→} ∧ Bob^{←}$.
The confidentiality and integrity authority of
principals are disjoint (rule \ruleref{ProjBasis}), so the confidentiality projection of an
integrity projection is $⊥$ and vice-versa: $(p^{←})^{→} ≡ ⊥ ≡ (p^{→})^{←}$.

\begin{figure}
  {\footnotesize
\begin{flalign*}
& \boxed{\stafjudge{\L}{p}{q}} &
\end{flalign*}
    \begin{mathpar}
    \stafrule[\L]{Bot}{}{p}{⊥}
    \and
    \stafrule[\L]{Top}{}{⊤}{p}
    \and
    \stafrule[\L]{Refl}{}{p}{p}
    \and
    \stafrule[\L]{Trans}{
      \stafjudge{\L}{\!p\!}{\!q} \\
      \stafjudge{\L}{\!q\!}{\!r}
    }{\!p\!}{\!r}
    
     \\ \\
    \stafrule[\L]{Proj}{\stafjudge{\L}{p}{q}}{p^{π}}{q^{π}}
    \and
    \stafrule[\L]{ProjR}{}{p}{p^{π}}
    \and
    \stafrule[\L]{ProjIdemp}{}{(p^{π})^{π}}{p^{π}}
    \and
    \stafrule[\L]{ProjBasis}{π ≠ π'}{⊥}{(p^{π})^{π'}}
    \and
    \stafrule[\L]{ProjDistConj}{}{p^{π} ∧ q^{π}}{(p ∧ q)^{π}} 
    \and
    \stafrule[\L]{ProjDistDisj}{}{(p ∨ q)^{π}}{p^{π} ∨ q^{π}} 

    \\\\
    
    \stafrule[\L]{ConjL}{
      \stafjudge{\L}{p_k}{p} \\\\
      k ∈ \{1,2\}
    }{p_1 ∧ p_2}{p}
    \and
    \stafrule[\L]{ConjR}{
      \stafjudge{\L}{p}{p_1} \\\\
      \stafjudge{\L}{p}{p_2}
    }{p}{p_1 ∧ p_2}
    \and
    \stafrule[\L]{ConjBasis}{}{p^{→} ∧ p^{←}}{p} 
    \and
    \stafrule[\L]{ConjDistDisjL}{}{(p ∧ q) ∨ (p ∧ r)}{p ∧ (q ∨ r)}
    \and
    \stafrule[\L]{ConjDistDisjR}{}{p ∧ (q ∨ r)}{(p ∧ q) ∨ (p ∧ r)} 

    \\\\
    \stafrule[\L]{DisjL}{
      \stafjudge{\L}{p_1}{p} \\\\
      \stafjudge{\L}{p_2}{p}
    }{p_1 ∨ p_2}{p}
    \and
    \stafrule[\L]{DisjR}{
      \stafjudge{\L}{p}{p_k} \\\\
      k ∈ \{1,2\}
    }{p}{p_1 ∨ p_2}
    \and
    \stafrule[\L]{DisjBasis}{}{⊥}{p^{→} ∨ q^{←}}
    \and
    \stafrule[\L]{DisjDistConjL}{}{(p ∨ q) ∧ (p ∨ r)}{p ∨ (q ∧ r)}
    \and
    \stafrule[\L]{DisjDistConjR}{}{p ∨ (q ∧ r)}{(p ∨ q) ∧ (p ∨ r)} 
    
    \hfill
    \end{mathpar}
  }
\caption{Static principal lattice rules. The projection \protect{$\pi$} may be either 
  confidentiality (\protect{$\rightarrow$}) or integrity (\protect{$\leftarrow$}).
  Adapted from (and equivalent to) the non-ownership fragment of FLAM's
  principal algebra~\cite{flam}.
}
  \label{fig:static}
\end{figure}

An advantage of this model is that secure information flow can be defined in
terms of authority.  An information flow policy $q$ is at least as _restrictive_
as a policy $p$ if $q$ has at least the confidentiality authority $p^{→}$ and $p$
has at least the integrity authority $q^{←}$. This relationship between the confidentiality
and integrity of $p$ and $q$ reflects the usual duality seen in
information flow control~\cite{integrity}. As in~\cite{flam}, we use the
following shorthand for relating principals by policy restrictiveness:
\begin{align*}
    p ⊑ q &≜ (p^{←}∧q^{→}) ≽ (q^{←}∧ p^{→}) \\
    p ⊔ q &≜(p ∧ q)^{→} ∧ (p ∨ q)^{←} \\
    p ⊓ q &≜  (p ∨ q)^{→} ∧ (p ∧ q)^{←} 
\end{align*}
Thus, $p ⊑ q$ indicates the direction of secure information flow: from $p$ to
$q$.  The information flow join $p ⊔ q$ is the least restrictive principal that
both $p$ and $q$ flow to, and the information flow meet $p ⊓ q$ is the most
restrictive principal that flows to both $p$ and $q$.

An interesting feature of this definition is that the equivalence classes 
under $≽$ and $⊑$ are the same.
\begin{align*}
  p ≡ q &= \stafjudge{\L}{p}{q} \text{ and } \stafjudge{\L}{q}{p} \\
     &= \stafjudge{\L}{p^{→} ∧ p^{←}}{q^{→} ∧ q^{←}} \text{ and } \stafjudge{\L}{q^{→} ∧ q^{←}}{p^{→} ∧ p^{←}} \\
     &= \stafjudge{\L}{p^{→}}{q^{→}}  \text{ and } \stafjudge{\L}{q^{→}}{p^{→}}\\
     & \quad \text{ and } \stafjudge{\L}{p^{←}}{q^{←}}  \text{ and } \stafjudge{\L}{q^{←}}{p^{←}}\\
     &= \stafjudge{\L}{p^{←} ∧ q^{→}}{q^{←} ∧ p^{→}} \text{ and } \stafjudge{\L}{q^{←} ∧ p^{→}}{p^{←} ∧ q^{→}} \\
     & =\stflowjudge{\L}{p}{q} \text{ and }\stflowjudge{\L}{q}{p} 
\end{align*}
These equivalence classes also form a (separate) complete distributive lattice
ordered by $⊑$ where $⊔$ is a join and $⊓$ is a meet. In this lattice, the _secret and untrusted_
principal $⊤^{→} ∧ ⊥^{←}$ is the top element since it is the most restrictive information flow policy.
Likewise, the _public and trusted_ principal $⊥^{→} ∧ ⊤^{←}$ is the bottom element since it is the
least restrictive policy.  Because of this tight relationship between $≽$ and $⊑$ we can
use the inference rules in Figure~\ref{fig:static} to reason about the relationship
between principals in both the authority lattice and the information flow lattice.
We often write projected principals by themselves, but 
these principals are equivalent to the conjunction of themselves
with the bottom element of the missing projection: e.g., $p^{→} ≡ p^{→} ∧ ⊥^{←}$.

In FLAM, the ability to ``speak for'' another principal is an integrity
relationship between principals.  This makes sense intuitively, because speaking
for another principal influences that principal's trust relationships and
information flow policies.  FLAM defines the _voice_ of a principal
$p$, written $\voice{p}$,
as the integrity necessary to speak for that
principal. Given a principal expressed in normal form\footnote{In
normal form, a principal is the conjunction of a confidentiality
principal and an integrity principal. See~\cite{flam} for details.} as $\confid q
\tjoin \integ r$, the voice of that principal is
\[
    \voice{\confid q \tjoin \integ r} \triangleq \integ q \tjoin \integ r
\] For example, the voice of $"Alice"$, $∇("Alice")$,  is $"Alice"^{←}$. The
voice of $"Alice"$'s confidentiality $∇("Alice"^{→})$ is also $"Alice"^{←}$.

All primitive principals speak for themselves: e.g., $\stafjudge{\L}{\Alice}{∇(Alice)}$,
but principals with asymmetric confidentiality and integrity authority may not:
\begin{align*}
& \notstafjudge{\L}{\Alice^{→}∧\Bob^{←}}{∇(\Alice^{→}∧\Bob^{←})} & \\
& \notstafjudge{\L}{\Alice^{→}∧\Bob^{←}}{\Alice^{←}∧\Bob^{←}} &
\end{align*}

\section{Flow-Limited Authorization Calculus}
\label{sec:flac}

FLAC uses information flow to reason about the security implications of dynamically computed authority.
Like previous information-flow type systems~\cite{sm-jsac},
FLAC incorporates types for
reasoning about information flow, but FLAC's type system goes further
by using Flow-Limited Authorization~\cite{flam}
to ensure that principals cannot use FLAC programs to exceed their
authority, or to leak or corrupt information.
FLAC is based on DCC~\cite{abadi06}, but unlike DCC, \lang supports reasoning
about authority that derives from the evaluation of \lang terms.  In contrast, all
authority in DCC derives from trust relationships defined by a fixed, external lattice
of principals.  Thus, using an approach based on DCC in systems where trust
relationships change dynamically could introduce vulnerabilities like
delegation loopholes, probing and poaching attacks, and authorization side channels~\cite{flam}.

\begin{figure}
\[
\small
\begin{array}{rcl}
\multicolumn{3}{l}{ n \in \mathcal N \text{  (principal names)}} \\
\multicolumn{3}{l}{ x \in \mathcal V \text{  (variable names)}} \\
\\
p,ℓ,\pc &::=&  n \sep \top \sep \bot \sep \confid{p} \sep \integ{p} \sep p ∧ p\sep p ∨ p  \\[0.4em] %
τ &::=&  \aftype{p}{p} \sep \voidtype \sep \sumtype{τ}{τ} \sep \prodtype{τ}{τ} \\[0.4em]  
  & &\sep  \func{τ}{\pc}{τ} \sep \says{ℓ}{τ} \sep X \sep \tfuncpc{X}{\pc}{τ} \\[0.4em]  
v &::=&  () \sep \pair{w}{w} \sep \delexp{p}{p} \sep \returng{ℓ}{w} \sep \inji{w}  \\[0.4em]
  &&  \sep \lamc{x}{τ}{\pc}{e} \; \textsf{(closed)} \sep \tlam{X}{\pc}{e} \; \textsf{(closed)}\\[0.4em]  
w &::=& v \sep \where{w}{v} \\[0.4em]
e &::=&  x \sep w \sep  e~e \sep \pair{e}{e} \sep \return{ℓ}{e} \\[0.4em]
   && \sep e~τ \sep \proji{e} \sep \inji{e} \\[0.4em]
  &&  \sep \lamc{x}{τ}{\pc}{e}  \sep \tlam{X}{\pc}{e} \\[0.4em]  
   && \sep \casexp{e}{x}{e}{e} \\[0.4em]  
   && \sep \bind{x}{e}{e} \sep \assert{e}{e} \\[0.4em] 
  && \sep \where{e}{v} \\[0.4em] 
\end{array}
\hfill
\]
\caption{\lang syntax. Terms using \texttt{where} are syntactically prohibited
in the source language and are produced only during evaluation.}
\label{fig:syntax}
\end{figure}

\begin{figure}
{\footnotesize
\begin{flalign*}
& \boxed{
e \stepsone e'
} &
\end{flalign*}
\begin{mathpar}
\erule{E-App}{}{(\lamc{x}{τ}{\pc}{e})~w }{\subst{e}{x}{w}}{}

\erule{E-TApp}{}{(\tlam{X}{\pc}{e})~τ}{\subst{e}{X}{τ}}{}

\erule{E-Unpair}{}{\proji{\pair{w₁}{w₂}}}{wᵢ}{}

\erule{E-Case}{}{(\casexp{(\inji{w})}{x}{e_1}{e_2}) }{ \subst{eᵢ}{x}{w}}{}

\erule{E-BindM}{}{\bind{x}{\returng{ℓ}{w}}{e}}{\subst{e}{x}{w}}{}

\erule{E-Assume}{}{\assume{\delexp{p}{q}}{e}}{\where{e}{\delexp{p}{q}}}{}

\erule{E-UnitM}{}{\return{\ell}{w}}{\returng{\ell}{w}}{}

\erule{E-Eval}{e \stepsone e'}{E[e]}{E[e']}{} 
\\
\begin{array}{rcl}
  E & ::= &  [\cdot] \sep E~e \sep w~E  \sep E~τ  \sep \pair{E}{e} \sep \pair{w}{E} \sep \proji{E} \sep \inji{E} \sep \return{ℓ}{E} \\[0.4em]
  & \sep &  \bind{x}{E}{e} \sep \assume{E}{e}   \sep \casexp{E}{x}{e}{e}  \sep \where{E}{v} \\[0.4em]
\end{array}
\hfill   
\end{mathpar}
}
\caption{\lang operational semantics}
\label{fig:semantics}
\end{figure}

Figure~\ref{fig:syntax} defines the \lang syntax.  The core \lang operational
semantics and evaluation contexts~\cite{wf94} in
Figure~\ref{fig:semantics} are mostly standard except for \ruleref{E-Assume} and \ruleref{E-UnitM},
which we discuss below, along with additional rules that handle the propagation of the "where" terms introduced by \ruleref{E-Assume}.

\ifreport
\begin{figure}[h]
\else
\begin{figure}
\fi
{\footnotesize
\begin{flalign*}
& \boxed{\TValGpc{e}{τ}} &
\end{flalign*}
\begin{mathpar}
\Rule{Var}{}{\TValP{Γ,x:τ,Γ';\pc}{x}{τ}}{x ∉ \text{dom}~Γ'}

\Rule{Unit}{}{\TValGpc{\void}{\voidtype}}

\Rule{Del}{}{\TValGpc{\delexp{p}{q}}{\aftype{p}{q}}}

\Rule{Lam}{\TValP{Γ \cont x\ty \tau_1;\pc'}{e}{τ_2}}
     {\TValGpc{\lamc{x}{τ_1}{\pc'}{e}}
               {(\func{τ_1}{\pc'}{τ_2})}}

\Rule{TLam}{\TValP{Γ,X;\pc'}{e}{τ}}
{\TValGpc{\tlam{X}{\pc'}{e}}{\tfuncpc{X}{\pc'}{τ}}}

\Rule{App}{
      \TValGpc{e}{(\func{τ_1}{\pc'}{τ_2})} \\\\
      \TValGpc{e'}{τ_1} \\
      \rflowjudge{Π}{\pc}{\pc'}
     } 
     {\TValGpc{(e~e')}{τ_2}}

\Rule{TApp}{\TValGpc{e}{\tfuncpc{X}{\pc'}{τ}} \\\\
      \rflowjudge{Π}{\pc}{\pc' }}
{\TValGpc{(e~τ')}{\subst{τ}{X}{τ'}}}{τ' \text{ well-formed in } Γ}

\Rule{Pair}{
  \TValGpc{e_1}{τ_1} \\ \TValGpc{e_2}{τ_2}}
     {\TValGpc{\pair{e_1}{e_2}}{\prodtype{τ_1}{τ_2}}}

\Rule{Unpair}{
  \TValGpc{e}{\prodtype{τ_1}{τ_2}}} 
   {\TValGpc{(\proji{e})}{τ_i}}

\Rule{Inj}{\TValGpc{e}{τᵢ} \\ i \in \{1, 2\} }
   {\TValGpc{(\inji{e})}{\sumtype{τ_1}{τ_2}}}

\Rule{Case}{
  \TValGpc{e}{\sumtype{τ_1}{τ_2}} \\
  \protjudge{Π}{\pc}{τ} \\\\
  \TValP{Γ,x:\tau_1;\pc}{e_1}{τ} \\
  \TValP{Γ,x:\tau_2;\pc}{e_2}{τ}} 
   {\TValGpc{\casexp{e}{x}{e_1}{e_2}}{τ}}

\Rule{UnitM}{
  \TValGpc{e}{τ} \\
  \rflowjudge{Π}{\pc}{\ell}
  }
   {\TValGpc{\return{ℓ}{e}}{\says{ℓ}{τ}}}

\Rule{Sealed}{
  \TValGpc{v}{τ}}
   {\TValGpc{\returng{ℓ}{v}}{\says{ℓ}{τ}}}

\Rule{BindM}{
  \TValGpc{e}{\says{ℓ}{τ'}} \\
  \TValP{Γ,x:τ';\pc ⊔ ℓ}{e'}{τ} \\\\
  \protjudge*{\pc ⊔ ℓ}{τ}}
   {\TValGpc{\bind{x}{e}{e'}}{τ}}

\Rule{Assume}{
   \TValGpc{e}{\aftype{p}{q}} \\\\
   \rafjudge{Π}{\pc}{\voice{q}} \\
   \rafjudge{Π}{\voice{\confid{p}}}{\voice{\confid{q}}}  \\\\ 
  \TVal{Π,\langle \aftypep{p}{q}\rangle; Γ;\pc}{e'}{τ}}
   {\TValGpc{\assert{e}{e'}}{τ}}

\Rule{Where}{
  \TValGpc{v}{\aftype{p}{q}} \\
   \rafjudge{Π}{\pcmost}{\voice{q}} \\
   \rafjudge{Π}{\voice{\confid{p}}}{\voice{\confid{q}}}  \\\\ 
  \TVal{Π,\langle \aftypep{p}{q}\rangle;Γ;\pc}{e}{τ} \\
}
{\TValGpc{(\where{e}{v})}{τ}}{}
  \hfill 
\end{mathpar} 
}
\caption{\lang type system.}
  \label{fig:ts}
\end{figure}

The core FLAC type system is presented in Figure~\ref{fig:ts}.  FLAC typing
judgments have the form $\TValGpc{e}{τ}$. The _delegation context_, $\Pi$,
contains a set of dynamic trust relationships $⟨p≽q⟩$ where $p ≽ q$
(read as ``$p$ acts for $q$'') is a delegation from $q$ to $p$.
The _typing context_, $Γ$, is a
associates variables to types, and $\pc$ is the _program counter label_, a FLAM
principal representing the confidentiality and integrity of control flow.  The
type system makes frequent use of judgments of the form $\rafjudge{\Pi}{p}{q}$ and
$\rflowjudge{\Pi}{p}{q}$ for comparisons between principals.  The derivation rules for
these judgments are presented in Figure~\ref{fig:robrules}, and are adapted from FLAM's
inference rules~\cite{flam}.\footnote{FLAM's rules~\cite{flam} also include {\em query} and {\em result} labels as 
part of the judgment context that represent the
confidentiality and integrity of a FLAM query context and result, respectively. The FLAM
delegation context also includes labels for delegations, whereas FLAC's delegation context does not.  
These labels are unnecessary in FLAC because we use FLAM judgments only in the type
system—these ``queries'' only occur at compile time and do not create
information flows about which delegations are in effect.  We formalize the connection 
between FLAC and FLAM in Appendix~\ref{app:proofs}.}%
These derivation rules are presented in terms of the acts-for ordering, but 
recall that since $⊑$ is defined in terms of $≽$, it makes sense to use either of these symbols 
to represent a derivation.
The rules are mostly straightforward except for \ruleref{R-Assume}, which allows delegations from
the context $Π$ to be used in derivations.  The premise \rafjudge{\delegcontext}{\voice{\confid{p}}}{\voice{\confid{q}}} 
enforces a well-formedness invariant on $Π$ that ensures that delegations of confidentiality are consistent
with the ability to speak for those principals.  This premise is related to FLAM's \textsc{Lift} rule. 
The \ruleref{Assume} typing rule ensures that all delegations added to $Π$ satisfy this invariant, and initial delegations
in $Π$ that fail to satisfy it cannot be used in derivations.

We also use judgments of the form $\protjudge*{p}{τ}$ to denote that type $τ$ is at least as
restrictive as the principal $p$. The derivation rules for these judgments are presented in
Figure~\ref{fig:protect} and discussed in more detail below.

\begin{figure}
{\footnotesize
\begin{mathpar}
	\Rule{R-Static}
	{\stafjudge{\L}{p}{q}}
	{\rafjudge{\delegcontext}{p}{q}} 

	\Rule{R-Assume}
	{
		\delexp{p}{q} \in \delegcontext \\
		\rafjudge{\delegcontext}{\voice{\confid{p}}}{\voice{\confid{q}}} 
	}
	{\rafjudge{\delegcontext}{p}{q}} 

	\Rule{R-ConjL}
	{ \rafjudge{\delegcontext}{p_{k}}{q} \\ k ∈ \{1,2\}}
	{	\rafjudge{\delegcontext}{p₁ ∧ p₂}{q_2}}

	\Rule{R-ConjR}
	{ \rafjudge{\delegcontext}{p}{q_1} \\
		\rafjudge{\delegcontext}{p}{q_2} 
	} 
	{\rafjudge{\delegcontext}{p}{q_1 \wedge q_2}} 

	\Rule{R-DisjL}
	{
		\rafjudge{\delegcontext}{p_1}{q} \\
		\rafjudge{\delegcontext}{p_2}{q} 
	} 
	{\rafjudge{\delegcontext}{p_1 \vee p_2}{q}} 

	\Rule{R-DisjR}
	{ \rafjudge{\delegcontext}{p}{q_{k}} \\ k ∈ \{1,2\}}
	{	\rafjudge{\delegcontext}{p}{q_1 ∨ q_2}}

	\Rule{R-Trans}
	{
		\rafjudge{\delegcontext}{p}{q} \\
		\rafjudge{\delegcontext}{q}{r} %
	} 
	{\rafjudge{\delegcontext}{p}{r}}

\end{mathpar}
}
  \caption{Inference rules for robust assumption, derived from FLAM~\cite{flam}.}
  \label{fig:robrules}
\end{figure}

The \ruleref{Del} rule types delegation expressions as singletons: for each delegation type there
is a unique delegation expression.
Rules \ruleref{Lam} and \ruleref{TLam} require the body of the abstraction to be well typed at
the annotated $\pc$. Rules \ruleref{App} and \ruleref{TApp} ensure functions may only be applied
in contexts which are no more restrictive than the annotation.  Rule \ruleref{TApp} additionally requires
that $τ'$ be well formed with respect to $Γ$. Specifically, $τ'$ may not have any free type variables
not bound by $Γ$.

Since FLAC is a pure functional language, it might seem odd for FLAC to have a
label for the program counter; such labels are usually used to control implicit
flows through assignments~(e.g., in \cite{ps03,myers-popl99}).  The purpose of
FLAC's $\pc$ label is to control a different kind of side effect: changes to the
delegation context, $Π$.\footnote{DFLATE~\cite{dflate}, an extension of FLAC for modeling distributed
applications with Trusted Execution Environments, 
uses the same {\pc} label to control implicit flows due to communication side-effects. Extensions
of FLAC to support mutable references or other effects could control implicit flows similarly.}

In order to control what information can influence the creation of a new
trust relationship in a delegation context, the type system tracks
the confidentiality and security of control flow.
Viewed as an authorization logic, FLAC's type system has the unique
feature that it expresses deduction constrained by an information flow
context.\footnote{FLAFOL~\cite{flafol} further develops the idea of constraining
  logical deduction with information flow constraints in a first-order logic.}
For instance, if we have $\func{φ}{p^{←}}{ψ}$ and $φ$, then (via \ruleref{App})
we may derive $ψ$ in a context with integrity $p^{←}$, but not in contexts that
don't flow to $p^{←}$. This feature offers needed control over how principals 
may apply existing facts to derive new facts.

Many FLAC terms are standard, such as pairs $\pair{e₁}{e₂}$, projections
$\proji{e}$, variants $\inji{e}$, 
and case expressions.  Function abstraction, $\lamc{x}{τ}{\pc}{e}$ and polymorphic type abstraction, $\tlam{X}{\pc}{e}$, include a $\pc$ _label_ that constrains the information flow context in which
the function may be applied. 
The rule \ruleref{App} ensures that
function application respects these policies, requiring that the robust FLAM
judgment $\rflowjudge{\Pi}{\pc}{\pc'}$ holds.  This judgment ensures
that the current program counter label, $\pc$, flows to the function label, $\pc'$. 

Branching occurs in "case" expressions, which conditionally evaluate one of two
expressions. The rule \ruleref{Case} ensures that both expressions have the same
type and thus the same protection level.  The premise $\protjudge*{\pc}{τ}$
ensures that this type protects the current $\pc$ label.
 
Like DCC, FLAC uses monadic operators to track
dependencies.  The monadic unit term $\return{ℓ}{w}$ (\ruleref{UnitM}) says that a value $w$ of type $τ$ is
_protected at level_ $ℓ$. This protected value has the type $\says{ℓ}{τ}$, meaning that it has the confidentiality and integrity
of principal $ℓ$. Because $w$ could implicitly reveal information about the dependencies 
of the computation that produced it, \ruleref{UnitM} requires that \rflowjudge{Π}{\pc}{\ell}.\footnote{Neither DCC~\cite{ccd99} nor
the original FLAC formalization~\cite{flac} included this premise.  DCC does not maintain a \pc label at all. FLAC originally
used a version of the DCC rule, but Cecchetti et al.~\cite{nmifc} and Gollamudi et al.~\cite{dflate} added the \pc restriction in
support of a non-commutative "says". See Section~\ref{sec:says} for additional details.}
When a monadic term $\return{ℓ}{w}$ steps to $\returng{ℓ}{w}$ we call it _sealed_ since all free values
have been substituted and the expression will not capture any additional information from its context.
Sealed terms type under the rule \ruleref{Sealed} which is more permissive since the $\pc$ premise is
unnecessary.

Computation on protected values must occur in a
protected context (``in the monad''), expressed using a monadic bind term.  
The typing rule $\ruleref{BindM}$ ensures that the result of the computation
protects the confidentiality and integrity of protected values. For instance,
the expression $\bind{x}{\return{ℓ}{v}}{\return{ℓ'}{x}}$ is only well-typed if
$ℓ'$ protects values with confidentiality and integrity $ℓ$. 
Since "case" expressions may use the variable $x$ for branching,
$\ruleref{BindM}$ raises the $\pc$ label to $\pc ⊔ ℓ$ to conservatively reflect the
control-flow dependency.

Protection levels are defined by the set of inference rules in
Figure~\ref{fig:protect}, adapted from~\cite{tz04}.  Expressions with unit type
(\ruleref{P-Unit}) do not propagate any information, so they protect
information at any ℓ.  Product types protect information at ℓ if both components
do (\ruleref{P-Pair}).  Function types protect information at ℓ if the return
type and function label does (\ruleref{P-Fun}), and polymorphic types protect information at whatever level
the abstracted type and type function label does (\ruleref{P-TFun}).
Finally, if $ℓ$ flows to $ℓ'$,
then $\says{ℓ'}{τ}$ protects information at ℓ (\ruleref{P-Lbl}).\footnote{DCC~\cite{ccd99} and the
original FLAC formalization included an additional protection rule that considered $ℓ$ to be protected by
$\says{ℓ'}{τ}$ if $τ$ protects $ℓ$ (even if $ℓ'$ does not).  This rule was removed by Cecchetti et al. and 
Gollamudi et al.~\cite{dflate} to make "says" non-commutative. Some variants of DCC~\cite{ccd08}
treat the "says" modality similarly.}
There are no protection rules for sum types or type variables since they do not protect information:
inspecting the constructor of a sum type value reveals information, and type variables may be instantiated
with types that offer different levels of protection or none at all. Because delegation expressions are singletons,
a protection rules for $\aftype{p}{q}$ types similar to the protection rule for $\voidtype$ would in principle
be admissible, but our examples and results did not require this permissiveness, and we have not explored
its consequences.

Occasionally it is more convenient to write protection relations in terms of only confidentiality
or integrity, so we also define a notation for authority projections on types in Figure~\ref{fig:typeproj}.
\begin{figure}
\begin{align*}
(\func{τ}{\pc}{τ})^{π}     &= (\func{τ}{\pc^{π}}{τ^{π}}) \\
(\says{ℓ}{τ})^{π}             &= \says{ℓ^{π}}{τ} \\
(\prodtype{τ}{τ})^{π}      &= (\prodtype{τ^{π}}{τ^{π}})  \\
(\tfuncpc{X}{\pc}{τ})^{π} &= (\tfuncpc{X}{\pc^{π} }{τ^{π}}) \\
\text{ otherwise } τ^{π} &= τ
\end{align*}
\caption{Authority projections on types}
\label{fig:typeproj}
\end{figure}

\begin{proposition}
  \[
  \protjudge*{ℓ^{π}}{τ^{π}} ⇔ \left\{\begin{array}{ll}
                                                          \protjudge*{ℓ^{→} ∧ ⊤^{←}}{τ} & \text{ for } π=→ \\
                                                          \protjudge*{ℓ^{←}}{τ} & \text{ for } π=← 
                                                        \end{array}\right.
   \]                                                  
\end{proposition}
\begin{proof} In the forward direction, by induction on the structure
of $τ$.  In the reverse direction, by induction on the derivation of
$\protjudge*{ℓ^{→} ∧ ⊤^{←}}{τ}$ (for $π=→$) and
$\protjudge*{ℓ^{←}}{τ}$ (for $π=←$).
\end{proof}

Most of the novelty of FLAC
 lies in its delegation values and "assume" terms.  These terms enable
expressive reasoning about authority and information flow control.  A delegation
value serves as evidence of trust. For instance, the term $\delexp{p}{q}$, read ``$p$ acts for $q$'',
 is evidence that $q$ trusts $p$.  Delegation values
have _acts-for types_; $\delexp{p}{q}$ has type $\aftype{p}{q}$.
\footnote{This correspondence 
with delegation values makes acts-for types a kind of singleton type~\cite{singletons}.}
The "assume" term enables programs to use evidence securely to create new flows between
protection levels.  In the typing context
$∅;x\ty\says{p^{←}}{τ};q^{←}$ (i.e., $Π=∅$, $Γ=x\ty\says{p^{←}}{τ}$, and $\pc=q^{←}$), the following
expression is not well typed:
\[
\bind{x'}{x}{\returnp{q^{←}}{x'}}
\]
since $p^{←}$ does not flow to $q^{←}$, as required by the premise $\protjudge*{ℓ}{τ}$ in rule \ruleref{BindM}.
Specifically, we cannot derive $\protjudge*{p^{←}}{\says{q^{←}}{τ}}$ since \ruleref{P-Lbl} requires the 
FLAM judgment $\rflowjudge{Π}{p^{←}}{q^{←}}$ to hold. 
  
However, the following expression is well typed:
\[
\assert{\delexp{p^{←}}{q^{←}}}{\bind{x'}{x}{\returnp{q^{←}}{x'}}}
\]
The difference is that the "assume" term adds a trust relationship, represented
by an expression with an acts-for type, to the delegation context.  In this case,
the expression $\delexp{p^{←}}{q^{←}}$ adds a trust relationship
that allows $p^{←}$ to flow to $q^{←}$.  
This is secure since $\pc=q^{←}$, meaning that only principals with integrity
$q^{←}$ have influenced the computation. With $⟨p^{←} ≽ q^{←}⟩$
in the delegation
context, added via the \ruleref{Assume} rule, the premises of \ruleref{BindM} are now satisfied, so the expression type-checks. 

Creating a delegation value requires no special privilege because the type
system ensures only high-integrity delegations are used as
evidence for enabling new flows.
Using low-integrity evidence for authorization would be
insecure since attackers could use delegation values to create new flows that reveal secrets or 
corrupt data.  The premises of the \ruleref{Assume} rule ensure the integrity of
dynamic authorization computations that produce values like
$\delexp{p^{←}}{q^{←}}$ in the example above.\footnote{These premises are
related to the robust FLAM rule \textsc{Lift}.}  The second premise,
$\rafjudge{Π}{\pc}{\voice{q}}$, requires that the $\pc$ has enough integrity
to be trusted by $q$, the principal whose security is affected. For instance, to
make the assumption $\aftypep{p}{q}$, the evidence represented by the term $e$
must have at least the integrity of the voice of $q$, written $∇(q)$.  Since the $\pc$
bounds the restrictiveness of the dependencies of $e$, 
this ensures that only information with integrity $∇(q)$ or higher may
influence the evaluation of $e$.  The third premise,
$\rafjudge{Π}{\voice{\confid{p}}}{\voice{\confid{q}}}$, ensures that
principal $p$ has sufficient integrity to be trusted to enforce $q$'s
confidentiality, $q^{→}$. This premise means that $q$ permits data to be
relabeled from $q^{→}$ to $p^{→}$.\footnote{More precisely, it means that the
voice of $q$'s confidentiality, $\voice{\confid{q}}$, permits data to be
relabeled from $q^{→}$ to $p^{→}$. Recall that
$\voice{\confid{\Alice}}$ is just $\Alice$'s integrity projection: $\Alice^{←}$.}

The $\pc$ constraints on function application ensure that functions containing "assume" terms can
only be applied in high-integrity contexts. For example, the following function declassifies one of
Alice's secrets to Bob. The first assume establishes that Bob is trusted to speak for Alice, and the second delegates Alice's confidentiality authority to Bob.  The "bind" term then relabels Alice's secret to a label observable by Bob.
\begin{align*}
  &"declassify" ::\;  \func{(\says{\Alice^{→}}{τ})}{\pc}{(\says{\Bob^{→}}{τ})} & \\ 
  &"declassify" = \lamc{x}{\says{\Alice^{→}}{τ}}{\pc}{} & \\
  & \qquad  \qquad\qquad    \qquad \assume{\delexp{\Bob^{←}}{\Alice^{←}}}{} & \\
  & \qquad  \qquad\qquad    \qquad \quad  \assume{\delexp{\Bob^{→}}{\Alice^{→}}}{} & \\
  & \qquad  \qquad\qquad    \qquad \qquad \bind{x'}{x}{\returnp{\Bob^{→}}{x'}} &
\end{align*}
If Bob could apply this function arbitrarily, then he could declassify
all of Alice's secrets—not just the ones she intended to release.
However, since the "assume" terms delegate Alice's confidentiality and integrity authority,
this function is only well typed if $\pc$ speaks for Alice, or $\rafjudge{Π}{\pc}{\Alice^{←}}$.
Otherwise the "assume" terms are rejected
by the type system.  The constraints in the \ruleref{App} rule
then ensure that this function can only be applied in a context that
flows to $\pc$.

Assumption terms evaluate to "where" expressions (rule \ruleref{E-Assume}).
These expressions are a purely formal bookkeeping mechanism (i.e., they would be
unnecessary in a FLAC-based implementation) to ensure that source-level terms
that were well-typed because of an "assume" term remain well-typed during evaluation. 
This helps us distinguish insecure FLAC terms from terms whose policies have been
legitimately downgraded.
These "where" terms record
and maintain the authorization evidence used to justify new flows of information
during evaluation. 
They are not part of the source language and 
generated only by the evaluation rules. 
The term $\where{e}{\delexp{p}{q}}$ records
that $e$ is typed in a context that includes the delegation $\delexp{p}{q}$.
 
The rule \ruleref{Where} gives a typing rule for "where" terms;
though similar to \ruleref{Assume}, it requires only that $∇(q)$ delegate to 
the distinguished label $\pcmost$, which is a fixed parameter of the type system.
The use of $\pcmost$ is purely technical:
our proofs in Section~\ref{sec:props} use $\pcmost$ to 
help reason about what new flows may have created by "assume" terms.
The only requirement is that $\pcmost$ be 
as trusted as the $\pc$ label used to 
type-check the source program (or programs) that generated the "where" term.  
Since the $\pc$ increases 
monotonically when typing subexpressions,
In our formal results, we choose $\pcmost$ to be $⊤^{←}$ since it is always valid.
Selecting a more restrictive label could offer
finer-grained reasoning about what downgrades may occur in non-source-level terms
since it restricts which "where"-terms are well-typed.

\begin{figure*}
{\footnotesize
\begin{flalign*}
& \boxed{
 e \stepsone e'
} &
\end{flalign*}
\begin{mathpar}
\erule*{W-App}{}{(\where{w}{v})~e}{\where{(w~e)}{v}}{}

\erule*{W-TApp}{}{(\where{w}{v})~τ}{\where{(w~τ)}{v}}{}

\erule*{W-UnPair}{}{\proj{i}{(\where{w}{v})}}{\where{(\proj{i}{w})}{v}}{}

\erule*{W-Case}{}{(\casexp{(\where{w}{v})}{x}{e_1}{e_2})}{\where{(\casexp{w}{x}{e_1}{e_2})}{v}}{}

\erule*{W-BindM}{}{\bind{x}{(\where{w}{v})}{e}}{\where{(\bind{x}{w}{e})}{v}}{}

\erule*{W-Assume}{}{\assume{(\where{w}{v})}{e}}{\where{(\assume{w}{e})}{v}}{}

\end{mathpar}
}
\caption{Propagation of "where" terms}
\label{fig:whererules}
\end{figure*}

Figure~\ref{fig:whererules} presents evaluation rules for where terms. 
The rules are designed to treat "where" values like the value they enclose.  
For instance, applying a "where" term (rule \ruleref{W-App}) simply moves
the value it is applied to inside the where term. If the "where" term was
wrapping a lambda expression, then it may now be applied via \ruleref{App}.
Otherwise, further reduction steps via \ruleref{W-App} may be necessary.
We use the syntactic category $w$ (see Figure~\ref{fig:syntax}) to refer to fully-evaluated
"where" terms, or "where" _values_.  In other words, 
a "where" value $w$ is an expression consisting of a value $v$ enclosed by one or more 
"where" clauses.  A "where" value usually behaves like a value, 
but it is occasionally convenient to distinguish them.\footnote{The original FLAC formalization
did not distinguish "where" values and values, and did not include the rules in Figure~\ref{fig:whererules}.
Unfortunately, this resulted in stuck terms when "where" terms were not propagated appropriately.
We have proven the above rules eliminate stuck terms for well-typed programs (Lemma~\ref{lemma:prog}).
}

\begin{figure}
{\footnotesize
\begin{flalign*}
& \boxed{\protjudge*{ℓ}{τ}} &
\end{flalign*}
\begin{mathpar}
	\Rule{P-Unit}{}
	{\protjudge{\delegcontext}{\ell}{\voidtype}}

	\Rule{P-Pair}
	{\protjudge{\delegcontext}{\ell}{\tau_1} \\
		\protjudge{\delegcontext}{\ell}{\tau_2} 
	}
	{\protjudge{\delegcontext}{\ell}{\prodtype{\tau_1}{\tau_2}}} 

     \protrule{P-Fun}{
            \protjudge{\delegcontext}{ℓ}{\tau_2} \\
            \protjudge{\delegcontext}{ℓ}{\pc'}}
          {ℓ}{\func{τ_1}{\pc'}{\tau_2}}

     \protrule{P-TFun}{
       \protjudge{\delegcontext}{ℓ}{\tau} \\
       \protjudge{\delegcontext}{ℓ}{\pc'}}{ℓ}{\tfuncpc{X}{\pc'}{\tau}}

	\Rule{P-Lbl}
	{\rflowjudge{\delegcontext}{\ell}{\ell'} 
	}
	{\protjudge{\delegcontext}{\ell}{\says{\ell'}{\tau}}}

\end{mathpar}
}
\caption{Type protection levels}
\label{fig:protect}
\end{figure}

\section{FLAC Proof theory}
\label{sec:pftheory}
\subsection{Properties of \texttt{says}} 
\label{sec:says}
FLAC's type system constrains how
principals apply existing facts to derive new facts.  For instance, a
property of "says" in other authorization logics (e.g., Lampson et
al.~\cite{labw91} and Abadi~\cite{abadi06}) is that implications that hold for
top-level propositions also hold for propositions of any principal ℓ:
\[
⊢ \func{(\func{\tau_1}{}{\tau_2})}{}{(\func{\says{ℓ}{\tau_1}}{}{\says{ℓ}{\tau_2}})}
\] 
The $\pc$ annotations on FLAC function types refine this property.
Each implication (in other words, each function) in FLAC is annotated with an upper
bound on the information flow context it may be invoked within.  To
lift such an implication to operate on propositions protected at label ℓ,
the label ℓ must
flow to the $\pc$ of the implication. Thus, for all $ℓ$ and $τᵢ$,
\begin{align*}
&⊢ \func{(\func{\tau_1}{ℓ}{\tau_2})}{ℓ}{(\func{\says{ℓ}{\tau_1}}{ℓ}{\says{ℓ}{\tau_2}})} &
\end{align*}
This judgment is a FLAC typing judgment in _logical form_, where terms have been 
omitted.  We write such judgments with an empty typing context (as above) when the judgment
is valid for any Π, Γ, and \pc. A judgment in logical form is valid if a _proof term_
exists for the specified type, proving the type is inhabited.
The above type has proof term
\begin{align*}
  & \lamc{f}{(\func{\tau_1}{ℓ}{\tau_2})}{ℓ}{} &\\
  & \quad\lamc{x}{\says{ℓ}{\tau_1}}{ℓ}{} \bind{x'}{x}{\return{ℓ}{f~x'}} &
\end{align*}
In order to apply $f$, we must first "bind" $x$, so according to
rules \ruleref{BindM} and \ruleref{App}, the function $f$ must have a label at least as
restrictive as $ℓ$,  and \ruleref{UnitM} requires the label of the returned value must also
be as restrictive as $ℓ$. We can actually prove a slightly more general version of the above theorem:
\begin{align*}
& \func{(\func{\tau_1}{\pc ⊔ ℓ}{\tau_2})}{ℓ}{(\func{\says{ℓ}{\tau_1}}{\pc}{\says{\pc ⊔ ℓ}{\tau_2}})} &
\end{align*}
This version permits using the implications in more restrictive
contexts, but doesn't map as well to a DCC theorem since the principal
of the return type differs from the argument type.

These refinements of DCC's theorems are crucial for supporting
applications like commitment schemes and bearer credentials. Our FLAC
implementations, presented in detail in Sections~\ref{sec:commit}
and~\ref{sec:bearer}, rely in part on restricting the $\pc$ to a
specific principal's integrity.  Without such refinements, principal
$q$ could open principal $p$'s commitments using "open", or create
credentials with $p$'s authority: $\spkfor{p^{→}}{\pc}{p^{←}}$.  With
these refinements, we can express privileged implications (functions)
that only trusted principals may apply.

Consider a DCC version of the
"declassify" function type from Section~\ref{sec:flac}:
\begin{align*}
  &\texttt{dcc\_declassify} ::\; \func{(\says{\Alice}{τ})}{}{(\says{\Bob}{τ})} & 
 \end{align*}
In DCC, functions are not annotated with $\pc$ labels and may be applied in any context. 
Therefore, _any_ principal could use \texttt{dcc\_declassify} to relabel Alice's information
to Bob—including Bob.

Other properties of "says" common to DCC and other logics 
(cf.~\cite{abadi03} for examples) are similarly refined by $\pc$
bounds. Two examples are: 
$ ⊢ \func{τ}{ℓ}{\says{ℓ}{τ}}$ which has proof term:
$\lamc{x}{τ}{ℓ}{\return{ℓ}{τ}}$ 
and 
\begin{align*}
⊢ \func{\says{ℓ}{(\func{\tau_1}{ℓ}{\tau_2})}}{ℓ}{(\func{\says{ℓ}{\tau_1}}{ℓ}{\says{ℓ}{\tau_2}})}
\end{align*}
with proof term:
\begin{align*}
  & \lamc{f}{\says{ℓ}{(\func{\tau_1}{ℓ}{\tau_2})}}{ℓ}{\bind{f'}{f}{}} &\\
  & \qquad\lamc{y}{\says{ℓ}{\tau_1}}{ℓ}{} \bind{y'}{y}{\return{ℓ}{f'~y'}} &
\end{align*}

Some theorems of DCC cannot be obtained in FLAC, due to the \pc restriction 
on \ruleref{UnitM} as well as the more restrictive protection relation.
For example, chains of "says" are not commutative in FLAC.  Given $ℓ₁$, $ℓ₂$, and $\pc$, 
\begin{align*}
& ⊬\func{\says{ℓ₁}{\says{ℓ₂}{τ}}}{\pc}{\says{ℓ₂}{\says{ℓ₁}{τ}}} &
\end{align*}
unless $\rflowjudge{Π}{ℓ₁⊔\pc}{ℓ₂}$ and $\rflowjudge{Π}{ℓ₂⊔\pc}{ℓ₁}$, which implies $ℓ₁$, $ℓ₂$, and $\pc$
must be equivalent in $Π$.
CCD~\cite{ccd08}, a logic related to DCC, is also non-commutative with respect to "says", but does not 
have an associated term language.

Distinguishing the nesting order of "says" types is attractive for
authorization settings since it encodes the provenance of statements.
It also enables modeling of cryptographic mechanisms in FLAC
(\emph{cf.}~\cite{dflate}) where \says{ℓ}{τ} is interpreted as a value
of type $τ$ protected by encryption key $ℓ^{→}$ and signing key
$ℓ^{←}$.  Preserving the order of nested types in this context is
useful for modeling decryption and verification of protected values.

\subsection{Dynamic Hand-off} 
\label{sec:handoff} 
Many authorization logics support delegation using a ``hand-off'' axiom.  In DCC, 
this axiom is a provable theorem:
\[
⊢ (q~"says"~(p ⇒q)) → (p ⇒ q)
\] 
where $p⇒q$ is shorthand for $$∀{X}.~(\func{\says{p}{X}}{}{\says{q}{X}})$$
However, $p⇒q$ is only inhabited if $p⊑q$ is derivable in the security lattice.  Thus, DCC
can reason about the consequences of an assumption that $p⊑q$ holds (whether it is true for the lattice
or not), but a DCC program cannot produce a term of type $p⇒q$ unless $p⊑q$.

FLAC programs, on the other hand, can create new trust relationships from
delegation expressions using "assume" terms. The type analogous to $p ⇒ q$ in
FLAC is $$\tfuncpc{X}{\pc}{(\func{\says{p}{X}}{\pc}{\says{q}{X}})}$$ which we write 
as $\spkfor{p}{\pc}{q}$.
FLAC programs construct terms of this type from proofs of authority, represented by
terms with acts-for types.  This feature enables a more general form of
hand-off, which we state formally below.
\begin{proposition}[Dynamic hand-off]
  \label{prop:handoff}
{
For all $ℓ$ and $\pc'$,
let  $\pc=ℓ^{→}∧∇(p^{→})∧q^{←}∧∇(pc')$ 
\begin{align*}
 & \func{\aftype{\voice{\confid{q}}}{\voice{\confid{p}}}}{\pc}{} \func{\flowtype{p}{q}}{\pc}{ \func{\flowtype{pc'}{q}}{\pc}{}} &\\ 
 & \qquad \tfuncpc{X}{\pc'}{(\func{\says{p}{X}}{\pc'}{\says{q}{X}})} & 
\end{align*}
}
\end{proposition} 
\noindent
{
\textit{Proof term.}
\begin{align*}
& \lamc{\pf_1}{\aftype{∇(q^{→})}{∇(p^{→})}}{\pc}{\assume{\pf_1}{}} &\\
& \quad \lamc{\pf_2}{\flowtype{p}{q}}{\pc} {\assume{\pf_2}{}} &\\
& \qquad \lamc{\pf_3}{\flowtype{pc'}{q}}{\pc} {\assume{\pf_3}{}} &\\
& \quad \qquad\tlam{X}{\pc'}{}\lamc{x}{\says{p}{X}}{\pc'} {\bind{x'}{x}{\return{q}{x'}}} &
\end{align*}
}
The principal $\pc=ℓ^{→}∧∇(p^{→})∧q^{←}∧∇(pc')$ restricts delegation
(hand-off) to contexts with sufficient integrity to authorize the
delegations made by the "assume" terms.  In other words, the context
that creates these delegations must be authorized by the combined
authority of $∇(p^{→})$, $q^{←}$, and $∇(pc')$.

The three arguments are proofs of authority with acts-for types: a proof of
$\aftypep{∇(q^{→})}{∇(p^{→})}$, a proof of $\flowtypep{p}{q}$, and a proof of $\flowtypep{pc'}{q}$. 
The $\pc$ ensures that the proofs have sufficient
integrity to be used in "assume" terms since it has the integrity of both 
 $∇(p^{→})$ and $q^{←}$.  Note that low-integrity or confidential delegation values
must first be bound via "bind" before the above term may be applied.  
Thus the $\pc$ would reflect the protection level of both arguments. 
Principals $ℓ^{→}$ and $\pc'$ are unconstrained, but the third proof argument ensures that flows from $pc'$ 
to $q$ are authorized, since a principal with access to $q$'s secrets could infer something 
about the context (protected at $pc'$) in which the hand-off function is applied. Other dynamic hand-off formulations 
are possible, Proposition~\ref{prop:handoff} simply has the fewest assumptions. Other formulations can eliminate 
the need for proofs $pf₁$, $pf₂$, and/or $pf₃$ if these relationships already exist in the context defining the hand-off.

Dynamic hand-off terms give FLAC programs a level of expressiveness and security
not offered by other authorization logics.  Observe that $\pc'$ may be chosen
independently of the other principals.  This means that although the $\pc$
prevents low-integrity principals from creating hand-off terms, a high-integrity
principal may create a hand-off term and provide it to an
arbitrary principal.  Hand-off terms in FLAC, then, are similar to capabilities
since even untrusted principals may use them to change the protection level of
values.  Unlike in most capability systems, however, the propagation of hand-off
terms can be constrained using information flow policies.

Terms that have types of the form in Proposition~\ref{prop:handoff} illustrate a
subtlety of enforcing information flow in an authorization mechanism.  Because
these terms relabel information from one protection level to another protection
level, the transformed information implicitly depends on the proofs of
authorization. FLAC ensures that the information security of
these proofs is protected—like that of all other values—even as the policies of
other information are being modified. Hence, authorization proofs cannot
be used as a side channel to leak information.

For example, if \Alice's trust relationship with \Bob is secret, she
might protect it at confidentiality $\Alice^{→}$.  If \Alice wants to
delegate trust to \Bob using an approach like
Proposition~\ref{prop:handoff}, the delegation protected at
$\Alice^{→}$ would first have to be bound:
\begin{align*}
& \bind{d}{\returnp{\Alice}{\delexp{\Bob}{\Alice}}}{...} &
\end{align*}
This would imply (by the \ruleref{BindM} typing rule) that in order to create a term with type
$\spkfor{\Bob}{\pc}{\Alice}$ in the body of the "bind", it must be the case that
$\protjudge{\Pi}{\Alice}{(\spkfor{\Alice}{\pc}{\Bob})}$, which in turn requires
that $\rflowjudge{\Pi}{\Alice}{\pc}$ and $\rflowjudge{\Pi}{\Alice}{\Bob}$.

Of course, if $\rflowjudge{\Pi}{\Alice}{\Bob}$ already holds, then there is no
need to delegate confidentiality authority to Bob.
Therefore, a typical approach would first declassify the delegation
to Bob (e.g., relabel it from $\says{\Alice}{\aftype{\Bob}{\Alice}}$
to $\says{\Bob^{→}∧\Alice{←}}{\aftype{\Bob}{\Alice}}$), before handing
off authority.  This requirement ensures that the disclosure of the
secret trust relationship is intentional, and is an example of a more
general principle in FLAC: that restricted terms cannot be used to
``downgrade themselves.''  We formalize this idea in
Section~\ref{sec:dcdi} with the Delegation Compartmentalization lemma.

\section{Semantic security properties of FLAC}
\label{sec:props} 

In this section, we formalize our semantic security guarantees.  Our
results are based on a bracketed semantics in the style of Pottier and
Simonet~\cite{ps03} extended to a trace-based semantic model.  The
observability of trace elements is defined by an erasure function and
our noninterference and robust declassification theorems are stated in
terms of indistinguishability on traces.
Because FLAC supports downgrading via "assume", the usual bracketed
semantic approach to noninterference via type preservation is
insufficient, so we develop an approach to characterize what
downgrades are possible in well-typed FLAC programs.
Finally, with the necessary infrastructure in place, we state and prove our
noninterference and robust declassification theorems.

\subsection{Trace indistinguishability}

We express our semantic security results in terms of the _traces_ of a program observable to an attacker. 
FLAC traces are simply the sequence of terms under the $\stepsone$ relation.  That is, each evaluation
step of the form $e \stepsone e'$ generates a new trace element $e'$.  We write the trace generated by 
taking $n$ steps from $e$ to $e'$ as $e \stepstot{t} e'$, where $t[0]=e$ and $t[n]=e'$.

FLAC traces are not fully observable to an attacker.  Trace elements
generated by protected information such as sealed values or protected
contexts such as within a "bind" or lambda term are hidden from the
attacker. This approach models scenarios where an attacker has a
limited ability to observe program values, for example when sealed
values are protected by a trustworthy a runtime or are signed and
encrypted.  For these scenarios, the \ruleref{UnitM} typing rule and an observability function
model how computation on sealed values,
and thus the observability of intermediate values, is limited to principals with
necessary permissions or cryptographic keys to access inputs and
produce outputs.\footnote{DFLATE~\cite{dflate} explores this connection 
with cryptographic enforcement more explicitly in a distributed extension of FLAC.}
An insecure program in this model allows an attacker to learn information by 
allowing protected information to flow to sealed values or contexts that are
observable by the attacker.

Before formally defining what portions of program trace are available to an
attack, we must first add some additional bookkeeping to our semantic
rules.  Specifically, we need additional notation for evaluating some intermediate expressions.
The notation $\octx{ℓ}{e}$ denotes an intermediate expression
$e$ evaluated in a context protected at $ℓ$.

\begin{figure}
\begin{mathpar}
\text{\underline{Syntax}} \hfill \\
e ::= \dots \sep \octx{ℓ}{e} \hfill
\\
\text{\underline{Evaluation contexts}}\hfill  \\
E ::= ...  \sep \octx{ℓ}{E} \hfill
\\
\text{\underline{Evaluation rules}} \hfill \\
\erule{E-App*}{}{(\lamc{x}{τ}{\pc}{e})~w }{\octx{pc}{\subst{e}{x}{w}}}{}

\erule{E-TApp*}{}{(\tlam{X}{\pc}{e})~τ}{\octx{pc}{\subst{e}{X}{τ}}}{}

\erule{E-BindM*}{}{\bind{x}{\returng{ℓ}{w}}{e}}{\octx{ℓ}{\subst{e}{x}{w}}}{}

\erule{O-Ctx}{}{\octx{ℓ}{w}}{w}{}

\hfill
\\
\text{\underline{Typing rules}} \hfill \\
\Rule{Ctx}{\TValP{Γ;\pc'}{e}{τ} \\
  \rflowjudge{Π}{\pc}{\pc'}
} {\TValGpc{\octx{\pc'}{e}}{τ}}
\hfill
\end{mathpar}

\caption{Extensions to support protection contexts.}
\label{fig:octx}
\end{figure}

Figure~\ref{fig:octx} presents syntax and evaluation rules for
introducing and eliminating protected contexts.  Rules
\ruleref{E-App*}, \ruleref{E-TApp*}, and \ruleref{E-BindM*} replace
their counterparts in Figure~\ref{fig:semantics} with rules that
introduce a protected context based on the relevant label.
Rule~\ruleref{O-Ctx} eliminates the protected context when the
intermediate expression is fully evaluated.  The typing rule
\ruleref{Ctx} ensures that expressions inside protected contexts are
well-typed at the annotated label, and that the $\pc$ flows to the
annotated label.

\begin{figure*}
  {\small
\begin{subfigure}{\textwidth}
\[
\begin{array}{l l  l}
                                    \observef{x}{\delegcontext}{p} & = &   x \\
                                    \observef{\delexp{a}{b}}{\delegcontext}{p} & = & \delexp{a}{b}  \\  %
                                    \observef{()}{\delegcontext}{p} & = &   () \\
                                     \observef{\return{\ell}{e}}{\delegcontext}{p} & = & \return{\ell}{\observef{e}{\delegcontext}{p}} \\
                                    \observef{\returng{\ell}{w}}{\delegcontext}{p} & = & 
                                               \begin{cases}
                                                    \returng{\ell}{\observef{w}{\delegcontext}{p}} & \text{if } \rflowjudge{\delegcontext}{\ell^{π}}{p^{π}} \\
                                                    \circ &   \text{otherwise}
                                               \end{cases}\\[1.25em]
                                    \observef{\lamc{x}{\tau}{\pc}{e}}{\delegcontext}{p} & = &
                                               \begin{cases}
                                                   \lamc{x}{\tau}{\pc}{\observef{e}{\delegcontext}{p}} & \text{if }  \rflowjudge{\delegcontext}{\pc^{π}}{p^{π}} \\
                                                   \circ &   \text{otherwise}
                                               \end{cases}\\[1.25em]
                                    \observef{\tlam{X}{\pc}{e}}{\delegcontext}{p} & = & 
                                               \begin{cases}
                                                   \tlam{X}{\pc}{\observef{e}{\delegcontext}{p}} & \text{if }  \rflowjudge{\delegcontext}{\pc^{π}}{p^{π}} \\
                                                   \circ  &  \text{otherwise}
                                               \end{cases}\\[1.25em]
                                    \observef{\octx{\ell}{e}}{\delegcontext}{p} & = & 
                                               \begin{cases}
                                                    \octx{\ell}{\observef{e}{\delegcontext}{p}} & \text{if }  \rflowjudge{\delegcontext}{\ell^{π}}{p^{π}} \\
                                                    \circ &   \text{otherwise}
                                               \end{cases}\\[1.25em]
                                    \observef{e₁~e₂}{\delegcontext}{p} & = & 
                                               \begin{cases}
                                                    \circ & \text{if } \observef{e_i}{\delegcontext}{p} = \circ \\
                                                    \observef{e_1}{\delegcontext}{p}~\observef{e_2}{\delegcontext}{p} & \text{otherwise}
                                                \end{cases} \\[1.25em]
                                    \observef{\pair{e_1}{e_2}}{\delegcontext}{p} & = & 
                                               \begin{cases}
                                                    \circ & \text{if } \observef{e_i}{\delegcontext}{p} = \circ \\
                                                    \pair{\observef{e_1}{\delegcontext}{p}}{\observef{e_2}{\delegcontext}{p}} & \text{otherwise}
                                                \end{cases} \\
                                    \observef{\proji{e}}{\delegcontext}{p} & = & \proji{\observef{e}{\delegcontext}{p}} \\
                                    \observef{\inji{e}}{\delegcontext}{p} & = &  
                                               \begin{cases}
                                                    \circ & \text{if } \observef{e}{\delegcontext}{p} = \circ \\
                                                   \inji{\observef{e}{\delegcontext}{p}} & \text{otherwise}
                                                \end{cases} \\
                                    \observef{\casexp{e}{x}{e_1}{e_2}}{\delegcontext}{p} & = & \casexpa{\observef{e}{\delegcontext}{p}} \\ 
                                                                                                                       && \quad \phantom{\mid}~\casexpb{x}{\observef{e_1}{\delegcontext}{p}} \\ 
                                                                                                                       && \quad \casexpc{x}{\observef{e_2}{\delegcontext}{p}} \\ 
                                    \observef{\bind{x}{e}{e'}}{\delegcontext}{p} & = & \bind{x}{\observef{e}{\delegcontext}{p}}{\observef{e'}{\delegcontext}{p}} \\
                                    \observef{\assume{e}{e'}}{\delegcontext}{p} & = &
                                    \begin{cases}
                                      \circ & \text{if } \observef{e}{\delegcontext}{p} =  \circ \text{ and } \\
                                      & \qquad \observef{e'}{\delegcontext}{p} =  \circ \\
                                      \assume{\observef{e}{\delegcontext}{p}}{\observef{e'}{\delegcontext}{p}} & \text{otherwise}
                                      \end{cases} \\
                                     \observef{\where{e}{v}}{\delegcontext}{p} & = & \observef{e}{\delegcontext}{p}\\
                                     
\end{array}
\]
\caption{Observation function for intermediate FLAC terms.}
\end{subfigure}
\begin{subfigure}{\textwidth}
\[
\begin{array}{l l  l}
                                    \observef{[e]}{Π}{p} & = & [\observef{e}{Π}{p}]\\
                                    \observef{[e;e']}{Π}{p} & = & \begin{cases}
                                                                 [\observef{e}{Π}{p}]  &\text{if }  \observef{e}{Π}{p} = \observef{e'}{Π}{p} \\
                                                                 [\observef{e}{Π}{p}; \observef{e'}{Π}{p} ]  & \text{otherwise}
                                                                 \end{cases} \\[1.1em] 
                                    \observef{[e;e']·t}{Π}{p} & = &\observef{\observef{[e;e']}{Π}{p}·t}{Π}{p}
\end{array}
  \]
\caption{Observation function for traces.}
 \end{subfigure}
}
\caption{Observation function definitions}
\label{fig:observe}
\end{figure*}

The portion of a FLAC program observable to an attacker is formally defined by an
observation function, $\observef{e}{Π}{p}$, defined in Figure~\ref{fig:observe}.
An expression $e$ is observable by principal $p$ in delegation context $Π$ depending on the
authority of $p$ relative to the protected terms in $e$. The projection $π$ specifies 
whether to consider the confidentiality of protected terms ($π=→$) or the integrity ($π=←$).
For example, sealed values such as $\returng{ℓ}{w}$ are observable by principal $p$ if
$p$ acts for $ℓ$, otherwise they are erased.  Protected contexts like $\octx{ℓ}{e}$ 
are treated similarly.  To simplify our proofs, we collapse terms whose subterms 
have been erased.  For example $⟨∘,∘⟩$ is collapsed to $∘$.

Erasing delegations in the context of "where" terms essentially hides the
delegations that justify a value's flow from the attacker.
This is consistent with the idea that attackers may learn _implicit_ information 
from flows that introduce new delegations. In other words, by observing differences in
outputs the attacker may infer the existence of secret delegations, but these delegations are
not _explicit_ in the output.  By contrast, these delegations are explicit in the output of
DFLATE~\cite{dflate} programs, modeling values that carry certified justifications of why their
flow was authorized by the program. 

The observability of trace elements is defined in Figure~\ref{fig:observe} in terms of 
the observability function $\observe$ for FLAC terms. We also lift the observability of
trace elements to traces in a natural way.  Note
that duplicate entries (which may occur due to evaluation in protected
contexts) are removed.  Deduplication avoids unintentional sensitivity to the number
of steps taken in unobservable contexts.
                                  
We now can define _trace indistinguishability_. Two traces $t$ and $t'$ are indistinguishable
to a principal $p^{π}$ in delegation context $Π$ if the observable elements of $t$ and $t'$ 
are equal.
\begin{definition}[Trace indistinguishability]
  \[
  \begin{array}{rcl}
  t ≈^{Π}_{p^{π}} t'  &\overset{Δ}{\iff} &  \observef{t}{Π}{p} = \observef{t'}{Π}{p} 
  \end{array}
 \]
 \end{definition}
\subsection{Bracketed semantics}

Our noninterference proof is based on the bracketed semantics approach
used by Pottier and Simonet~\cite{ps03}. This approach extends FLAC
with _bracketed expressions_ which represent two executions of a
program, and allows us to reason about noninterference in FLAC, a
_2-safety hyperproperty_~\cite{cs08}, as _type safety_ in the extended
language.  Any two FLAC terms $e₁$ and $e₂$ in the unbracketed
language can be combined into a term $\bracket{e₁}{e₂}$ in the
bracketed language.  For any term in the bracketed language, a
projection function $\outproj{·}{i}$, for $i ∈ \{1 , 2\}$, extracts
the term from each execution.  Specifically,
$\outproj{\bracket{e₁}{e₂}}{i} = eᵢ$ and projections are homomorphic
on other expressions; for example $\outproj{\lamc{x}{τ}{\pc}{e}}{i} =
\lamc{x}{τ}{\pc}{\outproj{e}{i}}$.  Since each element of the trace defined by evaluation 
of $e \stepstot{t} e'$ is an intermediate FLAC term, we define projections
on traces $\outproj{t}{i}$ as the sequence of projected terms
$\outproj{t[0]}{i}, ..., \outproj{t[n]}{i}$ where $\outproj{t[0]}{i} = \outproj{e}{i}$
and $\outproj{t[n]}{i} = \outproj{e'}{i}$. 

\begin{figure*}
{\small
\begin{mathpar}
\text{\underline{Syntax}} \hfill \\
\begin{array}{rcl}
  w &::=& \dots \sep \bracket{w}{w} \\[0.4em]
  e  &::=& \dots \sep \bracket{e}{e}
\end{array}\hfill \\

\text{\underline{Evaluation rules}} \hfill \\
    \berule{B-Step}{e_i \stepsone e'_i \\ e'_j =e_j \\ \{i, j \} = \{1, 2\} }{\bracket{e_1}{e_2}}{\bracket{e'_1}{e'_2}}{}

    \berule*{B-App}{}{\bracket{w₁}{w₂}~w}{\bracket{w₁~\outproj{w}{1}}{w₂~\outproj{w}{2}}}{}

    \berule*{B-TApp}{}{\bracket{w}{w'}~τ}{\bracket{w~τ}{w'~τ}}{}

    \berule*{B-UnPair}{}{\proj{i}{\bracket{\pair{w_{11}}{w_{12}}}{\pair{w_{21}}{w_{22}}}}}{\bracket{w_{1i}}{w_{2i}}}{}

    \berule*{B-BindM}{}{\bind{x}{\bracket{w}{w'}}{e}}{\bracket{\bind{x}{w}{\outproj{e}{1}}}{\bind{x}{w'}{\outproj{e}{2}}}}{}

    \berule*{B-Case}{ \{i, j\} = \{1, 2 \}}{\casexp{\bracket{w}{w'}}{x}{e_1}{e_2}\\}{\bracket{\casexp{w}{x}{\outproj{e_1}{1}}{\outproj{e_2}{1}}}{\casexp{w'}{x}{\outproj{e_1}{2}}{\outproj{e_2}{2}}}}{} 

    \berule*{B-Assume}{} 
           {\assume{\bracket{w}{w'}}{e}}{\bracket{\assume{w}{\outproj{e}{1}}}{\assume{w'}{\outproj{e}{2}}}}{}
     \hfill      
     \\\\
\text{\underline{Typing rules}} \hfill \\
    \Rule{Bracket}
	 {
	   \rflowjudge{\delegcontext}{(H^\pi \sqcup \pc^\pi)}{{\pc'^\pi}} \\
           \TValP{Γ;\pc'}{e_1}{\tau} \\
	   \TValP{Γ;\pc'}{e_2}{\tau} \\
	   \protjudge{\delegcontext}{H^{\pi}}{\tau^\pi}
           }
	 {\TValGpc{\bracket{e_1}{e_2}}{\tau}}

  \Rule{Bracket-Values}
         {
           \protjudge{\delegcontext}{H^\pi}{\tau^\pi} \\
           \TValGpc{w_1}{\tau} \\
           \TValGpc{w_2}{\tau}
         }
         {\TValGpc{\bracket{w_1}{w_2}}{\tau}}
   \\\\
\text{\underline{Observation function}} \hfill \\
  \begin{array}{l l  l}
    \observef{\bracket{e_1}{e_2}}{\delegcontext}{p} & = &
               \begin{cases}
                    \circ & \observef{e_i}{\delegcontext}{p} = \circ \\
                    \bracket{\observef{e_1}{\delegcontext}{p}}{\observef{e_2}{\delegcontext}{p}} & \text{otherwise}
                \end{cases}
  \end{array}
\hfill
\end{mathpar}
}
\caption{Extensions for bracketed semantics}
\label{fig:brackets}
\end{figure*}

Figure~\ref{fig:brackets} presents the bracket extensions for FLAC.
Where-values $w$ and expressions $e$ may be bracketed.
\ruleref{B-Step} evaluates expressions inside of brackets.
The remaining \text{B-*}
evaluation rules propagate brackets out of subexpressions.  
Note that projections are applied as the scope of brackets expand so that brackets can
never become nested.

The bracketed evaluation rules are designed to ensure
bracketed terms do not get stuck unless the unbracketed terms do.  We verify this with the
following lemma. Informally, if the bracketed term is stuck then either left or right execution is also stuck.
\begin{restatable}[Stuck expressions]{lemma}{stuck}
\label{lemma:stuck}
	If  $e$ gets stuck then $\outproj{e}{i}$ is stuck for some $i \in \{1,2\}$.
\end{restatable}
\begin{proof}
	By induction on the structure of $e$. See Appendix~\ref{app:proofs} for complete proof.
\end{proof}

The bracketed type system is parameterized by a fixed principal $H$
which specifies which policies are considered secret or untrusted,
and an authority projection $π$, depending on whether the type system
is verifying confidentiality ($→$) or integrity ($←$).
The typing rules \ruleref{Bracket} and \ruleref{Bracket-Values} illustrate the primary
purpose of the bracketed semantics: to link distinguishable
evaluations of an expression to the expression's type.
\ruleref{Bracket} requires that the bracketed expressions $e₁$ and
$e₂$ are typable at a $\pc'$ that protects $H^{π} ⊔ \pc^{π}$ and
$\pc'^{π}$ and have a type $τ^{π}$ protects
$H^{π}$. \ruleref{Bracket-Values} relaxes the restriction on $\pc$ for
where-values.
Some of the bracketed evaluation rules are necessary for completeness with respect to the unbracketed semantics,
but are not actually necessary for well typed programs.  Specifically \ruleref{B-Case} and \ruleref{B-Assume}
step on expressions that are never well typed since types of the form $\sumtype{τ}{τ'}$ and $\aftype{p}{q}$ 
cannot protect any instantiation of $H$.

Following Pottier and Simonet~\cite{ps03}, our first result on the bracketed semantics is that they are sound and complete with 
respect to the unbracketed semantics. 
By soundness, we mean that given a step in the bracketed execution, then at least one of the left or right projections take a step such that they are in relation with the bracketed execution. By completeness, we mean that given a left and right execution, we can construct a corresponding bracketed execution.
\begin{lemma}[Soundness]\label{lemma:sound}
	If  $e \stepsone e'$ then
        $\outproj{e}{k} \stepsto \outproj{e'}{k}$ for $k ∈ \{1,2\}$.
 \end{lemma}
\begin{proof}
  By induction on the evaluation of $e$.  Observe that all bracketed rules in Figure~\ref{fig:brackets} except \ruleref{B-Step} only expand brackets, so 
$\outproj{e}{k} = \outproj{e'}{k}$ for $k ∈ \{1,2\}$. For \ruleref{B-Step}, $\outproj{e}{i} \stepsone \outproj{e'}{i}$  and $\outproj{e}{j} = \outproj{e'}{j}$ .
\end{proof}

\begin{lemma}[Completeness]
  \label{lemma:complete}
  If $\bone{e} \stepsto w_1 $ and $\btwo{e} \stepsto w_2$, then there exists some $w$ such that $e \stepsto w$ 
and $\outproj{w}{i} = wᵢ$ for $i ∈ \{1,2\}$.
\end{lemma}
\begin{proof} 
  \OA{XXX: Should we address this comment?}
  \OA{Reviewer 2: 
In the proof of [Completeness], you implicitly perform a case analysis on
the fact that a program either terminates with a value, or gets
stuck, or diverges. Please make this reasoning explicit in your
proof, to ease the understanding of your proof sketch. Why don't you
perform an induction on the structure of e, as you did in the proof
of Lemma 2?
}    
  Assume $\bone{e} \stepsto w_1$ and $\btwo{e} \stepsto w_2$.  The extended set
of rules in Figure~\ref{fig:brackets} always move brackets out of subterms, and
therefore can only be applied a finite number of times.  Therefore, by
Lemma~\ref{lemma:sound}, if $e$ diverges, either $\bone{e}$ or $\btwo{e}$
diverge; this contradicts our assumption.

Furthermore, by Lemma~\ref{lemma:stuck}, if the evaluation of $e$ gets stuck, either
$\bone{e}$ or $\btwo{e}$ gets stuck. Therefore, since we assumed $\bgeti{e} \stepsto w_i$, then 
$e$ must terminate, so $e \stepsto w$. 
Finally, by induction on the number of evaluation steps and Lemma~\ref{lemma:sound}, $\outproj{w}{i} = wᵢ$ for $i ∈ \{1,2\}$. 
\end{proof}

Using the soundness and completeness results, we can relate properties
of programs executed under the bracketed semantics to executions under
the non-bracketed semantics.  In particular, since bracketed
expressions represent distinguishable executions of a program, and may
only have types that protect $H^{→}$ (or $H^{←}$), it is important to
establish that the type of a FLAC expression is preserved by
evaluation. Lemma~\ref{lemma:subjred} states this formally. This lemma
helps us reason about whether an expression is bracketed based 
on its type.

\begin{restatable}[Subject Reduction]{lemma}{subjred}
\label{lemma:subjred}
  Let  \TValGpc{e}{τ}.  If $e \stepsone e'$ then \TValGpc{e'}{τ}. 
\end{restatable}
\begin{proof}
By induction on the evaluation of $e$. See Appendix~\ref{app:proofs} 
for proof and supporting lemmas.
\end{proof}

\begin{restatable}[Progress]{lemma}{prog}
  \label{lemma:prog}
If  $\TVal{\Pi;\emptyset;\pc}{e}{τ}$, then either $e \stepsone e'$ or $e$ is a where value.
\end{restatable}
\begin{proof} 
  By induction on the derivation of $\TVal{\Pi;\emptyset;\pc}{e}{τ}$. See Appendix~\ref{app:proofs}
  for complete proof.
\end{proof}

\subsection{Delegation Compartmentalization and Invariance}
\label{sec:dcdi}

FLAC programs dynamically extend trust relationships, enabling new flows of
information.  Nevertheless, well-typed programs have end-to-end semantic
properties that enforce strong information security.  These properties derive
primarily from FLAC's control of the delegation context.  The
\ruleref{Assume} rule ensures that only high-integrity proofs of
authorization can extend the delegation context, and furthermore that
such extensions occur only in high-integrity, lexically-scoped contexts.

That low-integrity contexts cannot extend the delegation
context turns out to be a crucial property.  This property allows us to state a
useful invariant about the evaluation of FLAC programs.  Recall that "assume"
terms evaluate to "where" terms in the FLAC semantics.  Thus, FLAC
programs typically compute values containing a hierarchy of nested "where" terms.  The
terms record the values whose types were used to extend the delegation
context during type checking.

For a well-typed FLAC program, we can prove that certain trust relationships
could not have been added by the program. Characterizing these
relationships requires a concept of the minimal authority required to
cause one principal to act for another.  Although similar, this idea is distinct
from the voice of a principal.  Consider the relationship between $a$ and $a ∧
b$.  The voice of $a ∧ b$, $∇(a ∧ b)$, is sufficient integrity to add a
delegation $a ∧ b$ to $a$ so that $a ≽ a ∧ b$. Alternatively, having only the
integrity of $∇(b)$ is also sufficient to add a delegation $a ≽ b$, which also
results in $a ≽ a ∧ b$.

In our theorems, we need to be able to characterize what flows might become
enabled by a program.  For instance, if we want to reason about whether 
$a ≽ a ∧ b$ in some scope (and delegation context) of a program, 
we need to identify the _minimal necessary integrity_ $∇(b)$ 
that can close the gap in authority between 
the pair of principals $a$ and $a ∧ b$.  The following definitions are in
service of this goal.

The first definition formalizes the idea that two principals are considered
equivalent in a given context if they act for each other.
\begin{definition}[Principal Equivalence]
  \label{def:prineq}
  We say that two principals $p$ and $q$ are _equivalent_ in $Π$,
  denoted $\reqjudge{\Pi}{p}{q}$, if
  $$\rafjudge{\Pi}{p}{q} \text{ and } \rafjudge{\Pi}{q}{p}.$$
\end{definition}

Next, we define the _factorization_ of two principals in a given context.  For
two principals, $p$ and $q$, their factorization involves representing $q$ as
the conjunction of two principals $q_s ∧ q_d$ such that $p ≽ q_s$ in the desired
context.  
Note that $p$ need not act for $q_d$, and that 
factorizations always exists. For example, the factorization $q_d=q$ and $q_s=⊥$
is valid for any $p$, $q$, and $Π$.
\begin{definition}[Factorization]
  \label{def:factorize}
A _$Π$-factorization_ of an ordered pair of principals $(p, q)$
is a tuple $(p, q_s, q_d)$ such that
$\reqjudge{\Pi}{q}{q_s ∧ q_d}$ and $\rafjudge{\Pi}{p}{q_s}$.
\end{definition}

Factorization lets us split $q$'s authority into a portion ($q_{s}$) that delegates to $p$, 
and a portion ($q_{d}$) that may or may not. A minimal factorization makes $q_{d}$ as small as possible.
Specifically, a minimal factorization of $p$ and $q$ is a $q_s$ and $q_d$ 
such that $q_s$ has greater authority and $q_d$ has less authority than any
other factorization of $p$ and $q$ in the same context.
\begin{definition}[Minimal Factorization]
  \label{def:min-factorize}
A $Π$-factorization $(p, q_s, q_d)$ of $(p, q)$ is _minimal_ if for any
$Π$-factorization $(p, q_s', q_d')$ of $(p, q)$,
\[
\rafjudge{\Pi}{q_{s}}{q'_s} \text{ and } \rafjudge{\Pi}{q'_{d}}{q_d}
\]
\end{definition}

A minimal factorization $(p, q_s, q_d)$ of $p$ and $q$ for a given $Π$ and
$\pc$ identifies the authority necessary to cause $p$ to
act for $q$. Because $q_s$ is the principal with the greatest authority such
that $p ≽ q_s$ and $q ≡ q_s ∧ q_d$, then speaking for $q_d$ is sufficient
authority to cause $p$ to act for $q$ since adding the delegation $p ≽ q_d$
would imply that $p ≽ q$.  This intuition also matches with the fact that
$\rafjudge{Π}{p}{q_d}$ if and only if $q_d = ⊥$, which is the case if and only if
$\rafjudge{Π}{p}{q}$.  

Observe that minimal $Π$-factorizations are also trivially unique up
to equivalence.
\begin{proposition}[Subtraction equivalence]
  Let $(p, q_s, q_d)$ and $(p, q'_s, q'_d)$ be minimal factorizations of $p$ and
$q$ in $Π$.  Then $\reqjudge{\Pi}{q_s}{q'_s}$ and  $\reqjudge{\Pi}{q_d}{q'_d}$. 
\end{proposition}
\begin{proof}
By Definition~\ref{def:min-factorize}, \rafjudge{\Pi}{q_{s}}{q'_s} and 
 \rafjudge{\Pi}{q'_{s}}{q_s}. Therefore, $\reqjudge{\Pi}{q_s}{q'_s}$.
 Likewise, \rafjudge{\Pi}{q_{d}}{q'_d} and 
 \rafjudge{\Pi}{q'_{d}}{q_d}, so we have $\reqjudge{\Pi}{q_s}{q'_s}$. 
\end{proof}

Since the $q_d$ component of minimal factorization can be thought of as the
``gap'' in authority between two principals, we use $q_d$ to define the notion
of principal _subtraction_.
\begin{definition}[Principal Subtraction]
  \label{def:subtract}
  Let $(p, q_s, q_d)$ be a minimal $Π$-factorization of $(p, q)$.
  We define $q - p$ in $Π$ to be $q_d$.
  That is, $\reqjudge{\Pi}{q - p}{q_d}$.
  Note that $q - p$ is not defined outside of a judgement context.
\end{definition}
Since $q_d$ is unique up to equivalence in $Π$, $q - p$ is also unique 
for a given $\Pi$. 

To further illustrate principal subtraction as an authority gap, consider the
following equivalence:  If a principal acts for the authority gap
between it and any other principal, then it also acts for that principal.
\begin{lemma}[Authority Gap Identity]\label{lemma:agi}
For any $p$ and $q$, $\rafjudge{Π}{p}{q - p} ⇔ \rafjudge{Π}{p}{q}$ 
\end{lemma}
\begin{proof}
Let the minimal factorization of $p$ and $q$ in $Π$ be $(p, q_{s}, q_{d})$ where 
$\reqjudge{Π}{q}{q_{s} ∧ q_{d}}$ and $\rafjudge{Π}{p}{q_{s}}$.
In the forward direction, assume $\rafjudge{Π}{p}{q - p}$.   
For contradiction, assume $\notrafjudge{Π}{p}{q}$.
Then it must be the case that $\notrafjudge{Π}{p}{q_{d}}$.
But by Definition~\ref{def:subtract} $q - p$ is defined as
$q_{d}$, so this implies $\notrafjudge{Π}{p}{q - p}$,  a contradiction.
\end{proof}

Lemma~\ref{lem:minfactex} proves that minimal factorizations exist for all
contexts and principals, so principal subtraction is well defined.
\begin{lemma}[Minimal Factorizations Exist]
  \label{lem:minfactex}
  For any context $Π$ and principals $p, q$, there exists a minimal
  $Π$-factorization of $(p, q)$.
\end{lemma}
\begin{proof}
  Given $(p, q)$, we first let $q_s = p∨q$.
  By definition, $\rafjudge{\Pi}{p}{p∨q}$, and for all factorizations $(p, q_s', q_d')$,
  $\rafjudge{\Pi}{p}{q_s'}$ and $\rafjudge{\Pi}{q}{q_s'}$, so $\rafjudge{\Pi}{q_s}{q_s'}$.

  Now let $D = \{ r ∈ \L \mid \reqjudge{\Pi}{q}{q_s ∧ r} \}$.
  Using FLAM normal form~\cite{flam}, all principals in $\L$ can be represented
  as a finite set of meets and
  joins of elements in $\N ∪ \{⊤,⊥\}$, so $q$ and $q_s$ are finite.
  $Π$ is also finite, adding only finitely-many dynamic equivalences, so $D$ is
  finite up to equivalence. Moreover, since \reqjudge{\Pi}{q}{(p ∨ q) ∧ q} (by absorption) we have $q ∈ D$. Therefore $D$ is
  always non-empty and we can define $q_d = \bigvee D$.

  Now let $(p, q_s', q_d')$ be any $Π$-factorization of $(p, q)$. We must show that $\rafjudge{\Pi}{q_d'}{q_d}$.

  First, see that $\reqjudge{\Pi}{q}{q_s ∧ q_d'}$. for one direction,
observe that $\rafjudge{\Pi}{q}{q_s}$ and $\rafjudge{\Pi}{q}{q_d}$ (by
Definition~\ref{def:prineq}). For the other direction, since
$\rafjudge{\Pi}{q_s}{q_s'}$, we have $\rafjudge{\Pi}{q_s ∧ q_d'}{q_s'
∧ q_d'}$, so $\rafjudge{\Pi}{q_s ∧ q_d'}{q}$.

  Therefore, by the definition of $D$, we know $q_d' ∈ D$, so by the definition of $∨$ and
  $q_d$, $\rafjudge{\Pi}{q_d'}{q_d}$.
  Thus $(p, q_s, q_d)$ is a minimal $Π$-factorization of $(p, q)$.
\end{proof}

We can now state precisely which trust relationships may change in a
given information flow context.\footnote{The original delegation
invariance lemma~\cite{flac} was flawed due to a case where a minor
delegation could have a cascading effect that enabled new delegations,
breaking the desired invariant. This new (slightly more restrictive)
formulation, stated in terms of principal subtraction, addresses more
precisely the connection between the $\pc$ and the invariant trust relationships.}

\begin{lemma}[Delegation Invariance] \label{lemma:di}
  Suppose $\rafjudge{\delegcontext}{\pc}{\voice{t}}$. For all principals $p$ and $q$,
  if \rafjudge{\delegcontext,\delexp{r}{t}}{p}{q}, then either
  \rafjudge{\delegcontext}{p}{q} or $\rafjudge{\delegcontext}{\pc}{\voice{{q} - {p}}}$.  
\end{lemma}
\begin{proof} 
  
There are two cases: either $\rafjudge{\delegcontext}{p}{q}$ or $\notrafjudge{\delegcontext}{p}{q}$.
The case  where \rafjudge{\delegcontext}{p}{q} is trivial.
We prove the other case: if $\notrafjudge{\delegcontext}{p}{q}$ and $\rafjudge{\delegcontext, \delexp{r}{t}}{p}{q}$ then  $\rafjudge{\delegcontext}{\pc}{\voice{{q} - {p}}}$.

       Let $\delegcontext' = \delegcontext, \delexp{r}{t}$. Assume that
	\begin{align}
		\rafjudge{\delegcontext'}{p}{q} \label{eq:di1} \\
		\notrafjudge{\delegcontext}{p}{q} \label{eq:di12}
	\end{align}
        and
        Let $(p, q_s, q_d)$ be the minimal $\delegcontext$-factorization of $(p, q)$. So,
	\rafjudge{\delegcontext}{p}{q_s}
        Since $\notrafjudge{\delegcontext}{p}{q}$, this implies that \notrafjudge{\delegcontext}{p}{q_d}
	but from \eqref{eq:di1}, we have that
	\begin{equation}
		\rafjudge{\delegcontext'}{p}{q_d} \label{eq:di3}
	\end{equation}
	So any derivation of \eqref{eq:di3} must involve the derivation of
	$\rafjudge{\delegcontext'}{r}{t}$ via \ruleref{R-Assume}, since \ruleref{R-Assume} is the only rule that uses the contents of $\Pi$. 
        Therefore, without loss of generality, we can assume that either (a) $\rafjudge{\Pi}{t}{q_d}$,  or (b) for some $q₁$ and $q₂$ such that $\reqjudge{\Pi}{q_d}{q₁ ∧ q₂}$, 
        $\rafjudge{\Pi}{t}{q_1}$ and $\notrafjudge{\Pi}{t}{q_2}$; otherwise \rafjudge{\delegcontext'}{r}{t} is unnecessary to prove \eqref{eq:di3}.  %

        In fact, it must be the case that $\rafjudge{\Pi}{t}{q₁ ∧ q₂}$ (or
equivalently $\rafjudge{\Pi}{t}{q_d}$).  To see why, assume $\notrafjudge{\Pi}{t}{q₁ ∧ q₂}$: specifically, 
$\rafjudge{\Pi}{t}{q_1}$ and $\notrafjudge{\Pi}{t}{q_2}$. Since $\rafjudge{\Pi'}{p}{q_s \wedge (q₁ ∧ q₂)}$ but $\notrafjudge{\Pi}{t}{q_2}$, it must be that 
  $\notrafjudge{\Pi}{p}{q_s ∧ q₂}$, otherwise $\rafjudge{\Pi'}{p}{q_s \wedge (q₁ ∧ q₂)}$ wouldn't hold. %
	Therefore $(p, q_s \wedge q_2, q_1)$ is a $Π$-factorization of $(p,q)$. Since $(p, q_s, q_d)$ is a _minimal_ $Π$-factorization, we have that $\rafjudge{\delegcontext}{q_2}{q_d}$.
	From \ruleref{R-Trans}, we now have that $\rafjudge{\delegcontext}{t}{q_d}$, but since we assumed $\notrafjudge{\Pi}{t}{q_d}$, we have a contradiction. 
        Hence $\rafjudge{\Pi}{t}{q_d}$. 

        By the monotonicity of $\voice{·}$ with respect to $≽$ (Proven in Coq for FLAM~\cite{flam}), we have
        \rafjudge{\Pi}{\voice{t}}{\voice{q_d}}.  As shown above,
        \rafjudge{\Pi}{\pc}{\voice{t}}, so by \ruleref{R-Trans},
        \rafjudge{\Pi}{\pc}{\voice{q_d}}.  However, recall that \reqjudge{Π}{q
        - p}{q_d} (Definition~\ref{def:subtract})  and so
         $\rafjudge{\delegcontext}{\pc}{\voice{q - p}}$.
         Hence proved.
\end{proof}
\begin{corollary}
\label{cor:di}
  Suppose $\rafjudge{\delegcontext,\delexp{r_0}{t_0},...,\delexp{r_n}{t_n}}{\pc}{\voice{t_i}}$ for all $i ∈ [0,n]$.
  For all principals $p$ and $q$, if $\rafjudge{\delegcontext,\delexp{r_0}{t_0},...,\delexp{r_n}{t_n}}{p}{q}$
  then either \rafjudge{\delegcontext}{p}{q} or $\rafjudge{\delegcontext}{\pc}{\voice{{q} - {p}}}$.  
\end{corollary}
\begin{proof}
  For $i=0$, we apply Lemma~\ref{lemma:di}.
  For the inductive case, assume that for all $p$ and $q$, $\rafjudge{\delegcontext,\delexp{r_0}{t_0},...,\delexp{r_{n-1}}{t_{n-1}}}{p}{q}$ implies
  either \rafjudge{\delegcontext}{p}{q} or $\rafjudge{\delegcontext}{\pc}{\voice{{q} - {p}}}$.  
  We want to prove that for all $p$ and $q$, $\rafjudge{\delegcontext,\delexp{r_0}{t_0},...,\delexp{r_{n}}{t_{n}}}{p}{q}$ implies
  either \rafjudge{\delegcontext}{p}{q} or $\rafjudge{\delegcontext}{\pc}{\voice{{q} - {p}}}$ also. 
  By applying Lemma~\ref{lemma:di} to $\rafjudge{\delegcontext,\delexp{r_0}{t_0},...,\delexp{r_{n}}{t_{n}}}{p}{q}$, we obtain that  
  either $\rafjudge{\delegcontext,\delexp{r_0}{t_0},...,\delexp{r_{n-1}}{t_{n-1}}}{p}{q}$ or $\rafjudge{\delegcontext,\delexp{r_0}{t_0},...,\delexp{r_{n-1}}{t_{n-1}}}{\pc}{\voice{{q} - {p}}}$. 
  In the first case, applying the inductive hypothesis gives us that either \rafjudge{\delegcontext}{p}{q} or $\rafjudge{\delegcontext}{\pc}{\voice{{q} - {p}}}$ holds.  
  In the second case,
  applying the inductive hypothesis gives us that either $\rafjudge{\delegcontext}{\pc}{\voice{{q} - {p}}}$ or $\rafjudge{\delegcontext}{\pc}{\voice{\voice{{q} - {p}} - pc}}$ holds.  
   If $\rafjudge{\delegcontext}{\pc}{\voice{{q} - {p}}}$ holds, then we are done.  Therefore, assume $\rafjudge{\delegcontext}{\pc}{\voice{\voice{{q} - {p}} - pc}}$.
  Observe that $\voice{{q} - {p}}$ is an integrity principal. Specifically, $\reqjudge{Π}{\voice{{q} - {p}}^{→}}{⊥}$), so by Definition~\ref{def:factorize} $\reqjudge{Π}{(\voice{{q} - {p}} - pc)^{→}}{⊥}$.
  The voice of an integrity principal is just the principal (Lemma~\ref{lemma:voiceOfInteg}), so $\reqjudge{Π}{\voice{\voice{{q} - {p}} - pc}}{(\voice{{q} - {p}} - pc)}$. 
  Then by Lemma~\ref{lemma:agi}, we know that $\rafjudge{\delegcontext}{\pc}{(\voice{{q} - {p}} - pc)}$ is equivalent to $\rafjudge{\delegcontext}{\pc}{\voice{{q} - {p}}}$. 
 \end{proof} 

If $e$ is a well-typed, closed, source-level FLAC program—in other
words $\TVal{Π; ∅; pc}{e}{τ}$, for some $Π$, $\pc$, and $τ$—then
Lemma~\ref{lemma:di} is sufficient to characterize the delegations
introduced in $e$.  Since $e$ is closed, any delegation
$\delexp{r}{t}$ is introduced by an "assume" term that types under a
sub-derivation of $\TVal{Π; ∅; pc}{e}{τ}$.  Since the $\pc$ only
becomes more restrictive in subexpressions of $e$, we know that 
$\rafjudge{\delegcontext}{\pc}{\voice{t}}$.

More generally, however, an open FLAC programs may receive as input
(non-source-level) values which are typed in higher-integrity
contexts, and thus may use delegations where
$\notrafjudge{\delegcontext}{\pc}{\voice{t}}$.  Such delegations are
constrained only by $\pcmost$ (since
$\rafjudge{\delegcontext}{\pcmost}{\voice{t}}$ for well-typed
"where"-terms).  However, as long as such delegations are
``compartmentalized,'' they cannot be used to downgrade arbitrarily.

To characterize how non-source-level terms can affect downgrading, 
it will be convenient to distinguish source-level programs from their inputs.
We generalize our substitution notation from a single substitution $\subst{}{x}{v}$ 
to a set $S$, where $S(x)=v$ encodes the substitution $\subst{}{x}{v}$ and the substitution 
of all free variables in $e$ that are defined by $S$ is written $e~S$.
\begin{definition}
Given $\TValGpc{e}{τ}$, a _well-typed substitution on $e$ for $Γ$ in $Π$_ is a substitution $S$
 where for each free variable $xᵢ$ of $e$ with $Γ(xᵢ) = τᵢ'$, we have $\TValGpc{S(xᵢ)}{τᵢ'}$ and $e~S$ is closed.
\end{definition}

Now we can state our Delegation Compartmentalization lemma.  As long
as all terms in a well-typed substitution are either source-level
terms or are at least as restrictive as $H^{π}$, then if a
source-level program $e$ evaluates to $\where{w'}{\delexp{r}{t}}$,
either $\rafjudge{\delegcontext}{\pc}{\voice{t}}$ or the result of the
program is also as restrictive as $H^{π}$.

\begin{lemma}[Delegation Compartmentalization] \label{lemma:dc}
	 Suppose \TValGpc{e}{\tau}. Let $S$ be a well-typed substitution on $e$ for $Γ$ in $Π$ where
         $e$ is a source-level term and for all entries $\subst{}{y}{w_y} ∈ S$ such that $\TValGpc{w_y}{Γ(y)}$, and  _either_
         \begin{enumerate}
            \item $w_{y}$ is a source-level term, or 
            \item $\rflowjudge{Π}{H^{π}}{Γ(y)^{π}}$
         \end{enumerate}
         Then if $(e~S) \stepsto \where{w'}{\delexp{r}{t}}$, _either_ 
         \begin{enumerate}[label=(\alph*)]
           \item \label{dc:a} $\rafjudge{\delegcontext}{\pc}{\voice{t}}$,  or 
           \item \label{dc:b} $\rflowjudge{Π}{H^{π}}{τ^{π}}$
         \end{enumerate}
\end{lemma}
\begin{proof}
  From the subject reduction (Lemma~\ref{lemma:subjred}), we have $\TValGpc{\where{w'}{\delexp{r}{t}}}{\tau}$.
  From the typing rule  \ruleref{Where}, we have  $\TValGpc{v}{\aftype{r}{t}}$ for some $r$ and $t$, and
  $\TVal{\Pi'; \G; \pc}{w'}{\tau}$ for $\Pi' = \Pi, \delexp{r}{t}$. There are two cases.

\begin{description}
\item[Case 1:] Suppose $(e~S) \stepsto \where{w'}{\delexp{r}{t}}$ such that $\delexp{r}{t}$ belongs to $S$.
We know that all values substituted from $S$ are well-typed w.r.t $\Pi$ and $\pc$. Without loss of generality, let $\delexp{r}{t}$ be reduced from $w_y$. We have two cases:
\begin{description}
\item[Case $w_y$ is a source-level term:] We are given $\TValGpc{w_y}{Γ(y)}$. Since $w_y$ does not have any "where" terms, it must be the case that either $w_y$ itself is $\delexp{r}{t}$ expression or $\delexp{r}{t}$ appears in some "assume" term in $w_y$. The latter is only possible if $w_y$ is a lambda or type abstraction, because no other values can embed "assume" terms. From the monotonicity of $\pc$ (lemma~\ref{lemma:pcmonotone}), we have that any delegation $\delexp{r}{t}$ propagated from $w_y$ should satisfy $\rafjudge{\delegcontext}{\pc}{\voice{t}}$. Hence proved \ref{dc:a}
  
\item[Case $\rflowjudge{Π}{H^{π}}{Γ(y)^{π}}$:] Since \ref{dc:a} does
  not hold, we have that $\delexp{r}{t}$ appears in $S$. Without loss
  of generality, assume that $\delexp{r}{t}$ appears in $w_y$ and that
  $w_y$ is "where" term. Invoking Lemmas \ref{lemma:protectwhere} and
  \ref{lemma:protectwhere2} that state how delegations propagate from
 "where" terms , we have the required proof.
\end{description}
  
\item[Case 2:]
Suppose $(e~S) \stepsto \where{w'}{\delexp{r}{t}}$ such that $\delexp{r}{t}$ did not propagate from $S$.
Then, it must be the case that $(e~S) \stepsto E[\assume{\delexp{r}{t}}{e'}] \stepsto \where{w'}{\delexp{r}{t}}$.
From subject reduction (Lemma~\ref{lemma:subjred}), we have that
for some $\Pi'$, $\G'$, $\pc'$ and $\tau'$, $\TVal{\Pi'; \G'; \pc'}{\assume{\delexp{r}{t}}{e'}}{\tau'}$ and so $\rafjudge{\delegcontext'}{\pc'}{\voice{t}}$ (from the typing rule \ruleref{Assume}). Note that $E[\assume{\delexp{r}{t}}{e'}]$ gives us a valid $T[\assume{\delexp{r}{t}}{e'}]$ such that $T=E$ (Lemma~\ref{lemma:evalimpliessub}). Invoking the monotonicity of $\pc$ (Lemma~\ref{lemma:pcmonotone}) on \TValGpc{T[\assume{\delexp{r}{t}}{e'}]}{\tau}, we have that $ \rflowjudge{\Pi}{\pc}{\pc'}$. Since $\Pi'$ is an extension of $\Pi$ (Lemma~\ref{lemma:pimonotone}), we have $ \rflowjudge{\Pi'}{\pc}{\pc'}$.
Applying transitivity on $\rflowjudge{\Pi'}{\pc}{\pc'}$ and  $\rafjudge{\Pi'}{\pc'}{\voice{t}}$,  we have $\rafjudge{\Pi'}{\pc}{\voice{t}}$. Depending on the set difference $\Pi' - \Pi$, we have two more cases:
\begin{description}
\item[Case 2.1:] Case where some of the delegations in $\Pi' -\Pi$ have propagated from $S$. Then, going by the argument similar to the previous case, we have that either $\rafjudge{\Pi}{\pc}{\voice{t}}$ or $\rflowjudge{Π}{H^{π}}{τ^{π}}$.  Hence proved. 

\item[Case 2.2:] Case where none of the delegations in $\Pi' -\Pi$
  have propagated from $S$. Then, by \ruleref{Assume}, monotonicity of the $\pc$ (Lemma~\ref{lemma:pcmonotone}),  and \ruleref{R-Trans}, for each delegation $\delexp{a}{b}$ in
  $\Pi' - \Pi$, we have that $\rafjudge{\Pi'}{\pc}{\voice{b}}$.  Applying Corollary~\ref{cor:di} to $\rafjudge{\Pi'}{\pc}{\voice{t}}$ gives us 
  either $\rafjudge{\Pi}{\pc}{\voice{t}}$ or $\rafjudge{\Pi}{\pc}{\voice{\voice{t} - \pc}}$, which are equivalent by Lemma~\ref{lemma:agi}. 
\end{description}

\end{description}

\end{proof}

\subsection{Noninterference}
Lemmas~\ref{lemma:di} and~\ref{lemma:dc} are critical for our proof of _noninterference_,
a result that states that 
public and trusted output of a program cannot depend on restricted (secret or untrustworthy) information.  
Our proof of noninterference for FLAC programs relies on a proof of subject reduction under
a bracketed semantics, based on the proof technique of Pottier and
Simonet~\cite{ps03}. This technique is mostly standard, so we omit it here.
The complete proof of subject reduction and other results are found in Appendix~\ref{app:proofs}.

In other noninterference results based on bracketed semantics, including~\cite{ps03},
noninterference follows almost directly from the proof of subject reduction.
This is because the subject reduction proof shows that evaluating a term cannot
change its type.
In FLAC, however, subject reduction alone is insufficient; evaluation may 
enable flows from secret or untrusted inputs to public
and trusted types.

To see how, suppose $e$ is a well-typed program according to $\TValP{Γ,x\ty
τ;\pc}{e}{τ'}$.  Furthermore, let $H$ be a principal such that
$\protjudge*{H}{τ}$ and $\notprotjudge*{H}{τ'}$.  In other words, $x$ is a
``high'' variable (more restrictive; secret and untrusted), and $e$ evaluates to a
``low'' result (less restrictive; public and trusted).  In \cite{ps03},
executions that differ only in secret or untrusted inputs must evaluate to the
same value, since otherwise the value would not be well typed.  In FLAC,
however, if the $\pc$ has sufficient integrity, then an "assume" term could
cause $\protjudge{Π';\pc}{H}{τ'}$ to hold in a delegation context $Π'$
of a subterm of $e$.  Additionally, even if the $\pc$ is low integrity, an
input to the program may capture a delegation that the source program
could not assume directly.  For example, a dynamic hand-off term like those
discussed in Section~\ref{sec:handoff}, with type
$$\tfuncpc{X}{\pc}{(\func{\says{p}{X}}{\pc}{\says{q}{X}})}$$ 
could enable the same flows as a delegation $\delexp{p}{q}$, but without the condition
that $\rafjudge{Π}{\pc}{∇(q)}$.

The key to proving our result relies on using Lemma~\ref{lemma:di} to
constrain the delegations that can be added to $Π'$ by the source-level terms,
Lemma~\ref{lemma:dc} to specify how non-source-level terms must be
compartmentalized.  Thus
noninterference in FLAC is dependent on $H$ and its relationship to
$\pc$ and the type $τ'$. For confidentiality, most of this reasoning occurs in the proof of
Lemma~\ref{lemma:erasecons}, 
which does most of the heavy lifting
for the noninterference proof.  Specifically,
Lemma~\ref{lemma:erasecons} states that for a well-typed program
$e$, if the $\pc$ is insufficiently trusted to create new flows from
$H^{→}$ to $ℓ^{→}$, then if the portion of \outproj{e}{1} observable
to $ℓ^{→}$ under $Π$ is equal to the observable portion of
\outproj{e}{2}, then if $e \stepsto e'$, the observable portions of
$e'$ are still equal.

\begin{restatable}[Confidentiality Erasure Conservation]{lemma}{erasecons}
  \label{lemma:erasecons}
  Suppose \TValGpc{e}{\tau} and let $S$ be a well-typed substitution of $e$ for $Γ$ in $Π$.
  Then for some $H$ and $ℓ$ such that $\notrafjudge{Π}{ ℓ^{→}}{H^{→}}$ and  $\notrafjudge{Π}{\pc}{∇(H^{→} - ℓ^{→})}$,
  if $\observefc{\outproj{e~S}{1}}{Π}{ℓ^{→}} = \observefc{\outproj{e~S}{2}}{Π}{ℓ^{→}}$ and
    $e$ is a source-level term, and for all entries $\subst{}{y}{w_y} ∈ S$ with $\TValGpc{w_y}{Γ(y)}$, either $w_{y}$ is a
source-level term or \rflowjudge{Π}{H^{→} ∧ ⊤^{←}}{Γ(y)}
 then $(e~S) \stepsone e'$ implies $\observefc{\outproj{e'}{1}}{Π}{ℓ^{→}} = \observefc{\outproj{e'}{2}}{Π}{ℓ^{→}}$. 
\end{restatable}
\begin{proof}[Proof Sketch]
  By induction on the evaluation of $e$. The most interesting case is the step \ruleref{O-Ctx} where the term escapes it protection context. Here we invoke the delegation invariance (Lemma~\ref{lemma:di}) to argue that the delegations in a well-typed program cannot close the authority gap between $H^\rightarrow$ and $\ell^{\rightarrow}$, and thus do not change the observability of a term.

  Detailed proof is presented in Appendix~\ref{app:erasecons}.
\end{proof}

Lemma~\ref{lemma:erasecons} holds for mutiple steps as well and is presented in the Corollary~\ref{lemma:eraseconssteps}.

\begin{corollary}[Erasure Conservation for Multiple Steps]\label{lemma:eraseconssteps}
  Suppose \TValGpc{e}{\tau} and let $S$ be a well-typed substitution of $e$ for $Γ$ in $Π$.
  Then for some $H$ and $ℓ$ such that $\notrafjudge{Π}{ ℓ^{→}}{H^{→}}$ and  $\notrafjudge{Π}{\pc}{∇(H^{→} - ℓ^{→})}$,
  if $\observefc{\outproj{e~S}{1}}{Π}{ℓ^{→}} = \observefc{\outproj{e~S}{2}}{Π}{ℓ^{→}}$ and
    $e$ is a source-level term, and for all entries $\subst{}{y}{w_y} ∈ S$ with $\TValGpc{w_y}{Γ(y)}$, either $w_{y}$ is a source-level term or \rflowjudge{Π}{H^{→}}{Γ(y)^{→}}
 then $(e~S) \stepsto e'$ implies $\observefc{\outproj{e'}{1}}{Π}{ℓ^{→}} = \observefc{\outproj{e'}{2}}{Π}{ℓ^{→}}$. 
\end{corollary}
\begin{proof}
  From from the transitive closure of Lemma~\ref{lemma:erasecons}.
\end{proof}

Theorem~\ref{thm:niconf} concerns programs with an input of type $τ'$ and an output of type 
$τ$. The input $x$ is secret, thus $τ'$ must protect $H^{→}$ (Condition~\ref{c:nihigh}).
The outputs, however, are public and therefore information derived from the inputs 
should not flow to the outputs (Condition~\ref{c:noflow}). 
Therefore, if $\pc$ is insufficiently trusted to create new flows from $H^{→}$ to $ℓ^{→}$ 
(Condition~\ref{c:nodown}) 
then executions of $e$
that differ only in the value of $τ'$-typed inputs values must produce traces that are indistinguishable.

\begin{theorem}[Confidentiality Noninterference]
  \label{thm:niconf}
  Let $\TValP{\G;x:τ';\pc}{e}{τ}$ for some $H$ and $ℓ$ such that 
\begin{enumerate}
\item $\rflowjudge{\Pi}{H^{→}}{τ'^{→}}$  \label{c:nihigh}
\item $\notrflowjudge{\Pi}{H^{→}}{ℓ^{→}}$  \label{c:noflow}
\item$\notrafjudge{Π}{\pc}{∇(H^{→} - ℓ^{→})}$\label{c:nodown} 
\end{enumerate}
If $e$ is a source-level term, and $S$ is a well-typed substitution
for $Γ$ in $Π$ where for all entries $\subst{}{y}{w_y} ∈ S$ with $\TValGpc{w_y}{Γ(y)}$, either $w_{y}$ is a
source-level term or \rflowjudge{Π}{H^{→} ∧ ⊤^{←}}{Γ(y)}. Then for all  $w_z, z \in \{ 1, 2 \}$ such that \TValGpc{w_z}{τ'}, if
$\subst{(e~S)}{x}{w_z} \stepstot{t_{z}} w'_{z}$, then $t_{1} ≈^{Π}_{ℓ^{→}} t_{2}$.
\end{theorem}
\begin{proof}
       From the soundness and completeness properties of the bracketed
       language (Lemmas~\ref{lemma:sound} and~\ref{lemma:complete}), we can
       construct a bracketed execution
        $\subst{(e~S)}{x}{\bracket{v_1}{v_2}} \stepstot{t} v'$
        such that $\outproj{v'}{z} = v'_{z}$ and $\outproj{t}{z} = t_{z}$ for $z = \{1, 2\}$. We
        will occasionally write $v$ or $v'$ as shorthand for
        $\bracket{v₁}{v₂}$ or $\bracket{v'₁}{v'₂}$.

       Since \TValGpc{v_z}{τ'} for $z ∈ \{1,2\}$, then we have
       \TValGpc{\bracket{v₁}{v₂}}{τ'} via \ruleref{Bracket-Values}.
       Therefore, by Lemma~\ref{lemma:vsubst} (Variable Substitution) of
       $\TValP{\G;x:τ';\pc}{e}{\says{ℓ}{τ}}$, we have
       $\TValP{\G;\pc}{\subst{e}{x}{\bracket{v₁}{v₂}}}{\says{ℓ}{τ}}$.
       Then by induction of the number of evaluation steps in $\subst{(e~S)}{x}{v}
       \stepstot{t} e'$ and subject reduction (Lemma~\ref{lemma:subjred}), we have $\TValP{\G;\pc}{e'}{\says{ℓ}{τ}}$.

        We now want to prove that $t_{1} ≈^{Π}_{ℓ^{→}} t_{2}$. First, consider $\subst{(e~S)}{x}{\bracket{v_1}{v_2}}$. 
        To prove that $\observefc{\subst{(e~S)}{x}{v_1}}{Π}{ℓ^{→}} =\observefc{\subst{(e~S)}{x}{v₂}}{Π}{ℓ^{→}}$,   
        it suffices to show that $\observefc{v_1}{Π}{ℓ^{→}} =\observefc{v₂}{Π}{ℓ^{→}}$.
        Note that conditions \ref{c:nihigh} and \ref{c:noflow} satisfy all the necessary conditions for invoking Lemma~\ref{lemma:correctobserve}.
        By Lemma~\ref{lemma:correctobserve} (erasure of projected protected values is equal), we have that
        $\observefc{v_1}{Π}{ℓ^{→}} =\observefc{v₂}{Π}{ℓ^{→}}$.
        Thus $\observefc{\outproj{\subst{e}{x}{v}}{1}}{Π}{ℓ^{→}} =\observefc{\outproj{\subst{e}{x}{v}}{2}}{Π}{ℓ^{→}}$. 
        Invoking erasure conservation (Corollary~\ref{lemma:eraseconssteps}), we
        have $\observefc{\outproj{e'}{1}}{Π}{ℓ^{→}} =\observefc{\outproj{e'}{2}}{Π}{ℓ^{→}}$.
        Thus  $t_{1} ≈^{Π}_{ℓ^{→}} t_{2}$.
\end{proof}

Specializing our noninterference results on confidentiality provides more
precision, but integrity versions of Lemma~\ref{lemma:erasecons} and
Theorem~\ref{thm:niconf} hold by similar arguments.  We present
their statements here, but not the corresponding proofs.

\begin{lemma}[Integrity Erasure Conservation]
  \label{lemma:ierasecons}
  Suppose \TValGpc{e}{\tau} and let $S$ be a well-typed substitution of $e$ for $Γ$ in $Π$.
  Then for some $H$ and $ℓ$ such that $\notrafjudge{Π}{H^{←}}{ ℓ^{←}}$ and  $\notrafjudge{Π}{\pc}{ℓ^{←}-H^{←}}$,
  if $\observefi{\outproj{e~S}{1}}{Π}{ℓ^{←}} = \observefi{\outproj{e~S}{2}}{Π}{ℓ^{←}}$ and
    $e$ is a source-level term, and for all entries $\subst{}{y}{w_y} ∈ S$ with $\TValGpc{w_y}{Γ(y)}$, either $w_{y}$ is a
source-level term or \rflowjudge{Π}{H^{←}}{Γ(y)^{←}}
 then $(e~S) \stepsone e'$ implies $\observefi{\outproj{e'}{1}}{Π}{ℓ^{←}} = \observefi{\outproj{e'}{2}}{Π}{ℓ^{←}}$. 
\end{lemma}

\begin{theorem}[Integrity Noninterference]
  \label{thm:niinteg}
  Let $\TValP{\G;x:τ';\pc}{e}{τ}$ for some $H$ and $ℓ$ such that 
\begin{enumerate}
\item $\rflowjudge{\Pi}{H^{←}}{τ'}$  %
\item $\notrflowjudge{\Pi}{H^{←}}{ℓ^{←}}$  %
\item $\notrafjudge{Π}{\pc}{ℓ^{←}-H^{←}}$ %
\end{enumerate}
If $e$ is a source-level term, and $S$ is a well-typed substitution
for $Γ$ in $Π$ where for all entries $\subst{}{y}{w_y} ∈ S$ with $\TValGpc{w_y}{Γ(y)}$, either $w_{y}$ is a
source-level term or \rflowjudge{Π}{H^{←}}{Γ(y)^{←}}. Then for all  $w_z, z \in \{ 1, 2 \}$ such that \TValGpc{w_z}{τ'}, if
$\subst{(e~S)}{x}{w_z} \stepstot{t_{z}} w'_{z}$, then $t_{1} ≈^{Π}_{ℓ^{←}} t_{2}$.
\end{theorem}

Noninterference is a key tool for obtaining many of the security properties we
seek.  For instance, noninterference is essential for verifying the properties of
commitment schemes and bearer credentials discussed in Section~\ref{sec:ex}.  

Unlike some definitions of noninterference, our definition does not
prohibit all downgrading. Instead, Conditions~\ref{c:nihigh}
and~\ref{c:noflow} define relationships between principals $H$ and
$ℓ$, and Condition~\ref{c:nodown} ensures (via
Lemma~\ref{lemma:di}) those relationships remain unchanged by
delegations made within the program, and thus by Lemma~\ref{lemma:erasecons},
the trace of a program observable to an attack remains unchanged during evaluation.

\subsection{Robust declassification}

_Robust declassification_~\cite{zm01b} requires disclosures of secret information to be
independent of low-integrity information.  Robust declassification
permits some confidential information to be disclosed to an attacker,
but attackers can influence neither the decision to disclose
information nor the choice of what information is disclosed.
Therefore, robust declassification is a more appropriate security
condition than noninterference when programs are intended to disclose
information.

Following Myers et al.~\cite{msz06}, we extend our set of terms with a ``hole'' term $[•^{τ}]$
representing portions of a program that are under the control of an attacker.  Attackers may insert any
term of type $τ$ to complete the program.
We extend the type system with the following rule for holes, parameterized on the same $H$
used by the bracketed typing rules \ruleref{Bracket} and \ruleref{Bracket-Values}:
\par\nobreak
{\small
   \begin{mathpar}
     \Rule{Hole}{
      \rafjudge{Π}{H^{←}}{\pc^{←}} \\
      \rflowjudge{Π}{\pc^{→}}{Δ(H^{←})} \\
     }{\TValGpc{\bullet^{τ}}{τ}}
    \hfill 
    \end{mathpar}
}

Where $Δ(H^{←})$ is the _view_ of principal $H^{←}$. The view of a
principal is a dual notion to voice, introduced by Cecchetti \emph{et
al.}~\cite{nmifc} to represent the confidentiality observable to a principal.  
Given a principal in normal form $\confid q \tjoin \integ r$, the view 
of that principal is
\[
    Δ({\confid q \tjoin \integ r}) \triangleq \confid q \tjoin \confid r
\] In other words $Δ(H^{←})$ represents the confidentiality of information observable to the attacker $H^{←}$.
Therefore, the \ruleref{Hole} premises \rafjudge{Π}{H^{←}}{\pc^{←}} and \rflowjudge{Π}{\pc^{→}}{Δ(H^{←})} 
require that holes be inserted only in contexts that are controlled by and observable to the attacker.

We write $e[\vec{•}^{\vec{τ}}]$ to denote a program $e$ with holes.  Let an _attack_ be a
vector $\vec{a}$ of terms and $e[\vec{a}]$ be the program where $aᵢ$ is
substituted for $•ᵢ^{τᵢ}$.

\begin{definition}{$Π$-fair attacks}
  \label{def:fair}
An attack $\vec{a}$ is a $Π$-_fair attack_ on a
well-typed program with holes $\TValGpc{e[\vec{•}^{\vec{τ}^{*}}]}{τ}$ if the program $e[\vec{a}]$ is also well typed (thus $\TValGpc{e[\vec{a}]}{τ}$),
and furthermore, for each $aᵢ ∈ \vec{a}$, we have $\TValP{Γ^{*}_i;H^{←} ∧ Δ(H^{←})}{aᵢ}{τ^{*}_i}$ where each entry in $Γ^{*}_i$ has the form $y\ty \says{ℓ}{τ''}$  
and $\rflowjudge{Π}{ℓ^{→}}{Δ(H^{←})}$.
\end{definition}

By specifying the relationships between confidential information
labeled $H^{→}$ and the integrity of the attacker $H^{←}$, we can
precisely express the authority an attacker is able to wield in its 
$Π$-fair attacks against protected information.  Proposition~\ref{prop:niatk} 
states that as long as the attacker cannot observe secret information labeled $H^{→}$, or $\notrflowjudge{Π}{H^{→}}{Δ(H^{←})}$,  
then $Π$-fair attacks cannot reference secret variables directly.

\begin{proposition}[No free secrets]
  \label{prop:niatk}
  For $H$ and $Π$ such that $\notrflowjudge{Π}{H^{→}}{Δ(H^{←})}$,
  suppose $\vec{a}$ is a $Π$-fair attack on a program $e[\vec{•}^{\vec{τ}^{*}}]$ where $\TValP{\G,x\ty τ';\pc}{e[\vec{•}^{\vec{τ}^{*}}]}{τ}$. Then \rflowjudge{Π}{H^{→}}{τ'} implies
  $x$ is not a free variable of $aᵢ ∈ \vec{a}$ for all $i$.
\end{proposition}
\begin{proof}
Since $\vec{a}$ is a $Π$-fair attack, we have $\TValP{Γ'ᵢ;H^{←} ∧ Δ(H^{←})}{aᵢ}{τ^{*}_i}$ for each $aᵢ$ in $\vec{a}$ and $τ^{*}ᵢ$ in $\vec{τ}^{*}$.  By assumption,
\rflowjudge{Π}{H^{→}}{τ'} and \notrflowjudge{Π}{H^{→}}{Δ(H^{←})}, so by the definition of $Π$-fair attacks, $x$ cannot be in $Γ'ᵢ$.
Therefore, $x$ is not free in $aᵢ$.
\end{proof}

In other words, $Π$-fair attacks must be well-typed on variables observable to the attacker in delegation context $Π$.
Unfair attacks give the attacker enough power to break
security directly by creating new declassifications without exploiting existing ones.  Restricting attacks to noninterfering 
$Π$-fair attacks also rule out nonsensical scenarios such as when the ``attacker'' has the authority to read the confidential information.
Fair attacks under the conditions of Proposition~\ref{prop:niatk} represent a reasonable upper-bound on the power of the attacker over
low-integrity portions of the program.  Intuitively, low-integrity portions of the program are treated as though they are executed on
a host controlled by the attacker.

Our robust declassification theorem is stated in terms of an attacker's delegation
context $Π_{H}$ which is used to specify the noninterfering $Π_{H}$-fair
attacks on a program. The program itself is typed under a distinct delegation
context $Π$ that may be extended by \texttt{assume} terms to permit
new flows of information, including flows observable to the attacker.  

\begin{restatable}[Robust declassification]{theorem}{robdecl}
  \label{thm:robdecl}
Suppose $\TValP{Γ,x\ty τ',Γ';\pc}{e[\vec{•}^{\vec{τ}^{*}}]}{τ}$.
For $Π_{H}$ and $H$ such that 
  \begin{enumerate}
    \item $\protjudge{Π_H}{H^{→}}{τ'}$ 
    \item $\notrflowjudge{Π_H}{H^{→}}{Δ(H^{←})}$\label{cond:rd2} 
    \item $\notrafjudge{Π_H}{H^{←}}{∇(H^{→})}$, \label{cond:rd3} 
  \end{enumerate}
Then for all $Π_{H}$-fair attacks $\vec{a₁}$ and $\vec{a₂}$  such that  $\TValP{Γ,x\ty τ',Γ';\pc}{e[\vec{aᵢ}]}{τ}$ and $\TValP{Γ;\pc}{vᵢ}{τ'}$, if $\subst{e[a_{j}]}{x}{v_{i}} \stepsone e'_{ij}$ for $i,j ∈ \{1,2\}$,
then for the traces $t_{ij}=(\subst{e[a_{j}]}{x}{v_{i}})·e'_{ij})$, we have $$t_{11} ≈^{Π}_{Δ(H^{←})} t_{21} \iff t_{12} ≈^{Π}_{Δ(H^{←})} t_{22}$$
\end{restatable}
\begin{proof}[Proof Sketch]
 By induction on the evaluation of $e$. For detailed proof, refer to Section~\ref{app:robdecl}.
\end{proof}

Our formulation of robust declassification is in some sense more general than
previous definitions since it permits some endorsements. Previous definitions of
robust declassification~\cite{msz06,zm01b} forbid endorsement altogether.
_Qualified robustness_~\cite{msz06} treats endorsed values as having an arbitrary 
value, so executing a program with endorsements generates a \emph{set} of traces. Two
executions are considered indistinguishable under qualified robustness if, for each trace 
generated by one execution, there exists a low-equivalent trace that is generated by the 
other execution. 
Our definition of robust declassification does not apply to programs that endorse and then 
declassify values that may have been influenced by the attacker. Qualified robustness
does apply to such programs.  

For some programs, FLAC offers a stronger guarantee than the
possibilistic one offered by qualified robustness.  Consider
$H="bob"^{←} ∧ "alice"^{→}$ and two secret inputs
$\returng{"alice"^{→} ∧ "carol"^{←}} vᵢ$ for $i∈{1, 2}$.  Program
$p[•^{τ'}]$ declassifies "alice"'s secret to "bob", then passes the result
to a function that endorses it from "carol" to "bob".  The result is
passed to a function controlled by the attacker, $"bob"^{←}$.

\begin{align*}
  "declass" = & \lamc{y}{\says{"alice"^{→} ∧ "carol"^{←}}{τ}}{"alice"^{←}}{} & \\
                     &  \qquad {\assume{\delexp{"bob"}{"alice"}}{(\bind{y'}{{y}}{\returng{"bob"^{→} ∧ "carol"^{←}}{y'}})}} & \\
  "endorse" = &\lamc{y}{\says{"bob"^{→} ∧ "carol"^{←}}{τ}}{"carol"^{←}}{} & \\
                      & \qquad {\assume{\delexp{"bob"}{"carol"}}{(\bind{y'}{y}{\returng{"bob"^{→} ∧ "bob"^{←}}{y'}})}} & \\ 
  p[•^{τ'}] = &(\lamc{y}{\says{"bob"}{τ}}{"bob"^{←}}{[•^{τ'}]}) ~ ("endorse"~("declass"~x)) & \\
\end{align*}
\noindent $p$ is well typed in the following context: 
{\small
$$\TVal{\delexp{"alice"^{←}}{"bob"^{←}}, \delexp{"carol"^{←}}{"bob"^{←}}; x: \says{"alice"^{→} ∧ "carol"^{←}}{τ};("alice" ∧ "carol")^{←}}{p[•^{τ'}]}{\says{"bob"}{τ}}$$
}

The program $p$ and choice of $H$ fullfil the conditions of
Theorem~\ref{thm:robdecl} with the secret inputs substituted for the
variable $x$ and $Π_{H}=Π=\{\delexp{"alice"^{←}}{"bob"^{←}},
\delexp{"carol"^{←}}{"bob"^{←}}\}$.  Thus no fair attacks substituted
into the hole can violate robust declassification, for all traces
generated by different choices of attacks and inputs.  
Under the qualified robustness semantics, the value returned by
"endorse" would be treated as an arbitrary value and would guarantee
only that, out of all of the traces generated by each choice of return
value, there _exists_ some trace that satisfies robustness.

Theorem~\ref{thm:robdecl} also permits some declassifications that prior
definitions of robust declassification reject.  For example, our definition admits
declassifications of $x$ even if $\rflowjudge{Π}{H^{←}}{τ'}$. In other words, 
even though low-integrity attacks cannot influence declassification, 
it is possible to declassify a secret input that has low-integrity. Therefore, 
an attacker that is permitted to influence the secret input could affect how 
much information is revealed.
This is an example of a _malleable information flow_~\cite{nmifc}, which neither FLAC nor prior definitions (in the presence of endorsement)
prevent in general.  
Cecchetti \emph{et al.}~\cite{nmifc} present a language 
based on FLAC that replaces "assume" with restricted
declassification and endorsement terms to enforce nonmalleable
information flow.

\section{Examples revisited}
\label{sec:exre}
We now implement our examples from Section~\ref{sec:ex} in FLAC and discuss
their formal properties.
Using FLAC ensures that authority 
and information flow assumptions are explicit, and that programs using these
abstractions are secure with respect to those assumptions.
In this section, we discuss how FLAC types 
help enforce specific end-to-end security properties for commitment schemes and bearer credentials.

\subsection{Commitment Schemes}
\label{sec:commit}

\begin{figure}
{\small
\begin{flalign*}
& \texttt{commit}\ty \tfuncpc{N}{p^{←}}{\tfuncpc{X}{p^{←}}{\func{N}{p^{←}}{\func{\says{p^{→}}{X}}{p^{←}}{\says{p}{(N,X)}}}}} &\\
& \texttt{commit} = \tlam{N}{p^{←}}{\tlam{X}{p^{←}}{\lamc{n}{N}{p^{←}}{\lamc{x}{\says{p^{→}}{X}}{p^{←}}{}}}} &\\
&\quad \qquad \qquad  \assume{\delexp{⊥^{←}}{p^{←}}}{\bind{x'}{x}{\return{p}{(n,x')}}}
\end{flalign*}
\vspace{-1.5em}
\begin{flalign*}
& \texttt{reveal} \ty \tfuncpc{N}{∇(p^{→})}{\tfuncpc{X}{q^{←}}{\func{\says{p}{(N,X)}}{q^{←}}{\says{q^{→} ∧ p^{←}}{(N,X)}}}} &\\
& \texttt{reveal} = \tlam{N}{∇(p^{→}) ∧ p^{←}}{\assume{\delexp{∇(q^{→})}{∇(p^{→})}}{\assume{\delexp{q^{←}}{p^{←}}}{}}} \\
&\quad \qquad \qquad \assume{\delexp{q^{→}}{p^{→}}}{\tlam{X}{q^{←}}{\lamc{x}{\says{p}{(N,X)}}{q^{←}}{\bind{x'}{x}{\return{q^{→} ∧ p^{←}}{x'}}}}} \\ 
\end{flalign*}
\vspace{-1.5em}
\begin{flalign*}
& \texttt{open} \ty \tfuncpc{N}{q^{←}}{\tfuncpc{X}{q^{←}}{\func{(\tfuncpc{Y}{q^{←}}{\func{\says{p}{(N,Y)}}{q^{←}}{\says{q^{→} ∧ p^{←}}{(N,Y)}}})}{q^{←}}{}}} & \\
& \qquad  \qquad \func{\says{p}{(N,X)}}{q^{←}}{\says{q^{←}}{(\says{q^{→} ∧ p^{←}}{(N,X)})}} &\\
& \texttt{open} =\tlam{N}{q^{←}}{\tlam{X}{q^{←}}{\lamc{f}{(\tfuncpc{Y}{q^{←}}{\func{\says{p}{(N,Y)}}{q^{←}}{\says{q^{→} ∧ p^{←}}{(N,Y)}}})}{q^{←}}{}}} &\\
&\quad \qquad \qquad  \lamc{x}{\says{p}{(N,X)}}{q^{←}}{\return{q^{←}}{(f~X~x)}}
\end{flalign*}
}
\caption{FLAC implementations of commitment scheme operations.}
\label{fig:commit}
\end{figure}

Figure~\ref{fig:commit} contains the essential operations of a
commitment scheme—"commit", "reveal", and "open"—implemented in FLAC.
Principal $p$ commits to pairs of the form $(n,v)$ where $n$ is a
term that encodes a type-level integer $N$.
Any reasonable integer encoding is
permissible provided that each integer $N$ is a singleton type. For example, we could use pairs
to define the zero type as $\voidtype$ and the successor type as $\prodtype{\voidtype}{τ}$ where $τ$ is
any valid integer type. Thus an integer $N$ would be represented by $N$ nested pairs.
To prevent $q$ from influencing which committed values are revealed, 
$p$ commits to a single value for each integer type.

The "commit" operation takes a value of any type (hence $∀X$) with confidentiality $p^{→}$ and produces a
value with confidentiality and integrity $p$. In other words, $p$ _endorses_
\cite{zznm02} the value to have integrity $p^{←}$.

Attackers should not be able to influence whether principal $p$
commits to a particular value.  The $\pc$ constraint on "commit" ensures that only principal $p$ and
principals trusted with at least $p$'s integrity, $p^{←}$, may apply "commit" to
a value.
Furthermore, if the programmer omitted this constraint or
instead chose $⊥^{←}$, then "commit" would be rejected by the
type system. Specifically, the "assume" term would not type-check via rule
\ruleref{Assume} since the $\pc$ does not act for $∇(p^{←})$, which is equal to $p^{←}$. 

When $p$ is ready to reveal a previously committed value, it instantiates the function
"reveal" with the integer type $N$ of the committed value and sends the result to $q$.
The body of "reveal" creates delegations \delexp{∇(q^{→})}{∇(p^{→})} and \delexp{q^{→}}{p^{→}}, permitting
the inner function to relabel its argument from $p$ to $q^{→} ∧ p^{←}$.
The outer "assume" term establishes that principals speaking for $q^{→}$ 
also speak for $p^{→}$ by creating an integrity relationship between their voices. 
 With this relationship in place, the inner "assume" term may delegate $p$'s confidentiality to $q$.\footnote{
Specifically, it satisfies the \ruleref{Assume} premise $\rafjudge{Π}{\voice{\confid{p}}}{\voice{\confid{q}}}$.}

Only principals trusted by $∇(p^{→})$ and $p^{←}$ may instantiate "reveal" with the integer $N$ of the commitment, 
therefore $q$ cannot invoke "reveal" 
until it is authorized by $p$. If, for example, the type function had $\pc$ annotation $q$ instead of
$∇(p^{→})$, it would be rejected by the type system since the "assume" term would not check under the \ruleref{Assume} rule.
Note however, that once the delegation is established, the inner function may be invoked by $q$ since it
has the $\pc$ annotation $q^{←}$.  Without the delegation, however,
the "bind" term would fail to type-check under \ruleref{BindM}
since $p$ would not flow to $q^{→} ∧ p^{←}$.

Given a function of the type instantiated by "reveal" and a committed value 
of type $\says{p}{X}$, the "open" function may be 
invoked by $q$ to relabel the committed value to a type observable by $q$.
Because "open" may only be invoked by principals trusted by $q^{←}$, $p$ cannot
influence which value of type $\says{p}{(N,X)}$ "open" is applied to. If $q$ only accepts a
one value for each integer type $N$, then the argument must be the one originally 
committed to by $p$.
Furthermore, since the "reveal" function provided to $p$ is parametric with respect to $X$, it 
cannot return a value other than (a relabeled version of) its argument.
The return type $\says{q^{←}}{\says{q^{→} ∧ p^{←}}{X}}$ protects this relationship 
between the committed value and the opened one: it indicates that $q$ trusts 
that $\says{q^{→} ∧ p^{←}}{X}$ is the one originally committed to by $p$.
\footnote{The original commitment scheme presented in Arden and Myers~\cite{flac} 
contained a "receive" and "open" terms that are rejected by our type system because of the
updated \ruleref{UnitM} rule. The new premise $\rflowjudge{Π}{\pc}{ℓ}$ would require
these terms to be executed at a $\pc$ trusted by both $p$ and $q$. Since such a $\pc$ does not adequately 
model cryptographic commitment schemes where no trusted third party is required, we modified
the scheme to better fit the new type system requirements.}

The real power of FLAC is that the security guarantees of well-typed FLAC
functions like those above are compositional.  The FLAC type system ensures the
security of both the functions themselves and the programs that use them.  For
instance, the following code should be rejected because it would permit
$q$ to reveal $p$'s commitments.  
$$\tlam{N}{q^{←}}{\tlam{X}{q^{←}}{\lamc{x}{\says{p∧q^{←}}{X}}{q^{←}}{\assert{\delexp{q}{p}}{"reveal"~N~X~x }}}}$$
The $\pc$ constraints on this function all have the integrity of $q$, but the body of the function 
uses "assume" to create a new delegation from $p$ to $q$. If this "assume" was permitted by the type system, 
then $q$ could use the function to reveal $p$'s committed values.  Since \ruleref{Assume} requires the $\pc$ to speak for
$p$, this function is rejected.

FLAC's guarantees make it possible to state general security properties of all programs
that use the above commitment scheme, even if those programs are malicious.
For example, suppose we have a typing context that includes the "commit", "reveal", and "open"
functions from Figure~\ref{fig:commit}.
\begin{align*}
  Γ_{cro} = \; &"commit" \ty \tfuncpc{N}{p^{←}}{\tfuncpc{X}{p^{←}}{\func{N}{p^{←}}{\func{\says{p^{→}}{X}}{p^{←}}{\says{p}{(N,X)}}}}}, & \\
                  & "reveal" \ty \tfuncpc{N}{∇(p^{→})}{\tfuncpc{X}{q^{←}}{\func{\says{p}{(N,X)}}{q^{←}}{\says{q^{→} ∧ p^{←}}{(N,X)}}}} & \\
                  & "open"   \ty \tfuncpc{N}{q^{←}}{\tfuncpc{X}{q^{←}}{\func{(\tfuncpc{Y}{q^{←}}{\func{\says{p}{(N,Y)}}{q^{←}}{\says{q^{→} ∧ p^{←}}{(N,Y)}}})}{q^{←}}{}}} & 
\end{align*} 
We can consider programs under the control of principal $q$ by considering
source-level FLAC terms that type under $Γ_{cro}$ at $\pc = q^{←}$.  Note that
this category includes potentially malicious programs that attempt to abuse the commitment scheme operations
or otherwise attempt to access committed values.
Since we are interested in properties that hold for all principals $p$ and $q$, 
we want the properties to hold in an empty delegation context: $Π=∅$. Below, we omit
the delegation context altogether for brevity.

FLAC's noninterference guarantee helps rule out information that an attacker can influence or learn.
We instantiate the environment $Γ_{cro}$ with a well-typed substitution $S_{cro}$ that replaces "commit", "reveal", and "open"
with the terms defined in Figure~\ref{fig:commit}.
 We can now show that:
\begin{itemize}
  \item {\bfseries $q$ cannot learn a value that hasn't been revealed.}
    For simplicity, we instantiate the type variables $N$ and $X$ with $τ_N$ and $τ_X$.
    We can represent values of some type $τ_X$ observable to $q$ with the type  $\says{q^{→}}{τ_X}$.
    For any $e$ such that
     \begin{align*}
      & \TVal{Γ_{cro},x\ty \says{p}{(τ_N,τ_X)};q^{←}}{e}{\says{q^{→}}{τ}} &
      \end{align*}                                                                     
    and any committed values $v₁$ and $v₂$ where $\TVal{q^{←}}{v_i}{\says{p}{(τ_{N},τ_X)}}$,
    if  $\subst{e\;S_{cro}}{x}{v₁} \stepsto v₁'$ and  $\subst{e\;S_{cro}}{x}{v₂} \stepsto v₂'$, then 
    by Theorem~\ref{thm:niconf} with $H = p^{→} ∧ q^{←}$ and $ℓ =q^{→} ∧ p^{←}$ $t_{1} ≈^{∅}_{q^{→}} t_{2}$. 
    Note that the same approach can also be used to prove $q$ cannot learn $p$'s uncommitted secrets.
\vspace{1em}
  \item {\bfseries $p$ cannot cause $q$ to open modified value.}
    We represent a value opened by $q$ with the type
     $\says{q^{←}}{\says{q^{→} ∧ p^{←}}{(τ_{N}, τ_X)}}$, which represents
     an open value trusted by $q$. Only $q$ should have the authority to
     accept and open a commitment, so we want to prove that $p$ cannot
     produce a value of this type, even when the "open" operation is in
     scope. For any $e$ such that
     \begin{align*}
       & \TVal{Γ_{cro},x\ty \says{p}{(τ_N,τ_X)};p^{←}}{e}{\says{q^{←}}{(\says{q^{→} ∧ p^{←}}{(τ_{N}, τ_X)})}} &
     \end{align*}                                                                                                         
     and any committed values $v₁$ and $v₂$ where $\TVal{p^{←}}{v_i}{\says{p}{(τ_N,τ_X)}}$,
    if  $\subst{e\;S_{cro}}{x}{v₁} \stepsto v₁'$ and  $\subst{e\;S_{cro}}{x}{v₂} \stepsto v₂'$, then 
    by Theorem~\ref{thm:niinteg} with $H = q^{→} ∧ p^{←}$ and $ℓ =p^{→} ∧ q^{←}$ $t_{1} ≈^{∅}_{q^{←}} t_{2}$. 
\end{itemize}

For these properties we consider programs using our commitment scheme
that $q$ might invoke, hence we consider FLAC programs that type-check in the
$\G_{cro};\pc_q$ context. In the first property, we are concerned with programs
that produce values protected by policy $p$.  Since such programs
produce values with the integrity of $p$ but are invoked by $q$, we want to
ensure that no program exists that enables $q$ to obtain a value with $p$'s
integrity that depends on $x$, which is a value without $p$'s integrity.  The second
property concerns programs that produces values at $ℓ ⊓ q^{→}$ for any ℓ; these
are values readable by $q$. Therefore, we want to ensure that no program exists
that enables $q$ to produce such a value that depends on $x$ or $y$, which are
not readable by $q$.  

Each of these properties hold by a straightforward application of our
noninterference theorems (Theorems~\ref{thm:niconf} and~\ref{thm:niinteg}).  This result is
strengthened by our robust declassification theorem
(Theorem~\ref{thm:robdecl}), which ensures that attacks by $q$ on
$p$'s programs cannot subvert the intended declassifications.

\subsection{Bearer Credentials}
\label{sec:bearer}

We can also use FLAC to implement bearer credentials, our second example of a 
dynamic authorization mechanism.
We represent a bearer credential with authority $k$ in FLAC as a term with the type
$$\tfuncpc{X}{\pc}{\func{\says{\confid{k}}{X}}{\pc}{\says{\integ{k}}{X}}}$$
which we abbreviate as $\spkfor{\confid{k}}{\pc}{\integ{k}}$.  These terms
act as bearer credentials for a principal $k$ since they may be used as a proxy for $k$'s 
confidentiality and integrity authority. Recall that $k^{←} = k^{←} ∧ ⊥^{→}$ and $k^{→} = k^{→} ∧ ⊥^{←}$.
Then secrets protected by $\confid{k}$ can be declassified to $\confid{\bot}$, and untrusted 
data protected by $\integ{\bot}$ can be endorsed to $\integ{k}$. Thus this term
wields the full authority of $k$, and if $\pc=⊥^{←}$, the credential may be used in 
any context—any ``bearer'' may use it.  From such credentials, more restricted credentials can be derived.
For example, the credential $\spkfor{\confid{k}}{\pc}{⊥^{→}}$ grants the bearer authority to declassify
$k$-confidential values, but no authority to endorse values.

It is interesting to note that DCC terms analogous to FLAC terms with type
$\spkfor{\confid{k}}{\pc}{\integ{k}}$ would only be well-typed in DCC if $k$ is
equivalent to $⊥$.  This is because the function takes an argument with $k^{→}$
confidentiality and no integrity, and produces a value with $k^{←}$ integrity
and no confidentiality. Suppose $\L$ is a security lattice used to type-check DCC programs
with suitable encodings for $k$'s confidentiality and integrity.
If a DCC term has a type analogous to $\spkfor{\confid{k}}{}{\integ{k}}$, 
then $\L$ must have the property $k^{→}⊑⊥$ 
and $⊥⊑k^{←}$. This means that $k$ has no confidentiality and no integrity.
That FLAC terms may have this type for any principal $k$ makes it
straightforward to implement bearer credentials and demonstrates a useful
application of FLAC's extra expressiveness.

The $\pc$ of a credential $\spkfor{\confid{k}}{\pc}{\integ{k}}$ acts
as a sort of caveat: it restricts the information flow context in which the
credential may be used. We can add more general caveats to credentials by
wrapping them in lambda terms. To add a caveat $\phi$ to a credential with type
$\spkfor{\confid{k}}{\pc}{\integ{k}}$, we use a wrapper:
$$\lamc{x}{\spkfor{\confid{k}}{\pc}{\integ{k}}}{\pc}{\tlam{X}{\pc}{\lamc{y}{\phi}{\pc}{xX}}}$$
which gives us a term with type
$$\tfuncpc{X}{\pc}{\func{\phi}{\pc}{\func{\says{\confid{k}}{X}}{\pc}{\says{\integ{k}}{X}}}}$$
This requires a term with type $\phi$ (in which $X$ may occur) to be applied before the authority of $k$ can be used.
Similar wrappers allow us to chain multiple caveats; i.e., for caveats $\phi_1 \dots \phi_n$, we obtain the type
$$\tfuncpc{X}{\pc}{\func{\phi_1}{\pc}{\func{\dots}{\pc}{ \func{\phi_n}{\pc}{\func{\says{\confid{k}}{X}}{\pc}{\says{\integ{k}}{X}}}}}}$$
which abbreviates to
$$\spkfor{\confid{k}}{\phi_1\times\dots\times\phi_n;\pc}{\integ{k}}$$
We will also use the syntax $(\spkfor{\confid{k}}{\phi_1\times\dots\times\phi_n;\pc}{\integ{k}})~τ$ (suggesting a type-level application)
as an abbreviation for 
$$\func{\phi_1[X↦τ]}{\pc}{\func{\dots}{\pc}{ \func{\phi_n[X↦τ]}{\pc}{\func{\says{\confid{k}}{τ}}{\pc}{\says{\integ{k}}{τ}}}}}$$

Like any other FLAC terms, credentials may be protected by information flow
policies.  So a credential that should only be accessible to Alice might be
protected by the type
$\says{\Alice^{→}}{(\spkfor{\confid{k}}{\phi;\pc}{\integ{k}})}$.  This
confidentiality policy ensures the credential cannot accidentally be leaked to
an attacker.  A further step might be to constrain uses of this credential so
that only Alice may invoke it to relabel information.  If we require
$\pc=\Alice^{←}$, this credential may only be used in contexts trusted by Alice:
$\says{\Alice^{→}}{(\spkfor{\confid{k}}{\phi;\Alice^{←}}{\integ{k}})}$.  

A subtle point about the way in which we construct caveats is that the caveats
are polymorphic with respect to $X$, the same type variable the credential
ranges over.  This means that each caveat may constrain what types $X$ may be
instantiated with. For instance, suppose we want to encode a 
relation "isEduc" for specifying movie topics that are educational.  
One possible encoding is a polymorphic function of type
$$"isEduc" : \tfuncpc{X}{p}{\says{p}{(\func{X}{p}{X}, U)}}$$
where $p$ is the principal with movie classification authority, and $U$ is a unique
type-level integer (such as those described in Section~\ref{sec:commit}) 
associated with the "isEduc" relation (and no other). Values of type $\says{p}{(\func{τ}{p}{τ}, U)}$  
can be used as evidence that type $τ$ belongs to "isEduc", the relation associated with $U$.\footnote{Using 
an instantiated identity function type $\func{τ}{p}{τ}$ instead of
just $τ$ avoids having to produce a value of type $τ$ just to define
the relation. The $\pc$ label chosen for the identity function is
irrelevant.}

Only code that is at least as trusted as $\pc$ may apply this
function, therefore only authorized code may add types to the "isEduc"
relation. We might, for instance, apply $"isEduc"$ to types like
$"Biography"$ and $"Documentary"$, but not $"RomanticComedy"$.  Adding
$\says{\pc}{(\func{X}{\pc}{X}, U)}$
 as a caveat to a credential would mean that the bearer
of the credential could use the credential plus evidence of membership
in $"isEduc"$ to access biographies and documentaries. Since no term
of type
$\says{\pc}{(\func{"RomanticComedy"}{\pc}{"RomanticComedy"}, U)}$
exists (nor could be created by
the bearer), the bearer can only satisfy the caveat by instantiating
$X$ with "Biography" or "Documentary". Once $X$ is instantiated with
"Biography" or "Documentary", the credential cannot be used on a
"RomanticComedy" value.  Thus we have two mechanisms for constraining
the use of credentials: information flow policies to constrain
propagation, and caveats to establish prerequisites and constrain the
types of data covered by the credential.

As another example of using such credentials, suppose Alice hosts a file sharing
service. 
For a simpler presentation, we use free
variables
to refer to these files; for instance, $x₁\ty(\says{k₁}{"photo"})$ is a
variable that stores a photo (type "photo") protected by $k₁$.  For each such
variable $x₁$, Alice has a credential $\spkfor{k₁^{→}}{⊥^{←}}{k₁^{←}}$, and can
give access to users by providing this credential or deriving a more restricted one.
To access $x₁$, Bob does not need the full authority of Alice or $k₁$—a more restricted 
credential suffices.  Alice can provide Bob with a credential $c$ of type 
$(\spkfor{k₁}{\Bob^{←}}{k₁^{←}})\;"photo"$.  By applying this credential to $x₁$, 
Bob is able to access the result of type $\says{k₁^{←}}{"photo"}$ since its confidentiality is now public.

This example demonstrates an advantage of bearer credentials: access to $x₁$ can be provided to
principals other than $k₁$ in a decentralized way, without changing the policy
on $x₁$. Suppose Alice  has a credential with type
$\spkfor{k₁^{→}}{⊥^{←}}{k₁^{←}}$ and wants to issue the above credential to Bob. 
Alice can create such a credential using a wrapper that derives the more
constrained credential from her original one.
\begin{align*}
& \lamc{c}{\spkfor{k₁^{→}}{⊥^{←}}{k₁^{←}}}{\Alice^{←}}{} &\\
&\quad \lamc{y}{\says{k₁}{"photo"}}{\Bob^{←}}{}& \\
& \quad \quad \bind{y'}{y}{(c\;"photo")~(\return{k₁^{→}}{y'})}
\end{align*}
This function has type
$\func{(\spkfor{k₁^{→}}{⊥^{←}}{k₁^{←}})}{\Alice^{←}}{(\spkfor{k₁}{\Bob^{←}}{k₁^{←})~"photo"}}$:
given her root credential $\spkfor{k₁^{→}}{⊥^{←}}{k₁^{←}}$, Alice (or
someone she trusts) can create a restricted credential that allows Bob (or someone he trusts)
to access values of type "photo" protected under $k₁$.

Bob can also use this credential to share photos with friends.  For instance, the function
\begin{align*}
& \lamc{c}{(\spkfor{k₁}{\Bob^{←}}{k₁^{←}})~"photo"}{\Bob^{←}}{} &\\
 & \quad \assume{\delexp{\Carol^{←}}{\Bob^{←}}}{} &\\
 &\qquad \lamc{y}{\says{k₁}{"photo"}}{\Carol^{←}}{}& \\
& \qquad \quad \bind{y'}{y}{(c\;"photo")~(\return{k₁^{→}}{y'})}
\end{align*}
creates a wrapper around a Bob's credential that is invokable by Carol.  Using the "assume" term, 
Bob delegates authority to Carol so that his credential may be invoked at pc $\Carol^{←}$.

The properties of FLAC let us prove many general properties about such
bearer-credential programs. We first define the context we are interested in.
We can model the premise that only Alice has access to credentials by 
protecting the secrets and the credentials in the typing context:
\begin{align*}
 & Γ_{bc}=xᵢ\ty\says{kᵢ}{τᵢ},c_{i}\ty \says{kᵢ}{(\spkfor{kᵢ^{→}}{⊥^{←}}{kᵢ^{←}})} & 
\end{align*}
and delegating authority from principal $kᵢ$ to Alice in the delegation context: 
\begin{align*}
 & Π_{bc}=\delexp{\Alice}{kᵢ}&
\end{align*}
where $k_{i}$ is a primitive principal protecting the $i^{th}$ resource of type
$τ_{i}$, and $cᵢ$ is a credential for the $i^{th}$ resource and protected by
principal $kᵢ$, which delegates authority to Alice.
\begin{itemize}
  \item {\bfseries $p$ cannot access resources without a credential.}
    A concise way of representing values that are observable to $p$ (with any integrity) is
    the principal $p^{→}$, which is the most restrictive policy observable to $p$.  Since
    any less restrictive policy observable to $p$ flows to this policy, any output observable to $p$
    is protected at this type.  Suppose $e$ is a program that computes such an output.
    For any $e$ such that $\TVal{Γ_{bc};p^{←}}{e}{\says{p^{→}}{τ'}}$, for some $τ'$,
    we can show that the evaluation of $e$ is independent of the secret variables $xᵢ$.
    To see how, fix some index $n$ and define
    \begin{align*}
      & S_{bc} = [x_{j} ↦\returng{k_{j}}{v_{j}}][c_{i} ↦(\returng{kᵢ}{...})] &
    \end{align*}                                                                                                                                        
    for all $i$ and all $j$ where $j≠n$ and $\TVal{p^{←}}{v_{j}}{\says{k_{j}^{→}}{τ_{j}}}$.
    Additionally, suppose we have a two values $v¹_{n}$ and $v²_{n}$ such that 
    $\TVal{p^{←}}{v^{i}_{n}}{\says{k_{n}}{τ_{n}}}$ for $i ∈ \{1, 2\}$. 
    Then if $\subst{S_{bc}\;e}{x_{n}}{v¹_{n}} \stepstot{t₁} w¹_{n}$ and $\subst{S_{bc}\;e}{x_{n}}{v²_{n}} \stepstot{t₂} w²_{n}$,
    by Theorem~\ref{thm:niconf} with $H = k^{→}_{n} ∧ p^{←}$ (and $ℓ^{→} = p^{→}$), we have $t_{1} ≈^{Π_{bc}}_{p^{→}} t_{2}$.
    Since we left $n$ abstract, we can repeat this argument for all $xᵢ$.
\vspace{1em}

    \item {\bfseries Alice cannot accidentally disclose secrets by issuing credentials.}
    Suppose $e$ is a program that outputs a credential observable to $p$ that declassifies secrets protected at $k_{n}^{→}$.
    For any $e$ such that $\TVal{Γ_{bc},y\ty \says{\Alice}{τ}; k_{n}^{←}}{e}{\says{p^{→}}{\spkfor{k_{n}^{→}}{⊥^{←}}{p_{n}^{→}}}$, for some $τ$,
    we can show that the evaluation of $e$ does not depend on Alice's secret variable $y$.
    Define $S_{bc}$ as above and consider values $v¹$ and $v²$ such that $\TVal{p^{←}}{v^{i}}{\says{\Alice}{τ'}}$ for $i ∈ \{1, 2\}$. 
    Then if $\subst{S_{bc}\;e}{y}{v¹} \stepstot{t₁} w¹$ and $\subst{S_{bc}\;e}{y}}{v²} \stepstot{t₂} w²$,
    by Theorem~\ref{thm:niconf} with $H = \Alice^{→} ∧ p^{←}$ (and $ℓ^{→} = p^{→}$), we have $t_{1} ≈^{Π_{bc}}_{p^{→}} t_{2}$.
\end{itemize} 

These properties demonstrate the some of the power of FLAC's type system.  The
first ensures that credentials really are necessary for $p$ to access
protected resources, even indirectly.  If $p$ has no credentials, and
the type system ensures that $p$ cannot invoke a program that produces a value
$p$ can read (represented by $ℓ ⊓ p^{→}$) that depends on the value of any variable $xᵢ$.
The second property eliminates covert channels like the one discussed in Section~\ref{sec:ex2}. 
It implies that credentials issued by Alice do not leak information.
By implementing and using bearer credentials in FLAC, we can demonstrate these properties with relatively little effort
by appealing to Theorem~\ref{thm:niconf}.

\section{Related Work}
\label{sec:relwork}

Arden, Liu, and Myers~\cite{flam} introduced the Flow-Limited
Authorization Model for reasoning about trust relationships that may
be secret or untrustworthy.  FLAC's type systems and semantic results
rely on a fragment of this logic that reasons from the perspective of
the compiler.  In this fragment, all FLAM queries occur at compile
time, and the compiler does not depend on any (secret or public)
dynamic data to answer queries.  Furthermore, the compiler does
not communicate with other principals during derivation.  Our static
lattice rules (Figure~\ref{fig:static}) provide an alternate
axiomatization of the FLAM principal algebra, but we have proven them
equivalent in Coq for the ownership-free fragment of FLAM.

The original FLAC formalization~\cite{flac} contained several errors
both in the language definition and the proofs.  Some of the design
changes we made for the current formalization were to correct errors.
For example, we added a (necessary) $\pc$ annotation to type functions similar
to that for lambda values, and added missing "where"-propagation rules,
along with a progress result, Lemma~\ref{lemma:prog}.
Other changes were inspired by subsequent
work like DFLATE~\cite{dflate} and NMIFC~\cite{nmifc}.  For example,
protection contexts (Figure~\ref{fig:octx}) were inspired by the "TEE"
abstraction from DFLATE to properly handle the erasure of intermediate
expressions in the trace semantics, and both NMIFC and DFLATE first
adopted the more restrictive \ruleref{UnitM} and protection rules that
make the "says" modality non-commutative.  

The original noninterference and robust declassification results
were on output terms only, similar to the Pottier and
Simonet~\cite{ps03} noninterference result.  Our formalization
leverages a trace semantics and erasure function, similar to the
approach in DFLATE.  The new noninterference and robustness theorems
are thus stronger than the original statements since attackers not
only see the output of a program but also its (observable)
intermediate values.  We conjecture that most of the design changes
would be necessary to repair the proofs of the output-only versions of
our theorems (with suitable updates reflecting changes to the
Delegation Invariance statement), but we have not attempted to
distinguish which (if any) changes were necessary only for the
stronger results.

NMIFC and DFLATE each handle downgrading differently than FLAC.  NMIFC
does not permit arbitrary delegation via an "assume" term, but instead
uses "declassify" and "endorse" for downgrading confidentiality and
integrity explicitly.  DFLATE treats delegation contexts and "where" terms subtly differently
than FLAC.  First, DFLATE functions are explicitly annotated with the 
delegations they capture from the delegation context, and may only be applied 
in contexts where those delegations are valid.  In contrast, FLAC functions
implicitly capture delegations and two functions of the same type
may capture different delegations.
Second, rather than being purely formal bookkeeping mechanisms,
"where" terms serve as certificates of delegated authority at runtime and are
propagated between hosts.  One implication is that "where" delegations 
are directly observable, whereas FLAC's erasure function omits them.  

Flamio~\cite{flamio} is a course-grained, dynamic IFC language similar
to LIO~\cite{lio}, but using the FLAM distributed authorization logic
for dynamic label comparisons.  FLAC's type system also uses the FLAM
logic, but only a static, local fragment, which significantly
simplifies the information flows that occur due to authorization
checks.

Many languages and systems for authorization or access control have combined aspects
of information security and
authorization (e.g.,~\cite{wl04,hthz05,mk05,shtz06,mk06,RTI}) in dynamic settings.
However, almost all are susceptible to security vulnerabilities that
arise from the interaction of information flow and authorization~\cite{flam}.

DCC~\cite{ccd99,abadi06} has been used to model both authorization and
information flow, but not simultaneously.  DCC
programs are type-checked with respect to a static security lattice,
whereas FLAC programs can introduce new trust relationships during
evaluation, enabling more general applications.

Boudol~\cite{boudol08} defines an imperative language with terms that
enable or disable flows for a lexical scope—similar to "assume"
terms—but does not restrict their usage.  Rx \cite{shtz06} and
RTI~\cite{RTI} use labeled roles to represent information flow
policies. The integrity of a role restricts who may change policies.
However, information flow in these languages is not
robust~\cite{msz06}: attackers may indirectly affect how flows change
when authorized principals modify policies.

Some prior approaches have sought to reason about the information security of 
authorization mechanisms. Becker~\cite{becker2012information}
discusses _probing attacks_ that leak confidential information 
to an attacker. Garg and Pfenning~\cite{garg2006non} present a logic 
that ensures assertions made by untrusted
principals cannot influence the truth of statements made by other principals.

Previous work has studied information flow control with higher-order functions
and side effects.  In the SLam calculus~\cite{hr98}, implicit flows due to
side effects are controlled via _indirect reader_ annotations on types.
Zdancewic and Myers~\cite{zm02-hosc} and Flow Caml~\cite{ps03} control implicit
flows via $\pc$ annotations on function types. \lang also controls side effects
via a $\pc$ annotation, but here the side effects are changes in trust
relationships that define which flows are permitted.
Tse and Zdancewic~\cite{tz04} also extend DCC with a program-counter label but
for a different purpose:  their $\pc$ tracks information about the
protection context, permitting more terms to be typed.

DKAL$^{⋆}$~\cite{dkalstar} is an executable specification language for authorization
protocols, simplifying analysis of protocol implementations. FLAC
may be used as a specification language, but FLAC offers stronger
guarantees regarding the information security of specified protocols.
Errors in DKAL$^{⋆}$ specifications could lead to vulnerabilities.
For instance, DKAL$^{⋆}$ provides no intrinsic guarantees about confidentiality,
which could lead to authorization side channels or probing 
attacks.

The Jif programming language~\cite{myers-popl99,jif} supports 
dynamically computed labels through a simple dependent type system.
Jif also supports dynamically changing trust relationships through operations
on principal objects~\cite{sif07}. Because the signatures of principal operations
(e.g., to add a new trust relationship) are missing the constraints
imposed by FLAC, authorization can be used as a covert channel.

Aura~\cite{aura} embeds a DCC-based proof language and type system in a
dependently-typed general-purpose functional language.  As in DCC, Aura programs
may derive new authorization proofs using existing proof terms and a monadic
bind operator.  However, since Aura only tracks dependencies between proofs, it
is ill-suited for reasoning about the end-to-end information-flow properties of
authorization mechanisms.  
In general, dependently-typed languages (e.g., \cite{fstar, Coq:manual, agda, idris}) are
also expressive enough to encode relations and constraints like
those used in FLAC's type system, but all FLAC programs have the semantic
security guarantees presented in Section~\ref{sec:props} by construction.

\section{Discussion and Future Directions}
\label{sec:concl}

\OA{Reviewer 2:
No future work? No extension of the language would be interesting?
No limitations of the current work? More discussion is expected.
}
\OA{explore NMIFC (discuss limitations)?  further investigation of compartmentalization?
  quantification over principals?
  extensions to dependently-typed languages via FLAFOL?
 } 

Existing security models do not account fully for the interactions
between  authorization and information
flow. The result is that both the implementations and the uses of
authorization mechanisms can lead to insecure information flows that
violate confidentiality or integrity. The security of information flow
mechanisms can also be compromised by dynamic changes in trust.  This
paper has proposed FLAC, a core programming language that coherently
integrates these two security paradigms, controlling the
interactions between dynamic authorization and secure information
flow. FLAC offers strong guarantees and can serve as the foundation
for building software that implements and uses authorization securely.
Further, FLAC can be used to
reason compositionally about secure authorization and secure
information flow, guiding the design and implementation of future
security mechanisms.

We have already mentioned subsequent work like DFLATE~\cite{dflate},
which extends FLAC concepts to distributed TEE applications, and
NMIFC~\cite{nmifc}, which enforces nonmalleable information flow
control, a new downgrading semantic condition that generalizes robust
declassification to include integrity downgrading.  To simplify its
formalization, NMIFC removed the "assume" term from FLAC. Another
promising direction is extending FLAC to incorporate abstractions for
new mechanisms such as secure quorum replication, secure multi-party
computation, and other cryptographic mechanisms as security
abstractions in the language. Abstractions for secure mobile code
sharing, such as that supported by Mobile Fabric~\cite{oakland12},
would also be interesting since it would require considering the
provider of code in addition to its information flow properties.

Features such as quantification over principals and dynamic
principal values, as found in Jif~\cite{myers-popl99, zm04, jif} would
increase the usefulness of specifying security policies with FLAC
types.  Initial explorations of these features have been implemented
in Flame~\cite{flame}, an embedded Haskell DSL based on FLAC.

\ifblinded\relax
\else
\section*{Acknowledgments} 
We thank Mike George, Elaine Shi, and
Fred Schneider for helpful discussions, our anonymous reviewers for their 
comments and suggestions, and Jed Liu, Matt Stillerman, and Priyanka Mondal for feedback on 
early drafts and bugfixes.
We also wish to thank our anonymous reviewers, whose thorough and thoughtful feedback significantly
improved the clarity and rigor of our results.
This work was supported by grant
N00014-13-1-0089 from the Office of Naval Research, by MURI grant
FA9550-12-1-0400, by grant FA9550-16-1-0351 from
the Air Force Office of Scientific Research,  
by an NDSEG fellowship, %
by the Intelligence Advanced Research Projects Activity (2019-19-020700006), %
and by grants from the National Science Foundation
(CCF-0964409, CNS-1565387, CNS-1704845, CCF-1750060). %
\fi

\bibliographystyle{abbrvnat}
\bibliography{bibtex/pm-master,pm-master}

\appendices

\section{Proofs for Noninterference and Robust Declassification}
\label{app:proofs}

\subsection{FLAM and FLAC}
We formalize the relation between FLAM and FLAC.

\begin{lemma}[FLAC implies FLAM]
	Let $\mathcal{H}(c) = \{ ⟨p ≽ q ~|~ ⊥^{→}∧⊤^{←} ⟩ \sep ⟨p ≽ q⟩ ∈ Π \}$. If $\rafjudge{\delegcontext}{p}{q}$, then
	$\rafjudge{\mathcal{H};c;⊥^{→}∧⊤^{←};⊥^{→}∧⊤^{←}}{p}{q}$.
\end{lemma}
\begin{proof}
	Proof is by induction on the derivation of the robust assumption
	$\rafjudge{\delegcontext}{p}{q}$. Interesting case is \ruleref{R-Assume}. 
	\begin{description}
		\item[Case \ruleref{R-Assume}:] From the premises, we have that
	\begin{align}
		\delexp{p}{q} \in \delegcontext \label{eq:1}\\
		\rafjudge{\delegcontext;\pc; \ell}{\voice{\confid{p}}}{\voice{\confid{q}}} \label{eq:4}
	\end{align}

	From \eqref{eq:1} and \ruleref{Del}, we have that
	\[
		\afjudge{\mathcal{H};c;⊥^{→}∧⊤^{←};⊥^{→}∧⊤^{←}}{p}{q}
	\]

	From \ruleref[Weaken], we get that
	$\afjudge{\mathcal{H};c;  \wedge \voice{q})}{p}{q}$.
	From the \eqref{eq:4} and \ruleref{R-Lift} we thus have
	$\rafjudge{\mathcal{H};c;⊥^{→}∧⊤^{←};⊥^{→}∧⊤^{←}}{p}{q}$.

\item[Case R-Static:] Since $\stafjudge*{p}{q}$, we have from
	FLAM \ruleref{R-Static} that
	$\rafjudge{\mathcal{H};c;⊥^{→}∧⊤^{←};⊥^{→}∧⊤^{←}}{p}{q}$.

	\item[Case R-ConjR:] Trivial.
	\item[Case R-DisjL:] Trivial.
	\item[Case R-Trans:] Trivial.
	\item[Case R-Weaken:] Trivial.
	\end{description}
	
\end{proof}

\begin{lemma}[FLAM implies FLAC]
        For a trust configuration such that, for all $n≠c$, $\mathcal{H}(n)=∅$, and for all $⟨p,q,ℓ⟩ ∈ \mathcal{H}(c)$, $ℓ=⊥^{→}∧⊤^{←}$ and 
        $\rafjudge{\mathcal{H};c;⊥^{→}∧⊤^{←}; ⊥^{→}∧⊤^{←}}{∇(p→)}{∇(q→)}$,
        if $\rafjudge{\mathcal{H};c;⊥^{→}∧⊤^{←}; ⊥^{→}∧⊤^{←}}{p}{q}$, then $\rafjudge{\delegcontext}{p}{q}$.
\end{lemma}
\begin{proof}[Proof Sketch] By induction on the FLAM derivation. 
 Without loss of generality, we assume that the derivation of
$\rafjudge{\mathcal{H};c;⊥^{→}∧⊤^{←}; ⊤^{→}∧⊥^{←}}{p}{q}$ contains no
applications of \textsc{R-Weaken} or \textsc{R-Fwd}. Since all
delegations are local to $c$, public, and trusted, they are
unnecessary.

Next, observe that because delegations are public and trusted, any
non-robust FLAM derivation may be lifted to a robust one by adding an
application of FLAM's \textsc{R-Lift} rule wherever a \textsc{Del}
rule occurs, and applying relevant robust rules in place of the
non-robust rules.  Rules corresponding to each non-robust rule are
either part of the core robust rules or have been proven
admissible~\cite{flamtr}.  The rule for \textsc{R-Trans} is an
exception, since it requires an additional premise to be satisfied,
but since the query label is always $⊥^{→}∧⊤^{←}$ this is trivially
satisfied.
\end{proof}

\subsection{Proofs for Type Preservation (Lemma~\ref{lemma:subjred})}\label{app:subjred}
Before proving type preservation, we define the grammar for sub terms that is used in proving the monotonicity of the program counter, and prove few supporting lemmas.

\begin{figure}[t]
  \[
\begin{array}{rcl}
  T & ::= &  [\cdot] \sep T~e \sep e~T  \sep T~τ  \sep \pair{T}{e} \sep \pair{e}{T} \sep \proji{T} \sep \inji{T} \sep \return{ℓ}{T} \\[0.4em]
  & \sep &  \bind{x}{T}{e} \sep \bind{x}{e}{T} \sep \assume{T}{e} \sep \assume{e}{T}  \\[0.4em]
  & \sep  & \casexp{T}{x}{e}{e} \sep \casexp{e}{x}{T}{e}   \\[0.4em]
  & \sep & \casexp{e}{x}{e}{T} \sep \where{T}{v} \sep \where{e}{T} \\[0.4em]
 U &::=&  \pair{U}{w} \sep \pair{w}{U}\sep \returng{ℓ}{U} \sep \inji{U}  \\[0.4em]
  &&  \sep \lamc{x}{τ}{\pc}{U}  \sep \tlam{X}{\pc}{U} \\[0.4em]  
\end{array}
\]
\caption{Subterm language for expressions and values}
\label{fig:subterm}
\end{figure}

Every evaluation context can be represented as a subterm context but not the converse.
The following lemma is useful to convert evaluation contexts into subterm contexts.

\begin{lemma}[Evaluation Context implies Subterm Context]\label{lemma:evalimpliessub}
  For all $E$ and $e$, if \TValGpc{E[e]}{\tau} then exists $T$ such that  $T = E$ and  $\TValGpc{E[e]}{\tau}$.
\end{lemma}
\begin{proof}
 Induction on the structure of $E$.
\end{proof}

\begin{lemma}[Robust Assumption] \label{lemma:robext}
	If
	$\rafjudge{\delegcontext}{\pc}{\voice{b}}$, then
	$\rafjudge{\delegcontext\delegconcat \delexp{a}{b}}{\pc}{\voice{b}}$ for any $a, b \in \L$.
\end{lemma}
\begin{proof} By inspection of the rules in Figure~\ref{fig:robrules}. \end{proof}

\begin{lemma}[Robust Protection] \label{lemma:protext}
	If
	$\protjudge{\delegcontext}{\pc}{\tau}$,  
	then
	$\protjudge{\delegcontext\delegconcat \delexp{a}{b}}{\pc}{\tau}$ for any $a, b \in \L$. 
\end{lemma}
\begin{proof} By Lemma~\ref{lemma:robext} and inspection of the rules in Figure~\ref{fig:protect}. \end{proof}

Monotnicity of $\pc$ is standard in many information flow control type systems. FLAC, too, has one.

\begin{lemma}[Monotonicity of $\pc$]\label{lemma:pcmonotone}
  Let \TValGpc{e}{\tau}. If $e = T[e']$ such that for some $\Pi', \varcontext', \pc'$, \TVal{\delegcontext';\varcontext';\pc'}{e'}{\tau'}, then  \rflowjudge{\Pi}{\pc}{\pc'}.
  \end{lemma}
\begin{proof}
  Induction on the structure of the $T$. The interesting case is when $T = \bind{x}{e_1}{T'}$. Consider the typing of $T[e']$; that is  \TValGpc{T[e']}{\tau}. From the typing rule \ruleref{BindM}, we have that
  \TVal{\Pi';\G, x\ty \tau'; \pc \sqcup \ell}{e'}{\tau} (assuming  \TValGpc{e_1}{\says{\ell}{\tau'}} for some $\ell$). Here $\pc' = \pc \sqcup \ell$ and so \rflowjudge{\Pi}{\pc}{\pc'}.
\end{proof}

We need few helper lemmas to state the properties of delegation contexts.
\begin{lemma} \label{lemma:voicetrans}
  If \rafjudge{\Pi}{q}{\voice{t}} and \rflowjudge{\Pi}{p}{q} then \rafjudge{\Pi}{p}{\voice{t}}.
\end{lemma}
\begin{proof}
  Follows from robust acts-for inference rules.
\end{proof}

\begin{lemma}[\delegcontext Extension]\label{lemma:dpiextension}
	If {\TValGpc{e}{\tau}} then \TVal{\delegcontext\delegconcat \delexp{p}{q};\varcontext;\pc}{e}{\tau} for any $p, q \in \L$.
\end{lemma}
\begin{proof} By Lemma~\ref{lemma:robext}, Lemma~\ref{lemma:protext} and inspection of the typing rules.\end{proof}

\begin{lemma}[Monotonicity of Delegation Context]\label{lemma:pimonotone}
	If \TValGpc{e}{\tau} then for all $T, e'$ such that $T[e'] = e$,  \TVal{\delegcontext';\varcontext';\pc'}{e'}{\tau'} such that $\Pi \subseteq \Pi'$.
\end{lemma}
\begin{proof}[Proof Sketch]
The only sub terms that change the delegation context are "assume" and "where". That is, $T = \where{[\cdot]}{v}$ or $T = \assume{e}{[\cdot]}$. However, they add delegations. Thus, $\Pi \subseteq \Pi'$.
\end{proof}

The following lemma is required to prove the type preservation. It says that an  expression $e$ well-typed at $\pc$ is still well-typed at a reduced $\pc'$.

\begin{lemma}[PC Reduction]\label{lemma:pcred}
  Let \TValGpc{e}{\tau}.
  For all $\pc, \pc'$, such that 
  \rflowjudge{\delegcontext}{\pc'}{\pc}  then
  \TValP{\varcontext;pc'}{e}{\tau} holds.
\end{lemma}
\begin{proof} 
  Proof is by induction on the derivation of the typing judgment.
  \begin{description}
  \item[Case \ruleref{Var}: ] Straightforward from the corresponding typing judgment.
  \item[Case \ruleref{Unit}: ] Straightforward from the corresponding typing judgment.
  \item[Case \ruleref{Del}: ] Straightforward from the corresponding typing judgment.
  \item[Case \ruleref{Lam}: ] Straightforward from the corresponding typing judgment.
  \item[Case \ruleref{App}: ] 
                                         Given,
                         ${\TValGpc{e~e'}{τ}}$.
         From \ruleref{App}, we have
         \begin{align}
         \TValGpc{e}{\func{τ_1}{\pc''}{τ_2}} \\
	 \TValGpc{e'}{τ_1} \\
	 \rflowjudge{\delegcontext}{\pc}{\pc''}
         \end{align}
         Applying induction to the premises we have
                  \TValP{\varcontext;\pc'}{e}{\func{τ_1}{\pc''}{τ_2}} and
                  	 \TValP{\varcontext;\pc'}{e'}{τ_1}.
                   From \ruleref{R-Trans}, we have \rflowjudge{\delegcontext}{\pc'}{\pc''}.
                   Hence we have all the premises.
                   
  \item[Case \ruleref{TLam}: ] Straightforward from the corresponding typing judgment.
  \item[Case \ruleref{TApp}: ] Similar to App case.
  \item[Case \ruleref{Pair}: ] Straightforward from the corresponding typing judgment.
  \item[Case \ruleref{UnPair}: ] Straightforward from the corresponding typing judgment.
  \item[Case \ruleref{Inj}: ] Straightforward from the corresponding typing judgment.
  \item[Case \ruleref{Case}: ] Straightforward from the corresponding typing judgment.
  \item[Case \ruleref{UnitM}:] Given ${\TValGpc{\return{ℓ}{e}}{τ}}$, by \ruleref{UnitM} we have \rflowjudge{\delegcontext}{\pc}{ℓ} 
    and \TValP{\varcontext;\pc}{e}{τ}. By the induction hypothesis, we have \TValP{\varcontext;\pc'}{e}{τ}, and since 
\rflowjudge{\delegcontext}{\pc'}{\pc}, then by \ruleref{R-Trans}, we have \rflowjudge{\delegcontext}{\pc'}{ℓ}. Therefore by \ruleref{UnitM} we have {\TValP{\G;\pc'}{\return{ℓ}{e}}{τ}}.

  \item[Case \ruleref{Sealed}:] Straightforward from the corresponding typing judgment.
  \item[Case \ruleref{BindM}:] Given ${\TValGpc{\bind{x}{e}{e'}}{τ}}$, by \ruleref{BindM} we have 
     \begin{align}
        \TValP{\varcontext;\pc}{e}{\says{ℓ}{τ'}} \label{eq:pcredbind1} \\
        \TValP{\varcontext,x:τ';\pc ⊔ ℓ}{e'}{τ} \label{eq:pcredbind2}\\
        \protjudge*{\pc ⊔ ℓ}{τ}
      \end{align}
      Since \rflowjudge{\delegcontext}{\pc'}{\pc}. By the monotonicity of join with respect to $⊑$, we also have \rflowjudge{\delegcontext}{\pc' ⊔ ℓ}{\pc ⊔ ℓ}. 
      Therefore, by the induction hypothesis applied to \ref{eq:pcredbind1} and \ref{eq:pcredbind2}, we have 
     \begin{align}
        \TValP{\varcontext;\pc'}{e}{\says{ℓ}{τ'}} \\
        \TValP{\varcontext,x:τ';\pc' ⊔ ℓ}{e'}{τ}
      \end{align}
      and by \ruleref{R-Trans} we get \protjudge*{\pc ⊔ ℓ}{τ}.  Then via \ruleref{BindM} we get 
         $$\TValP{\G;\pc'}{\bind{x}{e}{e'}}{τ}$$

   \item[Case \ruleref{Assume}:] Given, {\TValGpc{\assume{e}{e'}}{τ}}, by \ruleref{Assume} we have  
         \begin{align}
                 \TValGpc{e}{\aftype{p}{q}} \label{eq:pcredassume1}\\
                 \rafjudge{\delegcontext}{\pc}{\voice{q}} \label{eq:pcredassume2}\\
                 \rafjudge{\delegcontext}{\voice{\confid{p}}}{\voice{\confid{q}}} \label{eq:pcredassume3}\\
                 \TVal{Π,\langle \aftypep{p}{q}\rangle; Γ;\pc}{e'}{τ} \label{eq:pcredassume5}
         \end{align}

         Applying induction hypothesis to \eqref{eq:pcredassume1} and \eqref{eq:pcredassume5} we have 
         \TValP{\G;\pc'}{e}{\aftype{p}{q}} and
                 \TVal{Π,\langle \aftypep{p}{q}\rangle; Γ;\pc'}{e'}{τ}.
                 Since $\rflowjudge{\delegcontext}{\pc'}{\pc}$, we have \rafjudge{\delegcontext}{\pc'}{\voice{q}}.
                 Combining, we have all the premises for \ruleref{Assume} and thus
                 \TValP{\G;\pc'}{\assume{e}{e'}}{τ}
                 
  \item[Case \ruleref{Where}:]
			Given, \TValGpc{\where{v}{e}}{τ}, by \ruleref{Where} we have 
	\begin{align}
		\TValGpc{v}{\aftype{p}{q}} \label{eq:pcredwhere1}\\
		\rafjudge{\delegcontext}{\pcmost}{\voice{q}} \label{eq:pcredwhere2}\\
		\rafjudge{\delegcontext}{\voice{\confid{p}}}{\voice{\confid{q}}} \label{eq:pcredwhere3}\\
		\TValGpc{e}{\tau} \label{eq:pcredwhere5}
	\end{align}
	
	Applying induction hypothesis to \eqref{eq:pcredwhere1} and \eqref{eq:pcredwhere5}, we have 
	\TValP{\G;\pc'}{v}{\aftype{p}{q}} and 
	\TVal{\delegcontext, \delexp{p}{q};\varcontext;\pc'}{e}{τ}. Then by \ruleref{Where},
        we have \TValP{\G;\pc'}{\where{v}{e}}{τ}.

   \item[Case \ruleref{Bracket}:] The premise \rflowjudge{\delegcontext}{H^\pi \sqcup \pc^\pi}{{\pc''}^\pi} implies \rflowjudge{\delegcontext}{H^\pi \sqcup \pc'^\pi}{{\pc''}^\pi}. The result follows from \ruleref{Bracket}.
                \item[Case \ruleref{Bracket-Values}:] Applying induction to the premises gives the required conclusion.
  \end{description}
\end{proof}

\begin{lemma}[Values PC]\label{lemma:wvaluespc}
If \TValGpc{w}{\tau}, then \TValP{\G;\pc'}{w}{τ} for any $\pc'$.
\end{lemma}
\begin{proof}
  By induction on the typing derivation of $w$. Observe that only
\ruleref{App}, \ruleref{Case}, \ruleref{UnitM}, \ruleref{BindM}, and
\ruleref{Assume} contain premises that constrain typing based on the
judgment \pc, and these rules do not apply to $w$ terms.
\end{proof}

We now prove a bunch of substitution lemmas. These are necessary whenever a program variable or a type variable is substituted.

\begin{lemma}[Variable Substitution]\label{lemma:vsubst}
  If \TValP{\G,x:τ';\pc}{e}{\tau} and \TValGpc{w}{τ'}, then \TValGpc{e[x \mapsto w]}{\tau}. 
\end{lemma}
\begin{proof}
  Proof is by induction on the typing derivation of $e$. Observe that by Lemma~\ref{lemma:wvaluespc} and Lemma~\ref{lemma:dpiextension}, we have \TVal{\Pi';\G;\pc'}{w}{τ'} for any $\pc'$ and $\Pi'$ such that $\Pi \subseteq \Pi'$. 
  Therefore, each inductive case follows by straightforward application of the induction hypothesis.
\end{proof}

\begin{lemma}[Variable Substitution Under Contexts]\label{lemma:ctxtvarsubst}
  If \TValP{\G,x:τ';\pc}{\octx{e}{\pc'}}{\tau} and \TValGpc{w}{τ'}, then \TValGpc{\octx{e[x \mapsto w]}{\pc'}}{\tau}. \end{lemma}
\begin{proof}
  Follows from Lemma~\ref{lemma:vsubst}.
\end{proof}

\begin{lemma}[Type Substitution]\label{lemma:tsubst}
	Let $\tau'$ be well-formed in $\varcontext, X, \varcontext'$.
	If $\TValP{\varcontext, X, \varcontext';\pc}{e}{\tau}$ then $\TValP{\varcontext, \varcontext'[X \mapsto \tau'];\pc}{e[X \mapsto \tau']}{\tau [X \mapsto \tau']}$. 
\end{lemma}
\begin{proof}
	Proof is by the induction on the typing derivation of $\TValP{\varcontext, X, \varcontext';\pc}{e}{\tau}$.
\end{proof}

\begin{lemma}[Projection Preserves Types]\label{lemma:dproj}
	If \TValGpc{e}{\tau}, then \TValGpc{\outproj{e}{i}}{\tau} for $i = \{1, 2\}$.
\end{lemma}
\begin{proof}
	Proof is by induction on the typing derivation of \TValGpc{e}{\tau}. The interesting case is $e = \bracket{e₁}{e₂}$.  By \ruleref{Bracket}, we have  
\TValP{\G;pc'}{eᵢ}{\tau} for some $\pc'$ such that \rflowjudge{\delegcontext}{(H^\pi \sqcup \pc^\pi)}{{\pc'^\pi}}.
      Therefore, by Lemma~\ref{lemma:pcred}, we have \TValP{\G;pc}{eᵢ}{\tau}. 
\end{proof}

We expand the proof for the helper lemma necessary to prove the adequacy of bracketed terms.
\stuck*
\begin{proof}
	We prove by induction on the structure of $e$.
	\begin{description}
		\item[Case $w$:] No reduction rules apply to terms in the syntactic category $w$ (including $\bracket{w}{w'}$). Hence \outproj{x}{i} is stuck as well.
		\item[Case $x$:] No reduction rules apply to a variable. Hence \outproj{x}{i} is stuck as well.
		\item[Case $ \bracket{e_1}{e_2}$:] By \ruleref{B-Step}, $e$ is only gets stuck is if both $e_1$ and $e_2$ get stuck.
		\item[Case $e~e'$:] Since $e~e'$ is stuck, then \ruleref{B-App}, \ruleref{W-App}, \ruleref{E-App} are not applicable. 
                    It follows that either (1) $e$ is not of the form $\bracket{w}{w'}$, $\where{w}{v}$, or $\lamc{x}{\tau}{\pc'}{e}$
                    or (2) $e$ has the form $\lamc{x}{\tau}{\pc'}{e}$, but $e'$ is stuck.  For the first case, \outproj{e}{i} is also not of 
                    the form $\bracket{w}{w'}$, $\where{w}{v}$, or $\lamc{x}{\tau}{\pc'}{e}$, so $\outproj{e~e'}{i}$ is also stuck. 
                    For the second case, applying the induction hypothesis gives us that \outproj{e'}{i} is stuck for some $i \in \{1, 2\}$, 
                    so \outproj{\lamc{x}{\tau}{\pc'}{e}~e'}{i} is stuck for the same $i$.
		\item[Case $e~\tau$:] Since $e~e'$ is stuck, then \ruleref{B-TApp}, \ruleref{W-TApp}, \ruleref{E-TApp} are not applicable. 
                    It follows that $e$ is not of the form $\bracket{w}{w'}$, $\where{w}{v}$, or $\tlam{X}{\pc'}{e}$
                    Therefore, \outproj{e}{i} is also not of the form $\bracket{w}{w'}$, $\where{w}{v}$, or $\lamc{x}{\tau}{\pc'}{e}$, so $\outproj{e~e'}{i}$ is also stuck. 
		\item[Case $\return{\ell}{e}$:] Since $\return{\ell}{e}$ is stuck, then \ruleref{E-UnitM} is not applicable, so $e$ does not have the form $w$. Therefore, 
                  \ruleref{E-UnitM} is also not applicable to \outproj{\return{\ell}{e}}{i}.  Therefore, $e$ must be stuck.
                   Applying induction hypothesis, it follows that $\outproj{e}{i}$ is also stuck and so $\outproj{\return{\ell}{e}}{i}$ is also stuck.
		\item[Case $\proj{j}{e}$:]  Similar to the above case.
		\item[Case $\inj{j}{e}$:] Similar to the above case.
		\item[Case $\pair{e}{e}$:] Similar to the above case.
		\item[Case $\casexp{e}{x}{e_1}{e_2}$:] 
                    Since \ruleref{B-Case}, \ruleref{W-Case}, and \ruleref{E-Case} are not applicable, it follows that $e$ is not of the form 
                    $\bracket{w}{w'}$, $\where{w}{v}$, or $\inj{j}{v}$. It follows that $\outproj{\casexp{e}{x}{e_1}{e_2}}{i}$ is also stuck. 
		\item[Case $\bind{x}{v}{e'}$:] Similar to the above case.
		\item[Case $\assume{e}{e'}$:] Similar to the above case.
		\item[Case $\where{e}{v}$:] Similar to the above case.
	\end{description}
\end{proof}

We are now ready to prove subject reduction.
\subjred*
\begin{proof}

  \begin{description}
  \item[Case \ruleref{E-App}:] Given $e = (\lamc{x}{τ}{\pc'}{e})~w$ and
    \TValGpc{(\lamc{x}{τ}{\pc'}{e})~v}{τ_2}. From \ruleref{App} we have,
    \begin{align} 
    \TValP{\G,x:τ';\pc'}{e}{τ} \\
    \TValGpc{w}{τ'} \\
    \rflowjudge{\delegcontext}{pc'}{pc}
   \end{align} 

    Therefore, via PC reduction (Lemma~\ref{lemma:pcred}) and variable substitution (Lemma~\ref{lemma:vsubst}), we have that
     \TValP{\G;\pc'}{\subst{e}{x}{w}}{τ}.
        
  \item[Case \ruleref{E-TApp}:] Similar to above case, but using Lemma~\ref{lemma:tsubst}.
  \item[Case \ruleref{E-Case1}:] Given
    $e = \casexp{(\ione{w})}{x}{e_1}{e_2}$  and $e' = \subst{e_1}{x}{v}$.
    Also, 
    \TValP{\G; \pc}{\casexp{(\ione{w})}{x}{e_1}{e_2}}{τ}. From the premises we have,
    \TValP{\G; \pc}{\ione{v}}{τ' + τ''} and
    \TValP{\G,x:τ';\pc}{e_1}{τ}. From \ruleref{Inj}, we have 
    \TValP{\G; \pc}{v}{τ'}.
    Invoking variable substitution lemma (Lemma~\ref{lemma:vsubst}), we have
    \TValP{\G; \pc}{\subst{e_1}{x}{w}}{τ}.

  \item[Case \ruleref{E-Case2}:] Similar to above.
  
  \item[Case \ruleref{E-UnitM}:] Given $e = \return{\ell}{w}$ and $e' = \returng{\ell}{w}$.
    Also, 
    \TValGpc{\return{\ell}{w}}{\says{\ell}{τ}}.
    From the premises it follows that
    \TValGpc{\returng{\ell}{w}}{\says{\ell}{τ}}.

  \item[Case \ruleref{E-BindM}:] Given $e = \bind{x}{w}{e'} $ and $e' = \subst{e'}{x}{w}$
    Also,
   \TValGpc{\bind{x}{w}{e'}}{\tau}. From the premises, we have the following:
    \begin{align}
      \TValGpc{w}{\tau'} \label{eq:sbindm1} \\
      \TValP{\G, x:\tau'; \pc \sqcup \ell}{e'}{\tau} \label{eq:sbindm2} \\
      \protjudge{\delegcontext}{\pc \sqcup \ell}{\tau} \label{eq:sbindm4} \\
      \rafjudge{\delegcontext}{p}{\pc}
    \end{align}

    We have to prove that
    \TValGpc{\subst{e'}{x}{w}}{\tau}.
    Since we have that $\rafjudge{\delegcontext}{p}{\pc}$,
    applying PC reduction (Lemma~\ref{lemma:pcred}) to the premise \eqref{eq:sbindm2}, we have
    \TValP{\G, x:\tau'; \pc}{e'}{\tau}.

    Invoking variable substitution lemma (Lemma~\ref{lemma:vsubst}), we thus have
    \TValGpc{\subst{e'}{x}{w}}{\tau}.
    
  \item[Case \ruleref{E-Assume}:]
 Given $e = \assume{\delexp{p}{q}}{e'}$ and $e' = \where{e'}{\delexp{p}{q}}$. Also,
 	\TValP{\G; \pc}{\assume{\delexp{p}{q}}{e'}}{\tau}. From \ruleref{Assume}, we have
        \begin{align}
                \TValGpc{\delexp{p}{q}}{\aftype{p}{q}} \\
		\TVal{\delegcontext,\delexp{p}{q}; \G; \pc}{e'}{\tau} \\
		\rafjudge{\delegcontext}{\pc}{\voice{q}} \\
		\rafjudge{\delegcontext}{\voice{\confid{\pl}}}{\voice{\confid{q}}}
       \end{align}
       We need to prove:
       \[
              \TValGpc{\where{e'}{\delexp{p}{q}}}{\tau}
       \]
       Comparing with the given premises, we already have the required premises.
        \begin{align}
		\TValGpc{\delexp{p}{q}}{\aftype{p}{q}} \\
		\TVal{\delegcontext, \delexp{p}{q}; \G; \pc}{e}{\tau} \\
		\rafjudge{\delegcontext}{\pcmost}{\voice{q}} \\
		\rafjudge{\delegcontext}{\voice{\confid{p}}}{\voice{\confid{q}}}
       \end{align}
       Hence proved.
       
  \item[Case \ruleref{E-Eval}:] 
     For $e = E[e]$ and $e' = E[e']$ where $\TValGpc{E[e]}{\tau}$, we have $\TVal{Π';\G';\pc'}{e}{\tau'}$ 
     for some $\pc'$, $τ'$, $Π'$, and $\G'$ such that $Π' ⊇ Π$, $\G' ⊇ \G$. By the induction hypothesis we have 
     $\TVal{Π';\G';\pc'}{e'}{\tau'}$.  Observe that with the exception of \ruleref{Sealed} and \ruleref{Where}, the premises
     of all typing rules use terms from the syntactic category $e$.  Therefore if a derivation for $\TValGpc{E[e]}{\tau}$ exists,
     it must be the case that derivation for $\TValGpc{E[e']}{\tau}$ exists where the derivation of $\TVal{Π';\G';\pc'}{e}{\tau'}$ 
      is replaced with $\TVal{Π';\G';\pc'}{e'}{\tau'}$.  Rules \ruleref{Sealed} and \ruleref{Where} have premises that use terms from the
      syntactic category $v$, but since these are fully evaluated, $e$ cannot be equal to a $v$ term since no $e'$ exists such that 
      $v \stepsto e'$.

  \item[Case \ruleref{W-App}:]
    Given $e = {(\where{w}{\delexp{p}{q}})~e}$ and 
    $e' = \where{(w~e)}{\delexp{p}{q}}$. We have to prove that 
    \[
    \TValGpc{\where{(w~e)}{\delexp{p}{q}}}{\tau}
    \]
    From \ruleref{App} we have: 
    \begin{align}
    \TValGpc{\where{w}{\delexp{p}{q}}}{\func{τ_1}{\pc'}{τ}}\label{eq:dsubwapp1}  \\
    \TValGpc{e}{τ₁}\label{eq:dsubwapp2} \\
    \rflowjudge{\delegcontext}{\pc}{\pc'} \label{eq:dsubwapp3} \\
       \end{align}
    Rule \ruleref{Where} gives us the following:
    \begin{align}
    \TValGpc{\delexp{p}{q}}{\aftype{p}{q}} \label{eq:dsubwapp4}\\
    \rafjudge{\delegcontext}{\pcmost}{\voice{q}} \label{eq:dsubwapp5} \\
    \rafjudge{\delegcontext}{\voice{\confid{p}}}{\voice{\confid{q}}} \label{eq:dsubwapp6}\\
    \TVal{\delegcontext, \delexp{p}{q}; \G; \pc}{w}{\func{τ_1}{\pc'}{τ}} \label{eq:dsubwapp7}
    \end{align}
    We now want to show that $e'$ is well typed via \ruleref{Where}. The key premise is to show that the 
    subexpression $(w~e)$ is well-typed via \ruleref{App}. That is,
    \begin{equation}
    \TVal{\delegcontext, \delexp{p}{q}; \G; \pc }{w~e}{τ} \label{eq:dsubwapp10}
    \end{equation}

    Applying Lemma~\ref{lemma:dpiextension} (extending delegation contexts for well-typed terms) to \eqref{eq:dsubwapp2} and   Lemma~\ref{lemma:robext} (extending delegation contexts for assumptions) to \eqref{eq:dsubwapp3}, we have:
    \begin{align}
    \TVal{\delegcontext, \delexp{p}{q}; \G; \pc}{e}{τ₁} \\
    \rflowjudge{\delegcontext, \delexp{p}{q}}{\pc}{\pc'} 
    \end{align}
    Combining with \eqref{eq:dsubwapp7}, we have \eqref{eq:dsubwapp10} which when combined with remaining premises (\eqref{eq:dsubwapp4}, \eqref{eq:dsubwapp5} and \eqref{eq:dsubwapp6})  give us
    \TValGpc{\where{(w~e)}{\delexp{p}{q}}}{\tau}.

    \item[Case \ruleref{W-TApp}:]
    Given $e = {(\where{w}{\delexp{p}{q}})~\tau}$ and 

    $e' = \where{(w~\tau')}{\delexp{p}{q}}$. We have to prove that 
    \[
    \TValGpc{\where{(v~\tau')}{\delexp{p}{q}}}{\tau[X \mapsto \tau']}
    \]
    From \ruleref{TApp} we have: 
    \begin{align}
    \TValGpc{\where{w}{\delexp{p}{q}}}{\tfuncpc{X}{\pc'}{τ}}\label{eq:dsubwtapp1}  \\
    \rflowjudge{\delegcontext}{\pc}{\pc'} \label{eq:dsubwtapp2}
    \end{align}
    Rule \ruleref{Where} gives us the following:
    \begin{align}
    \TValGpc{\delexp{p}{q}}{\aftype{p}{q}} \label{eq:dsubwtapp4}\\
    \rafjudge{\delegcontext}{\pcmost}{\voice{q}} \label{eq:dsubwtapp5} \\
    \rafjudge{\delegcontext}{\voice{\confid{p}}}{\voice{\confid{q}}} \label{eq:dsubwtapp6}\\
    \TVal{\delegcontext, \delexp{p}{q}; \G;\pc}{v}{\tfuncpc{X}{\pc'}{τ}} \label{eq:dsubwtapp7}
    \end{align}
    We now want to show that $e'$ is well typed via \ruleref{Where}. The key premise is to show that the 
    subexpression $(v~\tau')$ is well-typed via \ruleref{TApp}. That is,
    \begin{equation}
    \TVal{\delegcontext, \delexp{p}{q}; \G; \pc}{w~\tau'}{τ[X \mapsto \tau']} \label{eq:dsubwtapp10}
    \end{equation}

    Applying Lemma~\ref{lemma:robext} (extending delegation context for well-typed terms) to \eqref{eq:dsubwtapp2}, we get:	
    \begin{align}
    \rflowjudge{\delegcontext, \delexp{p}{q}}{\pc}{\pc'} 
    \end{align}
    Combining with \eqref{eq:dsubwtapp7}, we have \eqref{eq:dsubwtapp10} which when
    combined with remaining premises (\eqref{eq:dsubwtapp4}, \eqref{eq:dsubwtapp5} and \eqref{eq:dsubwtapp6}) give
       \TValGpc{\where{(w~\tau')}{\delexp{p}{q}}}{\tau [X \mapsto \tau']}.
       
    \item[Case \ruleref{W-UnPair}:]
    Given $e = {\proj{i}{(\where{\pair{w_1}{w_2}}{\delexp{p}{q}})}}$ and $e' = {\where{(\proj{i}{\pair{w_1}{w_2}})}{\delexp{p}{q}}}$. We have to prove that 
		\[
			\TValGpc{\where{(\proj{i}{\pair{w_1}{w_2}})}{\delexp{p}{q}}}{\tau_i}
		\]
		From \ruleref{UnPair}, we have:
		\begin{align}
			\TValGpc{\where{\pair{w_1}{w_2}}{\delexp{p}{q}}}{\prodtype{τ_1}{τ_2}} \label{eq:dwunpair1} 
		\end{align}
		From \eqref{eq:dwunpair1} and \ruleref{Where}, we have:
                \begin{align}
			\TValGpc{\delexp{p}{q}}{\aftype{p}{q}} \label{eq:dwunpair2}\\
			\rafjudge{\delegcontext}{\pcmost}{\voice{q}} \label{eq:dwunpair3}\\
			\rafjudge{\delegcontext}{\voice{\confid{p}}}{\voice{\confid{q}}} \label{eq:dwunpair4}\\
			\TVal{\delegcontext, \delexp{p}{q}; \G; \pc}{\pair{w_1}{w_2}}{\prodtype{\tau_1}{\tau_2}} \label{eq:dwunpair5}
                \end{align}
		From \eqref{eq:dwunpair5} and \ruleref{UnPair}, we have:
			\TVal{\delegcontext, \delexp{p}{q}; \G; \pc}{\proj{i}{\pair{w_1}{w_2}}}{\tau_i}.
		Combining with remaining premises (\eqref{eq:dwunpair2} to \eqref{eq:dwunpair5}) we have
			\TValGpc{\where{\proj{i}{\pair{w_1}{w_2}}}{\delexp{p}{q}}}{\tau_i}.
                        
     \item[Case \ruleref{W-Case}:]
			Given $$e = \casexp{(\where{w}{\delexp{p}{q}})}{x}{e_1}{e_2} $$ and $$e'= \where{(\casexp{w}{x}{e_1}{e_2})}{\delexp{p}{q}}$$
We have to prove that 
		\[
			\TValGpc{\where{(\casexp{w}{x}{e_1}{e_2})}{\delexp{p}{q}}}{\tau}
		\]
                From \ruleref{Case} we have: 
                \begin{align}
			\TValGpc{\where{w}{\delexp{p}{q}}}{τ_1 + τ_2} \label{eq:dcase1} \\
			\protjudge{\delegcontext}{\pc \sqcup \ell}{\tau} \label{eq:dcase2} \\
			\rflowjudge{\delegcontext}{\pc}{\ell} \label{eq:dcase3} \\
			\TValP{\varcontext,~x:τ_1; \pc \sqcup \ell }{e_1}{τ} \label{eq:dcase4} \\
			\TValP{\varcontext,~x:τ_2; \pc \sqcup \ell}{e_2}{τ} \label{eq:dcase5} 
		\end{align}
	From \eqref{eq:dcase1} and rule \ruleref{Where}, we get the following:
                \begin{align}
			\TValGpc{\delexp{p}{q}}{\aftype{p}{q}} \label{eq:dcase6}\\
			\rafjudge{\delegcontext}{\pcmost}{\voice{q}} \label{eq:dcase7}\\
			\rafjudge{\delegcontext}{\voice{\confid{p}}}{\voice{\confid{q}}} \label{eq:dcase8}\\
			\TValGpc{w}{\sumtype{\tau_1}{\tau_2}} \label{eq:dcase9}
                \end{align}
		The key premise to prove is that:
		\[
			\TVal{\delegcontext, \delexp{p}{q}; \G; \pc}{(\casexp{w}{x}{e_1}{e_2})}{\tau}
		\]
		which follows from \eqref{eq:dcase9} and extending delegations in equations \eqref{eq:dcase2} to \eqref{eq:dcase5} (Lemma~\ref{lemma:dpiextension}).
		Combining with remaining premises, we have
		$\TValGpc{\where{(\casexp{w}{x}{e_1}{e_2})}{\delexp{p}{q}}}{\tau}$.
                
                \item[Case \ruleref{W-BindM}:]
			Given $e = {\bind{x}{(\where{w}{\delexp{p}{q}})}{e}}$ and $e' = {\where{(\bind{x}{w}{e})}{\delexp{p}{q}}}$. We have to prove that 
		\[
			\TValGpc{\where{(\bind{x}{w}{e})}{\delexp{p}{q}}}{\tau}
		\]
        From \ruleref{BindM} we have: 
                \begin{align}
			\TValGpc{(\where{w}{\delexp{p}{q}})}{\says{\ell}{\tau'}} \label{eq:dbind1} \\
			\TValP{\G, x:\tau'; \pc \sqcup \ell}{e}{\tau} \label{eq:dbind2} \\
			\protjudge{\delegcontext}{\pc \sqcup \ell }{\tau} \label{eq:dbind4}
	\end{align}
	From \eqref{eq:dbind1} and rule \ruleref{Where}, we get the following:
                \begin{align}
			\TValGpc{\delexp{p}{q}}{\aftype{p}{q}} \label{eq:dbind5}\\
			\rafjudge{\delegcontext}{\pcmost}{\voice{q}} \label{eq:dbind6}\\
			\rafjudge{\delegcontext}{\voice{\confid{p}}}{\voice{\confid{q}}} \label{eq:dbind7}\\
			\TVal{\delegcontext, \delexp{p}{q}; \G; \pc}{w}{\says{\ell}{\tau'}} \label{eq:dbind8}
                        \end{align}
        We now want to show that $e'$ is well typed via \ruleref{Where}. That is, we need the following premises, 
	\begin{align}
		\TVal{\delegcontext, \delexp{p}{q}; \G; \pc}{\bind{x}{w}{e}}{\tau} \label{eq:dbind12} \\
	       \TValGpc{\delexp{p}{q}}{\aftype{p}{q}} \label{eq:dbind13}\\
			\rafjudge{\delegcontext}{\pcmost}{\voice{q}} \label{eq:dbind14}\\
			\rafjudge{\delegcontext}{\voice{\confid{p}}}{\voice{\confid{q}}} \label{eq:dbind15}
       \end{align}
     Extending the delegation context (Lemma~\ref{lemma:dpiextension}) in the premises  \eqref{eq:dbind2}, \eqref{eq:dbind4} and  from \eqref{eq:dbind8} we have \eqref{eq:dbind12}.
                
      We already have \eqref{eq:dbind13} from \eqref{eq:dbind5}; \eqref{eq:dbind14} from \eqref{eq:dbind6}; \eqref{eq:dbind15} from \eqref{eq:dbind7}.
       Combining, we have 
		\TValGpc{\where{(\bind{x}{w}{e})}{\delexp{p}{q}}}{\tau}.

  \item[Case \ruleref{W-Assume}:]
      Given $e = \assume{\where{w}{\delexp{a}{b}}}{e}$ and $e' =\where{\assume{w}{e}}{\delexp{a}{b}}$.
      From \ruleref{Assume}, we have
      \begin{align}
                \TValGpc{\where{v}{\delexp{a}{b}}}{\aftype{r}{s}} \\
		\TVal{\delegcontext, \delexp{r}{s}; \G; \pc}{e}{\tau} \label{eq:ddwassume2} \\
		\rafjudge{\delegcontext}{\pc}{\voice{s}} \label{eq:ddwassume3} \\
		\rafjudge{\delegcontext}{\voice{\confid{r}}}{\voice{\confid{s}}}  \label{eq:ddwassume4}
      \end{align}
      Expanding the first premise using \ruleref{Where}, we have
      \begin{align}
      		\TValGpc{\delexp{a}{b}}{\aftype{a}{b}}  \label{eq:ddwassume5} \\
		\TVal{\delegcontext, \delexp{a}{b}; \G; \pc}{v}{\aftype{r}{s}}  \label{eq:ddwassume6}\\
		\rafjudge{\delegcontext}{\pcmost}{\voice{b}} \label{eq:ddwassume7} \\
		\rafjudge{\delegcontext}{\voice{\confid{a}}}{\voice{\confid{b}}} \label{eq:ddwassume8}
      \end{align}
      We want to show
      \begin{align}
         \TValGpc{\delexp{a}{b}}{\aftype{a}{b}}  \label{eq:ddwassume9} \\
         \TVal{\delegcontext, \delexp{a}{b}; \G; \pc }{\assume{v}{e}}{\tau}  \label{eq:ddwassume10}\\
          \rafjudge{\delegcontext}{\pcmost}{\voice{b}} \label{eq:ddwassume11} \\
          \rafjudge{\delegcontext}{\voice{\confid{a}}}{\voice{\confid{b}}} \label{eq:ddwassume12}
      \end{align}
Extending delegation context (Lemma~\ref{lemma:dpiextension}) in the premises \eqref{eq:ddwassume2}, \eqref{eq:ddwassume3}, \eqref{eq:ddwassume8} and combining with topmost premise we have the \eqref{eq:ddwassume10}. Remaining premises follow from \eqref{eq:ddwassume5}, \eqref{eq:ddwassume7} and \eqref{eq:ddwassume8}.
            
  \item[Case \ruleref{B-Step}:]
  		Given $e  = \bracket{e_1}{e_2}$ and $e' = \bracket{e'_1}{e'_2}$. Also $\TVal{\delegcontext;\varcontext; \pc}{\bracket{e_1}{e_2}}{\tau}$. 
		We have to prove
		\[
			\TVal{\delegcontext;\varcontext; \pc}{\bracket{{e'}_1}{{e'}_2}}{\tau} 
		\]
		Without loss of generality, let $i = 1$. Thus from the premises of \ruleref{B-Step}, we have $e_1 \stepsone {e'}_1$ and $ {e'}_2 = e_2$.
                Since sealed values can not take a step, inverting the well-typedness of bracket is only possible through \ruleref{Bracket} and not through \ruleref{Bracket-Values}.
		From \ruleref{Bracket}, we have
		\begin{align}
			\TVal{\delegcontext; \varcontext;\pc'}{e_1}{\tau} \label{eq:dbracket1}\\
			\TVal{\delegcontext; \varcontext; \pc'}{e_2}{\tau} \label{eq:dbracket2}\\
			\rflowjudge{\delegcontext}{( H^\pi \sqcup \pc^\pi)}{{{\pc'}^\pi}} \label{eq:dbracket3}\\
			\protjudge{\delegcontext}{H^{\pi}}{\tau^\pi} \label{eq:dbracket4}
		\end{align}
		Since \eqref{eq:dbracket1} holds, applying induction to the premise $e_1 \stepsone e'_1$, we have that
			$\TVal{\delegcontext;\varcontext;\pc'}{e'_1}{\tau}$. 
			Combining with remaining premises (\eqref{eq:dbracket2} to \eqref{eq:dbracket4}) we thus have that
			$\TVal{\delegcontext;\varcontext;\pc}{\bracket{e'_1}{e'_2}}{\tau}$. 

         \item[Case \ruleref{B-App}:] Given $e = \bracket{w_1}{w_2}~w'$ and $e'= \bracket{w_1~\outproj{w'}{1}}{w_2~\outproj{w}{2}}$. Also given that $\TValGpc{ \bracket{w_1}{w_2}~w'}{\tau_2}$ is well-typed, from \ruleref{App}, we have the following:
         \begin{align}
                 & \TValGpc{ \bracket{w_1}{w_2}}{\func{\tau_1}{\pc''}{\tau_2}}  \label{eq:bappsubred1} \\
                & \TValGpc{w'}{\tau_1}  \label{eq:bappsubred2}\\
                & \rflowjudge{\delegcontext}{\pc}{\pc''}  \label{eq:bappsubred3}
         \end{align}
         Thus from \ruleref{Bracket-Values}, we have $\rflowjudge{\delegcontext}{H^\pi}{(\func{\tau_1}{\pc''}{\tau_2})^\pi}$. That is, from the definition of type projection (Figure~\ref{fig:typeproj}), we have $\rflowjudge{\delegcontext}{H^\pi}{\func{\tau_1}{{\pc''}^\pi}{\tau_2^\pi}}$.  From \ruleref{P-Fun}, we thus have
         \begin{align}
                & \rflowjudge{\delegcontext}{H^\pi}{\tau_2^\pi} \label{eq:bappsubred4} \\ 
                & \rflowjudge{\delegcontext}{H^\pi}{{\pc''}^\pi}  \label{eq:bappsubred5}
         \end{align}
         We need to prove
         \[
          \TValGpc{\bracket{w_1~\outproj{w'}{1}}{w_2~\outproj{w'}{2}}}{\tau_2}
         \]
         That is we need the following premises of \ruleref{Bracket}.
         \begin{align}
                \TValP{\G;\pc'}{w_1~\outproj{w'}{1}}{\tau_2}  \label{eq:bappsubred6} \\
                \TValP{\G;\pc'}{w_2~\outproj{w'}{2}}{\tau_2}   \label{eq:bappsubred7}\\
                \rflowjudge{\delegcontext}{H^\pi \sqcup \pc^\pi}{{\pc'}^\pi}  \label{eq:bappsubred8} \\
                \rflowjudge{\delegcontext}{H^\pi }{\tau_2^\pi}  \label{eq:bappsubred9}
         \end{align}

         Let $\pc' = \pc''$.
         We have  \eqref{eq:bappsubred8} from   \eqref{eq:bappsubred3} and \eqref{eq:bappsubred5}.
         We already have  \eqref{eq:bappsubred9} from  \eqref{eq:bappsubred4}.
         To prove \eqref{eq:bappsubred6}, we need the following premises:
            \begin{align}
                \TValP{\G;\pc'}{w_1}{\func{\tau_1}{\pc''}{\tau_2}}  \label{eq:bappsubred10} \\
                \TValP{\G;\pc'}{\outproj{w'}{1}}{\tau_2}  \label{eq:bappsubred11} \\
                \rflowjudge{\delegcontext}{\pc'}{\pc''} \label{eq:bappsubred12}
            \end{align}
            The last premise \ \eqref{eq:bappsubred12} holds trivially (from reflexivity).
            Applying Lemma~\ref{lemma:wvaluespc} (values can be typed under any \pc) to  \eqref{eq:bappsubred1} we have \eqref{eq:bappsubred10}.
            Applying Lemma~\ref{lemma:wvaluespc} (values can be typed under any \pc) and  Lemma~\eqref{lemma:dproj} (projection preserves typing) to \eqref{eq:bappsubred2} we have \eqref{eq:bappsubred11}.
            Thus from \ruleref{App}, we have  \eqref{eq:bappsubred6}. Similarly,  \eqref{eq:bappsubred7} holds.
            Hence proved.

	\item[Case \ruleref{B-TApp}:] Similar to above (\ruleref{B-App}) case.
	\item[Case \ruleref{B-UnPair}:]
		Given $e =\proj{i}{\bracket{\pair{w_{11}}{w_{12}}}{\pair{w_{21}}{w_{22}}}}$ and $e'  = \bracket{w_{1i}}{w_{2i}}$.
		Also 
			\TValGpc{\proj{i}{\bracket{\pair{w_{11}}{w_{12}}}{\pair{w_{21}}{w_{22}}}}}{\tau_i} 
		We have to prove
		\[
			\TValGpc{\bracket{w_{1i}}{w_{2i}}}{\tau_i} 
		\]
		From \ruleref{UnPair}, we have:
		\begin{align}
			\TValGpc{\bracket{\pair{w_{11}}{w_{12}}}{\pair{w_{21}}{w_{22}}}}{\prodtype{\tau_1}{\tau_2}} \label{eq:dbunpair1}
		\end{align}
                Since they are already values, they can be inverted using \ruleref{Bracket-Values}. This approach is more conservative.
		\begin{align}
			\TValGpc{\pair{w_{11}}{w_{12}}}{\prodtype{\tau_1}{\tau_2}} \label{eq:dbunpair2}\\
			\TValGpc{\pair{w_{21}}{w_{22}}}{\prodtype{\tau_1}{\tau_2}} \label{eq:dbunpair3}\\
			\protjudge{\delegcontext}{H^{\pi}}{(\prodtype{\tau_1}{\tau_2})^\pi} \label{eq:dbunpair5}
                        \end{align}	
		From \eqref{eq:dbunpair2}, \eqref{eq:dbunpair3} and \ruleref{UnPair}, we have
			$\TValGpc{w_{1i}}{\tau_i}$ and $\TValGpc{w_{2i}}{\tau_i}$  for $i = \{1, 2\}$. 
		From \eqref{eq:dbunpair5},  type projection (Figure~\ref{fig:typeproj}) and \ruleref{P-Pair}, we have
		$\protjudge{\delegcontext}{H^\pi}{\tau_i^\pi}$.
		Combining with other premises, 
                $\TValGpc{\bracket{w_{1i}}{w_{2i}}}{\tau_i}$ follows from \ruleref{Bracket-Values}.
                
	\item[Case \ruleref{B-BindM}:] 
		Given $e = {\bind{x}{\bracket{\returng{\ell}{w_1}}{\returng{\ell}{w_2}}}{e}}$.
		We have that:
		$$e' = \bracket{\bind{x}{\returng{\ell}{w_1}}{\outproj{e}{1}}}{\bind{x}{\returng{\ell}{w_2}}{\outproj{e}{2}}}$$
		Also $\TValGpc{\bind{x}{\bracket{\returng{\ell}{w_1}}{\returng{\ell}{w_2}}}{e}}{\tau}$. 
		From \ruleref{BindM}, we have
		\begin{align}
			&\TValGpc{\bracket{\returng{\ell}{w_1}}{\returng{\ell}{w_2}}}{\says{\ell}{\tau'}} \label{eq:db1bindm1} \\
			&\TValP{\G, x:\tau'; \pc \sqcup \ell}{e}{\tau} \label{eq:db1bindm2} \\
                      	& \protjudge{\delegcontext}{\pc \sqcup \ell}{\tau} \label{eq:db1bindm4}
		\end{align}
		From \eqref{eq:db1bindm1} and \ruleref{Bracket-Values}, we have
		\begin{align}
			\TValGpc{\returng{\ell}{w_1}}{\says{\ell}{\tau'}} \label{eq:db1bracket1}\\
                        \TValGpc{\returng{\ell}{w_2}}{\says{\ell}{\tau'}} \label{eq:db1bracket2}\\
	                \rflowjudge{\delegcontext}{H^\pi}{\ell^\pi} \label{eq:db1bindm3}
		\end{align}
		We have to prove that
		\[
                \TValGpc{\bracket{\bind{x}{\returng{\ell}{w_1}}{\outproj{e}{1}}}{\bind{x}{\returng{\ell}{w_2}}{\outproj{e}{2}}}}{\tau}
	        \]
		For some $\widehat{\pc}$ we need the following premises to satisfy \ruleref{Bracket}:
		\begin{align}
			\TValP{\G;\widehat{\pc}}{\bind{x}{\returng{\ell}{w_1}}{\outproj{e}{1}}}{\tau} \label{:eq:db1bracket8}\\
			\TValP{\G; \widehat{\pc}}{\bind{x}{\returng{\ell}{w_2}}{\outproj{e}{2}}}{\tau} \label{:eq:db1bracket9}\\
			\rflowjudge{\delegcontext}{(H^\pi \sqcup \pc^\pi)}{{\widehat{\pc}^\pi}} \label{:eq:db1bracket10}\\
			\protjudge{\delegcontext}{H^\pi}{\tau^\pi} \label{:eq:db1bracket11}
		\end{align}
		A natual choice for $\widehat{pc}$ is $\pc \sqcup \ell$. 
                From Lemma~\ref{lemma:wvaluespc} (values  can be typed under any $\pc$), we have
                \[
        	\TValP{\G;\widehat{\pc}}{\returng{\ell}{w_i}}{\tau'}
	        \]
               Applying Lemma~\ref{lemma:dproj} (bracket projection preserves typing) to \eqref{eq:db1bindm2}, we have 
                \[
        	\TValP{\G, x:\tau';\widehat{\pc}}{\outproj{e}{i}}{\tau}
	        \]
                From \ruleref{BindM}, we therefore have \eqref{:eq:db1bracket8} and \eqref{:eq:db1bracket9}.
		Applying \ruleref{R-Trans} to \eqref{eq:db1bindm3} and \eqref{eq:db1bindm4}, we have \eqref{:eq:db1bracket11}.
                Thus we have all required premises.

  \item[Case \ruleref{B-Case}:] Does not occur. Not well-typed.
  \item[Case \ruleref{B-Assume}:] Does not occur. Not well-typed.
    
  \end{description}
\end{proof}

\subsection{Proof for Progress (Lemma~\ref{lemma:prog})}\label{app:prog}
\prog*

\begin{proof}
  Proof is by induction on the derivation of the typing judgment.
  \begin{description}
  \item[Case Var:] Does not occur as $e$ is closed, and $\G$ is empty.
  \item[Case Unit:] Already a value.
  \item[Case Del:] Already a value.
  \item[Case Lam:] Already a value.
  \item[Case TLam:] Already a value.
  \item[Case App:] Given  $\TVal{\G;\pc}{e~e'}{τ}$. From \ruleref{App}, we have the following
    \begin{align}
      \TValGpc{e}{(\func{τ₁}{\pc'}{τ₂})} \label{eq:prog1}\\
      \TValGpc{e'}{τ₁}  \label{eq:prog2}\\
      \rflowjudge{Π}{\pc}{\pc'}
    \end{align}
    Applying IH to \eqref{eq:prog1} and \eqref{eq:prog2}, we have that $e$ takes a step or is some value such that the type is $\func{\tau₁}{\pc'}{\tau₂}$.  Similarly, $e'$ either takes a step or is some value $w$. When both $e$ and $e'$ are values, we have the following cases: $e = \lamc{x}{\tau₁}{\pc'}{e^{'''}}$ or $e=\where{\lamc{x}{\tau₁}{\pc'}{e^{'''}}}{v}$.
    \begin{description}
    \item[Case 1:]. Let $e =\lamc{x}{\tau₁}{\pc'}{e^{'''}}$.  Then by \ruleref{E-App}, we have that $\lamc{x}{\tau₁}{\pc'}{e^{'''}}~w \stepsone e^{'''}[ x\mapsto w]$.
      \item[Case 2:] Let  $e=\where{\lamc{x}{\tau₁}{\pc'}{e^{'''}}}{v}$. Then by \ruleref{W-App}, we have that $\where{\lamc{x}{\tau₁}{\pc'}{e^{'''}}}{v}~w \stepsone \where{\lamc{x}{\tau₁}{\pc'}{e^{'''}}~w}{v} $.
      \end{description}
  \item[Case TApp:] Similar to \ruleref{App} case except that the argument is a type and thus does not take a step by itself. 
  \item[Case Pair:] Applying IH to the premises of \ruleref{Pair}, we have that either $e_i$ takes a step or is already a value for $i \in \{1, 2\}$.
  \item[Case UnPair]  Given  $\TVal{\G;\pc}{\proji{e}}{τ}$. Applying I.H. to the premises of \ruleref{UnPair}, we have that either $e$ takes a step or is already a value. If $e$ is a value then from the well-typedness of $e$, we have that $e = \pair{w₁}{w₂}$ or $e=\where{\pair{w₁}{w₂}}{v}$. In the former case, $\proj{i}{e}$ takes a step according to \ruleref{E-UnPair}. In the latter case, it takes a step according to \ruleref{W-UnPair}.
    
    \item[Case Inj:] Given  $\TVal{\G;\pc}{\inji{e}}{τ}$. Applying I.H. to the premises of \ruleref{InjI}, we have that either $e$ takes a step or is already a value. If $e$ is a value we have that $\inji{e}$ is also a value.
    \item[Case Case] Given  $\TVal{\G;\pc}{\casexp{e}{x}{e₁}{e₂}}{τ}$. Applying I.H. we have that $e$ either takes a step or is already some value. By well-typedness, we have that either $e = \inji{w}$ or $e = \where{\inji{w}}{v}$. In the former case, it takes a step following \ruleref{E-Case}. In the latter case, it takes a step according to \ruleref{W-Case}.
    \item[Case UnitM:] Applying I.H. to the premises of \ruleref{UnitM} we have that $e$ either takes a step or is already a some value $w$. In the latter case, it takes a step as per \ruleref{E-UnitM}.
    \item[Case Sealed:] Already a value.
    \item[Case BindM:] Applying I.H. to the premises of \ruleref{BindM}, $e$ takes a step or is already a value. In the latter case, from the well-typedness of $e$ we have that either $e = \returng{\ell}{w}$ or $e = \where{\returng{\ell}{w}}{v}$. In the former case, it takes a step according to \ruleref{E-BindM}. In the latter case, it takes a step as per \ruleref{W-BindM}.
    \item[Case Assume:] Applying I.H. to the premises of \ruleref{Assume}, we have that either $e$ takes a step or is some value. In the latter case, from the well-typedness, we have that either $e = \delexp{p}{q}$ or $e = \where{\delexp{p}{q}}{v}$. In the former case, it takes a step according to \ruleref{E-Assume}. In the latter case, it takes a step according to \ruleref{W-Assume}.
      \item[Case Where:] Applying I.H. to the premises of \ruleref{Where}, we have that either $e$ takes a step or is some value. In the latter case, the entire term is a value.
  \end{description}
\end{proof}

\subsection{Proofs for Erasure Conservation (Lemma~\ref{lemma:erasecons})}\label{app:erasecons}

To prove erasure conservation, we first prove that substituting program or type variable and evaluation context substitution preserves erasure.
\begin{lemma}[Substitution preserves erasure]
  \label{lemma:substerase}
  Let \TValP{\G, x:\tau'; \pc}{e}{\tau} and \TValP{\G, \pc}{w}{\tau'}. For all $pc', \ell$, if
  $\observef{\outproj{e}{1}}{Π}{ℓ} =\observef{\outproj{e}{2}}{Π}{ℓ}$ and
   $\observef{\outproj{w}{1}}{Π}{ℓ} =\observef{\outproj{w}{2}}{Π}{ℓ}$
  then
      $\observef{\outproj{\subst{e}{x}{w}}{1}}{Π}{ℓ} =\observef{\outproj{\subst{e}{x}{w}}{2}}{Π}{ℓ}$.
\end{lemma}
\begin{proof}
  Proof is by induction on the structure of the expression $e$.
  \begin{description}
  \item[Case Var:] If $e = x$ then $\outproj{\subst{e}{x}{w}}{i} = \outproj{w}{i}$. We are already given that   $\observef{\outproj{w}{1}}{Π}{ℓ} =\observef{\outproj{w}{2}}{Π}{ℓ}$. Unfolding the definition of \observe function, we thus have that 
  $\observef{\outproj{\subst{e}{x}{w}}{1}}{Π}{ℓ} =\observef{\outproj{\subst{e}{x}{w}}{2}}{Π}{ℓ}$
    \medskip
    If $e \ne x$ then $\outproj{\octx{pc'}{\subst{e}{x}{w}}}{i} = \outproj{\octx{pc'}{w}}{i}$. We are already given that $\observef{\outproj{e}{1}}{Π}{ℓ} =\observef{\outproj{e}{2}}{Π}{ℓ}$.

  \item[Case Unit:] Substitution does not affect $e$.
    \item[Case \delexp{p}{q}:] Substitution does not affect $e$.
    \item[Case \return{\ell'}{e'}:]
      We are given that
      $\observef{\outproj{\return{\ell'}{e'}}{1}}{Π}{ℓ} =\observef{\outproj{\return{\ell'}{e'}}{2}}{Π}{ℓ}$.
      We have to prove
      \[
      \observef{\outproj{\subst{\return{\ell'}{e'}}{x}{w}}{1}}{Π}{ℓ} = \observef{\outproj{\subst{\return{\ell'}{e'}}{x}{w}}{2}}{Π}{ℓ}
      \]
      This implies, from the projection definition, we have to prove
      \[
      \observef{\return{\ell'}{\outproj{\subst{e'}{x}{w}}{1}}}{Π}{ℓ} = \observef{\return{\ell'}{\outproj{\subst{e'}{x}{w}}{2}}}{Π}{ℓ}
      \]
      If $\ell'$ is higher than $\ell$, it is trivial as both sides are holes. If not, we have to prove that
      \[
      \returng{\ell'}{\observef{\outproj{\subst{e'}{x}{w}}{1}}{Π}{ℓ}} = \returng{\ell'}{\observef{\outproj{\subst{e'}{x}{w}}{2}}{Π}{ℓ}}
      \]
       Our argument proceeds as follows. Applying the projection to the given terms, we have
      \[
      \observef{\outproj{\return{\ell'}{e'}}{i}}{Π}{ℓ} =\observef{\return{\ell'}{\outproj{e'}{i}}}{Π}{ℓ}
      \]
      Unfolding the definition of \observe function, the interesting case is when $\ell'$ is not higher (in lattice) than $\ell$.  That is,
      \[
      \observef{\return{\ell'}{\outproj{e'}{i}}}{Π}{ℓ} = \returng{\ell'}{\observef{\outproj{e'}{i}}{Π}{ℓ}}
      \]
      
      Applying I.H to $e'$, we have that,
      \[
      \observef{\outproj{\subst{e'}{x}{w}}{1}}{Π}{ℓ} = \observef{\outproj{\subst{e'}{x}{w}}{2}}{Π}{ℓ}
      \]
   Thus,
      \[
      \returng{\ell'}{\observef{\outproj{\subst{e'}{x}{w}}{1}}{Π}{ℓ}} = \returng{\ell'}{\observef{\outproj{\subst{e'}{x}{w}}{2}}{Π}{ℓ}} \] Hence proved.
   \item[Case \returng{\ell'}{w'}:] Similar to above case.
   \item[Other:] The argument for other cases is a straight forward application of inductive hypothesis.
   \end{description}
\end{proof}

\begin{lemma}[Type substitution preserves erasure]
  \label{lemma:typesubsterase}
  Let \TValP{\G, X; \pc}{e'}{\tau'}. For all $pc', \ell$, if  $\observef{\outproj{e'}{1}}{Π}{ℓ} =\observef{\outproj{e'}{2}}{Π}{ℓ}$  then
      $\observef{\outproj{\subst{e'}{X}{\tau}}{1}}{Π}{ℓ} =\observef{\outproj{\subst{e'}{X}{\tau}}{2}}{Π}{ℓ}$.
\end{lemma}
\begin{proof}
Proof by inducting on the structure of expression. Moreover, \observe function does not erase types.
\end{proof}

\begin{lemma}[Evaluation Context substitution preserves erasure]
  \label{lemma:contextsubsterase}
Let   \TValGpc{E[e]}{\tau}  \label{eq:eceval2}. Then,
    $\observef{\outproj{E[e]}{1}}{Π}{ℓ} = \observef{\outproj{E[e]}{2}}{Π}{ℓ}$
   $\iff$
   $\observef{\outproj{e}{1}}{Π}{ℓ} = \observef{\outproj{e}{2}}{Π}{ℓ}$
\end{lemma}
\begin{proof}
  Consider the forward direction. We have to prove
  \[
   \observef{\outproj{E[e]}{1}}{Π}{ℓ} = \observef{\outproj{E[e]}{2}}{Π}{ℓ} \implies
   \observef{\outproj{e}{1}}{Π}{ℓ} = \observef{\outproj{e}{2}}{Π}{ℓ}
   \]
   Proof is by induction on the structure of the evaluation context.
   \begin{description}
     \item[Case $\hole$:] Given $E=\hole$ and so $E[e] = e$. Thus
  \[
   \observef{\outproj{E[e]}{1}}{Π}{ℓ} = \observef{\outproj{E[e]}{2}}{Π}{ℓ} \implies
   \observef{\outproj{e}{1}}{Π}{ℓ} = \observef{\outproj{e}{2}}{Π}{ℓ}
   \]

 \item[Case $E~e'$:] Given
    \[
   \observef{\outproj{E[e]~e'}{1}}{Π}{ℓ} = \observef{\outproj{E[e]~e'}{2}}{Π}{ℓ}
   \]
   This implies,
   \begin{align}
   \observef{\outproj{E[e]}{1}}{Π}{ℓ} = \observef{\outproj{E[e]}{2}}{Π}{ℓ} \label{eq:ctxterase1} \\
   \observef{\outproj{e'}{1}}{Π}{ℓ} = \observef{\outproj{e'}{2}}{Π}{ℓ}  \label{eq:ctxterase2}
   \end{align}

   From \ruleref{App}, we have $\TValGpc{E[e]}{\func{\tau_1}{\pc'}{\tau_2}}$ and $\TValGpc{e'}{\tau_1}$.
   Using  \eqref{eq:ctxterase1} applying I.H to the former, we have 
    \begin{equation}
   \observef{\outproj{e}{1}}{Π}{ℓ} = \observef{\outproj{e}{2}}{Π}{ℓ}  \label{eq:ctxterase3}
   \end{equation}
    Hence proved.
    \item[Other:] Argument similar to previous cases follows.
   \end{description}
\end{proof}

The root of a "where" term is the  term obtained after peeling off all the outer delegations.
\begin{definition}[Root term]\label{def:whereroot}
  \begin{align*}
  \whereroot{e} = {\begin{cases}
      \whereroot{e'} & \mbox{ if } e = \where{e'}{v} \\
      e          & \mbox{ o.w }
      \end{cases}
  }
  \end{align*}
\end{definition}

We have the property that erasing a "where" term is equal to erasing its root.
\begin{lemma}[\observe of "where" term]\label{lemma:whereobserve}
For all $e, \Pi$ and $p$, \observef{e}{\Pi}{p} = \observef{\whereroot{e}}{\Pi}{p}
\end{lemma}
\begin{proof}[Proof Sketch]
  Immediate from the definitions of \observe (Figure~\ref{fig:observe}) and \whereroot{e} (Definition \ref{def:whereroot}).
\end{proof}

The following lemma is a sanity check on the correctness of the
erasure function (Figure~\ref{fig:observe}). Intuitively, two
well-typed values should be observationally equivalent with respect to
an observer that cannot observe secretsor untrusted values. In the
below lemma, $p^\pi$ is the observer and secret or untrusted data is
labelled at $H^\pi$. This is crucial to proving the erasure
conservation lemma.

\begin{lemma}[Correctness of \observe]\label{lemma:correctobserve}
  For all $\Pi', \G,  \pc$ and  $i \in \{1,2\}$, if \TVal{\Pi';\G;\pc}{w_i}{\tau}  such that \protjudge{\Pi}{H^\pi}{\tau^\pi} and  \notrflowjudge{\Pi}{H^\pi}{p^\pi}, then $\observef{w_1}{\Pi}{p}$ = $\observef{w_2}{\Pi}{p}$. 
\end{lemma}
\begin{proof}
  Proof is by induction on the structure of the value $\bracket{w_1}{w_2}$. We only show valid in which $w_1$ and $w_2$ has same type and structure.\footnote{Ideally, the requirement that $w_1$ and $w_2$ have the same structure does not follow from the fact that $w_1$ and $w_2$ have same type, and thus should be stated in the lemma. For simplicity, we are not stating it. However, the callers of this lemma do maintain that property.}
  Note that in all the below cases, the term $\bracket{w_1}{w_2}$ is well-typed by the typing rule \ruleref{Bracket-Values}.
  \begin{description}
  \item[Unit:] Given $\bracket{w_1}{w_2} = \bracket{()}{()}$. From the definition of \observe (Figure~\ref{fig:observe}), we have $\observef{()}{\Pi}{p}$ = $()$, and hence $\observef{\outproj{w}{1}}{\Pi}{p}$ = $\observef{\outproj{w}{2}}{\Pi}{p}$.
    \item[\delexp{p}{q}:] Similar to above
  \item[UnitM:] Given $\bracket{w_1}{w_2}= \bracket{\returng{\ell}{w'_1}}{\returng{\ell}{w'_2}}$ and   \TVal{\Pi';\G;\pc}{\returng{\ell}{w'_i}}{\tau} such that
    $\tau = \says{\ell}{\tau'}$.
    From the definition of type projection in Figure~\ref{fig:typeproj}, we have that \protjudge{\Pi}{H^\pi}{\tau^\pi} is $\protjudge{\Pi}{H^\pi}{\says{\ell^\pi}{\tau'}}$.
    From \ruleref{P-Lbl}, we thus have $\rflowjudge{\Pi}{H^\pi}{\ell^\pi}$.
     This combined with the given premise  \notrflowjudge{\Pi}{H^\pi}{p^\pi} gives  $\notrflowjudge{\Pi}{\ell^{\pi}}{p^\pi}$.
    And so, $\observef{w_i}{\Pi}{p}$ = $\circ$.
  \item[Lam:]  Given $\bracket{w_1}{w_2}= \bracket{\lamc{x}{\tau_1}{\pc'}{e_1}}{\lamc{x}{\tau_1}{\pc'}{e_2}}$ such that $\tau = \func{\tau_1}{\pc'}{\tau_2}$. Since we have \protjudge{\Pi}{H^\pi}{(\func{\tau_1}{\pc'}{\tau_2})^\pi}, from the definition of type projection in Figure~\ref{fig:typeproj}, it follows that \protjudge{\Pi}{H^\pi}{\func{\tau_1}{\pc'^\pi}{\tau_2^\pi}}.  From \ruleref{P-Fun} we have that \rflowjudge{\Pi}{H^\pi}{\pc'^\pi}. This combined with the given premise  \notrflowjudge{\Pi}{H^\pi}{p^\pi} gives    \notrflowjudge{\Pi}{\pc'^\pi}{p^\pi}.
    From the definition of \observe (Definition~\ref{fig:observe}), we have $\observef{ \lamc{x}{\tau_1}{\pc'}{e_i}}{\Pi}{p}$ = $\circ$ if $\notrflowjudge{\Pi}{\pc'^\pi}{p^\pi}$, and thus $\observef{w_1}{\Pi}{p}$ = $\observef{w_2}{\Pi}{p}$.
  \item[TLam:] Similar to the above case.
  \item[Inji:] This is an invalid case since we have $\bracket{w_1}{w_2} = \bracket{\inji{w_1}}{\inji{w_2}}$, however, $w$ cannot be well-typed.
  \item[Pair:] Given $\bracket{w_1}{w_2} = \bracket{\pair{w_{11}}{w_{12}}}{\pair{w_{21}}{w_{22}}}$. We have that $\TVal{\Pi';\G; \pc}{\pair{w_{i1}}{w_{i2}}}{\tau_1 \times \tau_2}$. We have to prove
    \[
    \observef{\pair{w_{11}}{w_{12}}}{\Pi}{p} = \observef{\pair{w_{21}}{w_{22}}}{\Pi}{p}
    \]
    It suffices to prove the following.
    \begin{align*}
      \observef{w_{11}}{\Pi}{p} = \observef{w_{21}}{\Pi}{p} \\
      \observef{w_{12}}{\Pi}{p} = \observef{w_{22}}{\Pi}{p}
    \end{align*}
    
    Consider the terms from the pair projection $w_{11}$ and $w_{21}$, and  $w_{12}$ and $w_{22}$.
    Since these are well-typed from the typing rule  \ruleref{Pair},
    applying I.H. using the corresponding typing judgments of the terms yields the required proof.
    
  \item[Where:] We have $\bracket{w_1}{w_2}= \bracket{\where{w'_1}{v_1}}{\where{w'_2}{v_2}}$.
    From the definition of \observe (Definition~\ref{fig:observe}), we have that $\observef{\where{w_i}{v_i}}{\Pi}{p}$ =  $\observef{w_i}{\Pi}{p}$. We have that the type of $w_i$ (i.e., $\tau$) is protected.
    \medskip

    Let the root value of $w_i$ be $v'_i$. Then from the definition of \whereroot{w_i} and Lemma~\ref{lemma:whereobserve}, we have that
    \[
    \observef{w_i}{\Pi}{p} = \observef{v'_i}{\Pi}{p}
    \]
    Also, the type of $v'_i$ is protected w.r.t $\Pi$ (given from \ruleref{Bracket-Values}). Note that applying \observe on $v'_i$ yields a hold. The argument uses induction on the structure of values.
    \medskip
    
    We now proceed with the case analysis on the structure of the values. That is,  $v'_i$ is either $()$, $\delexp{r}{t}$, $\lamc{x}{\tau_1}{\pc'}{e'}$, $\tlam{X}{\pc'}{e'}$,or $\pair{w'_{i1}}{w'_{i2}}$. Applying  \observe function to the all the values except the pair, yields a hole.

    For the pair, we invoke the inductive argument similar to the one proved above. We only sketch the high-level proof: the root values for  $w'_{i1}$ and $w'_{i2}$  are protected and are thus erased yielding a hole.

    \medskip
    Since     $    \observef{\where{w_i}{v_i}}{\Pi}{p} =    \observef{w_i}{\Pi}{p}$ $= \observef{v'_i}{\Pi}{p} = \circ$, we have the required proof that
    $\observef{\where{w_1}{v_1}}{\Pi}{p} = \observef{\where{w_i}{v_2}}{\Pi}{p}$.
   \end{description}
 Hence proved. 
\end{proof}

The following three lemmas are useful for establishing the properties of "where" propagation. 

\begin{lemma}[Correctness of Where propagation]\label{lemma:correctwhere}
  Let \TVal{\Pi;\G, x:\tau';\pc}{e}{\tau} and  \TVal{\Pi;\G;\pc}{w}{\tau'} such that \protjudge{\Pi}{H^\rightarrow  \wedge \top^{\leftarrow}}{\tau'} and $\delexp{r}{t} \not \in e$.
  If $e[x \mapsto w]$ $\stepsone$ $\where{w'}{\delexp{r}{t}}$, then \protjudge{\Pi}{H^\rightarrow \wedge \top^{\leftarrow}}{\tau}.
\end{lemma}
\begin{proof}
  Proof is by induction on the structure of $e$.
  \begin{description}

  \item[Var:] Not a valid expression, since $x$ is substituted with a value, it cannot take a step. However, note that \TVal{\Pi;\G, x:\tau';\pc}{x}{\tau} implies $\tau = \tau'$ and thus we have \protjudge{\Pi}{H^\rightarrow}{\tau^\rightarrow}.
  \item[App:] Given $e[x \mapsto \where{w}{\delexp{r}{t}}]$ $\stepsone$ $\where{w'}{\delexp{r}{t}}$ such that $e = e_1~e_2$. From the relevant rules  \ruleref{E-Eval}, \ruleref{E-App} and \ruleref{W-App}, only \ruleref{W-App} is possible. It implies that $e_1 = x = \where{w''}{\delexp{r}{t}}$ such that $\tau' = \func{\tau_1}{\pc'}{\tau_2}$ and $\tau = \tau_2$. Since $\protjudge{\Pi}{H^\rightarrow}{\tau'}$, from \ruleref{P-fun}, we have $\protjudge{\Pi}{H^\rightarrow}{\tau_2}$. Hence proved.
    \item[Proj:] Given $e[x \mapsto \where{w}{\delexp{r}{t}}]$ $\stepsone$ $\where{w'}{\delexp{r}{t}}$ such that $e = \proji{e'}$. From the relevant rules \ruleref{E-Eval}, \ruleref{E-Proj} and \ruleref{W-Proj}, only \ruleref{W-Proj} is relevant.  Thus $e' = x$ and $ e = \proji{x}$ and $e'[x \mapsto {\where{w'}{\delexp{r}{t}}}]  = \proji{\where{w'}{\delexp{r}{t}}}$ such that $\tau' = \tau_1 \times \tau_2$ and $\tau = \tau_i$ for $i \in \{1, 2\}$. Since $\protjudge{\Pi}{H^\rightarrow}{\tau'}$, from rule \ruleref{P-Prod}, we have that Since $\protjudge{\Pi}{H^\rightarrow}{\tau_i}$. Hence proved.
    \item[Case:] Given $e[x \mapsto \where{w}{\delexp{r}{t}}]$ $\stepsone$ $\where{w'}{\delexp{r}{t}}$ such that such that      $$ e = \casexp{e'}{y}{e_1}{e_2}$$
      From the relevant rules \ruleref{E-Eval}, \ruleref{E-Case} and \ruleref{W-Case}, only \ruleref{W-Case} is relevant.
      Thus $$e = (\casexp{x}{y}{e_1}{e_2}) $$ and $$e [x \mapsto {\where{w'}{\delexp{r}{t}}}] = (\casexp{(\where{w''}{\delexp{r}{t}})}{y}{e_1}{e_2})$$
      Not a valid case since $x$ has to be of sum type, and sum type cannot be protected.
    \item[BindM:] Given $e[x \mapsto \where{w}{\delexp{r}{t}}]$ $\stepsone$ $\where{w'}{\delexp{r}{t}}$ such that such that $e = (\bind{y}{e_1}{e_2}) $.
      From the relevant rules \ruleref{E-Eval}, \ruleref{E-BindM} and \ruleref{W-BindM}, only \ruleref{W-BindM} is relevant.
   Thus   $e = (\bind{y}{x}{e_2}) $ and $e [x \mapsto {\where{w}{\delexp{r}{t}}}] = (\bind{y}{(\where{w}{\delexp{r}{t}})}{e_2})$. Without loss of generality, assume that $\tau' = \says{\ell}{\tau''}$. Then, from \ruleref{BindM}, we have that $\protjudge{\Pi}{\pc \sqcup \ell}{\tau}$. Since $\protjudge{\Pi}{H^\rightarrow \wedge \top^{\leftarrow}}{\tau'}$ it implies that  $\protjudge{\Pi}{H^\rightarrow \wedge \top^{\leftarrow}}{\tau}$. Hence proved.
    \item [Assume:] Given $e[x \mapsto \where{w}{\delexp{r}{t}}]$ $\stepsone$ $\where{w'}{\delexp{r}{t}}$ such that such that  $e = (\assume{e_1}{e_2}) $. From the relevant rules \ruleref{E-Eval}, \ruleref{E-Assume} and \ruleref{W-Assume}, only \ruleref{W-Assume} is relevant. Thus $e = (\assume{x}{e_1}) $ and $e [x \mapsto {\where{w'}{\delexp{r}{t}}}] = (\assume{(\where{w''}{\delexp{r}{t}})}{e_1})$. From \ruleref{Assume} we have $\tau = \tau'$. Hence proved.
  \item[Other:] Not valid cases.
  \end{description}

\end{proof}

\begin{lemma}[Protected Where terms]\label{lemma:protectwhere}
  Let \TVal{\Pi;\G, x:\tau';\pc}{e_0}{\tau} and \TVal{\Pi;\G;\pc}{\where{w}{v}}{\tau'} such that \protjudge{\Pi}{H^\rightarrow  \wedge \top^{\leftarrow}}{\tau'},  $v=\delexp{r}{t}$, $v \not \in e_0$ and $\notrafjudge{\Pi}{\pc}{\voice{t}}$.
  If $e_0[x \mapsto \where{w}{v}] \stepsto \where{w'}{v}$, then \protjudge*{H^{π}}{τ^{π}}.
\end{lemma}
\begin{proof}
  Without loss of generality, let
  \begin{enumerate}
  \item  \label{pw:1} $e_0[x \mapsto \where{w}{v}]$ $\stepsto E[e]$ such that $e = T[\where{w}{v}]$
  \item \label{pw:2} $E[e] \stepsone$ $E[\where{e'}{v}]$
  \item  \label{pw:3} $ E[\where{e'}{v}] \stepsto E'[\where{w'}{v}]$
  \item  \label{pw:4} $E'[\where{w'}{v}] \stepsone \where{w'}{v}$
  \end{enumerate}

  Consider item~\ref{pw:1}. Since $e₀$ and $\where{w}{v}$ are well typed, by variable substitution (Lemma~\ref{lemma:vsubst}) so is $e_0[x \mapsto \where{w}{v}]$.
  Then, by subject reduction (Lemma~\ref{lemma:subjred}), we know that $E[e]$ is well-typed.
  Therefore, $e$ is also well-typed (for some context and type), thus \TVal{\Pi';\G';\pc'}{e}{\tau''}.
  \medskip

  Consider item~\ref{pw:2}.
  We are given   $\protjudge{\Pi}{H^{π}}{\tau'^{π}}$, and so
    \protjudge{\Pi'}{H^\rightarrow \wedge \top^{\leftarrow}}{\tau''}
  Also, applying subject reduction(Lemma~\ref{lemma:subjred}), we have   \TVal{\Pi';\G';\pc'}{\where{e'}{v}}{\tau''}.
  Applying Lemma ~\ref{lemma:correctwhere}, we have that \TVal{\Pi';\G';\pc'}{\where{e'}{v}}{\tau''} such that
  \protjudge{\Pi'}{H^\rightarrow \wedge \top^{\leftarrow}}{\tau''}. (Note that $\Pi \subseteq \Pi'$.)
  \medskip

  Consider item~\ref{pw:3}.
  From subject reduction (Lemma~\ref{lemma:subjred}),
  we have that
  \TVal{\Pi'';\G'';\pc''}{\where{w'}{v}}{\tau'''}.
  \medskip

  Consider item~\ref{pw:4}.
  From subject reduction (Lemma~\ref{lemma:subjred}), we have that
  \TVal{\Pi;\G;\pc}{\where{w'}{v}}{\tau}.
  Thus $\Pi'' = \Pi, \G'' = \G, \pc'' = \pc'$.
  Also, $E' = [\cdot]$ or $E' = \octx{\pc''}{[\cdot]}$. This is because, there are no more W-* steps.
  \medskip
  
  Consider item~\ref{pw:3} again.
  We have that
  \TVal{\Pi;\G;\pc}{\where{w'}{v}}{\tau'}.
  Repeatedly applying Lemma~\ref{lemma:correctwhere}, we have that 
    \protjudge{\Pi}{H^\rightarrow \wedge \top^{\leftarrow}}{\tau'}.
    Hence proved.
\end{proof}

However, what happens if the type of $\where{w}{v}$ (from $S$) is not protected? The following lemma states that even in such cases, propagation of $v$ ensures that the type of wrapped term is protected. The key insight is that $\where{w}{v}$ must be a subterm of a protected expression, and operating on protected expressions occurs only in protected contexts.

\begin{lemma}[Protected Where terms-2]\label{lemma:protectwhere2}
  Let \TVal{\Pi;\G, x:\tau';\pc}{e_0}{\tau} and \TVal{\Pi;\G;\pc}{U[\where{w}{v}]}{\tau'} such that the following hold.
  \begin{enumerate}
  \item $\tau'$ is protected. That is, $\protjudge{\Pi}{H^\rightarrow  \wedge \top^{\leftarrow}}{\tau'}$
  \item the type of $\where{w}{v}$ is not protected. That is, $ \TVal{\Pi';\G';\pc'}{\where{w}{v}}{\tau''}$ such that \protjudge{\Pi'}{H^\rightarrow  \wedge \top^{\leftarrow}}{\tau''}, for some $\Pi', \G'$ and $\pc'$.
  \item $v=\delexp{r}{t}$ for some $r$ and $t$, and $v \not \in e_0$
  \item $\notrafjudge{\Pi}{\pc}{\voice{t}}$
  \end{enumerate}
  If $e_0[x \mapsto U[\where{w}{v}]] \stepsto \where{w'}{v}$, then \protjudge{\Pi}{H^\rightarrow \wedge \top^{\leftarrow}}{\tau}.
\end{lemma}
\begin{proof}
  Since the type of $\where{w}{v}$ is unprotected but is a subterm of a term with protected type,
  the only possibility is $U[\where{w}{v}] = \return{\ell}{U'[\where{w}{v}]}$ such that $\ell$ protects $H^\rightarrow \wedge \top^{\leftarrow}$ or  $U[\where{w}{v}] = \lamc{y}{\tau_1}{\ell}{U'[\where{w}{v}]}$ such that $\func{\tau_1}{\ell}{\tau_2}$ protects  $H^\rightarrow \wedge \top^{\leftarrow}$.  Then, we have the following two possibilities.
  \begin{enumerate}
    \item
  \begin{enumerate}
  \item  \label{pw2:1} $e_0[x \mapsto U[\where{w}{v}]]$ $\stepsto E[\bind{y}{\return{\ell}{U'[\where{w}{v}]}}{e}]$
  \item  \label{pw2:2} $E[\bind{y}{\return{\ell}{U'[\where{w}{v}]}}{e'}]$ $\stepsone E[\octx{\ell}{e'[y \mapsto U'[\where{w}{v}]]}]$
  \item \label{pw2:3} $ E[\octx{\ell}{e'[y \mapsto U'[\where{w}{v}]]}] \stepsto$ $E[\octx{\ell}{\where{w''}{v}}]$
  \item \label{pw2:4}  $E[\octx{\ell}{\where{w''}{v}}] \stepsone E[\where{w''}{v}]$ 
  \item  \label{pw2:5} $ E[\where{w''}{v}] \stepsto E'[\where{w'}{v}]$
  \item  \label{pw2:6} $E'[\where{w'}{v}] \stepsone \where{w'}{v}$
  \end{enumerate}

  Consider item \ref{pw2:2}. From the typing rule \ruleref{BindM}, we have that the type of $e'[y \mapsto U'[\where{w}{v}]]$ protects $\ell$, and thus protects  $H^\rightarrow \wedge \top^{\leftarrow}$.
  
\medskip
Consider ~\ref{pw2:3}. Applying subject reduction (Lemma~\ref{lemma:subjred}), we thus have that the type of $\where{w''}{v}$ is protected.

\medskip
Consider steps from \ref{pw2:4} to \ref{pw2:6}. Invoking Lemma~\ref{lemma:protectwhere}, we have that the type of $\where{w'}{v}$ is protected.

\item
  \begin{enumerate}
  \item  \label{pw3:1} $e_0[x \mapsto U[\where{w}{v}]]$ $\stepsto E[\lamc{y}{\tau_1}{\ell}{U'[\where{w}{v}]}~w_2]$
  \item  \label{pw3:2} $E[\lamc{y}{\tau_1}{\ell}{U'[\where{w}{v}]}~w_2]$ $\stepsone E[\octx{\ell}{U'[\where{w}{v}][y \mapsto w_2]}]$
  \item \label{pw3:3} $ E[\octx{\ell}{U'[\where{w}{v}][y \mapsto w_2]}] \stepsto$ $E[\octx{\ell}{\where{w''}{v}}]$
  \item \label{pw3:4}  $E[\octx{\ell}{\where{w''}{v}}] \stepsone E[\where{w''}{v}]$ 
  \item  \label{pw3:5} $ E[\where{w''}{v}] \stepsto E'[\where{w'}{v}]$
  \item  \label{pw3:6} $E'[\where{w'}{v}] \stepsone \where{w'}{v}$
  \end{enumerate}

  Consider item \ref{pw3:1}.
  From the typing rule \ruleref{App}, we have that the type of $U'[\where{w}{v}]$ (in which $y$ is free) is $\tau_2$. From variable substitution (Lemma~\ref{lemma:vsubst}), the type of  $U'[\where{w}{v}][y \mapsto w_2]$ is $\tau_2$.

  \medskip
  Consider item \ref{pw3:2}.
  We are given that $\func{\tau_1}{\ell}{\tau_2}$ protects  $H^\rightarrow \wedge \top^{\leftarrow}$.
  From \ruleref{P-Fun},  this implies $\tau_2$ protects $H^\rightarrow \wedge \top^{\leftarrow}$.
  That is, the type of   $U'[\where{w}{v}][y \mapsto w_2]$ protects $H^\rightarrow \wedge \top^{\leftarrow}$.
  
 \medskip
 Consider ~\ref{pw3:3}. Applying subject reduction (Lemma~\ref{lemma:subjred}) to the evaluation steps inside the evaluation context $E$, we thus have that the type of $\where{w''}{v}$ protects $H^\rightarrow \wedge \top^{\leftarrow}$.

 \medskip
 Consider steps from \ref{pw3:4} to \ref{pw3:6}. Invoking Lemma~\ref{lemma:protectwhere}, we have that the type of $\where{w'}{v}$ protects $H^\rightarrow \wedge \top^{\leftarrow}$.

\end{enumerate}
  
Hence proved.
\end{proof}

\begin{lemma}[$\observe$ is preserved under monotonic $\Pi$] \label{lemma:observefew}
  If  $\observef{w_1}{Π,v}{p} =  \observef{w_2}{Π,v}{p}$ then  $\observef{w_1}{Π}{p} =  \observefc{w_2}{Π}{p}$.
\end{lemma}
\begin{proof}[Proof Sketch]
  Induction on the cases of $\observe$.
  Attacker $p$ cannot enable more information flows under $\Pi$ than under $\Pi, v$. Thus, from the definition of $\observe$, in each case terms erased to holes under $\Pi, v$ still get erased to holes under $\Pi$.
\end{proof}

We now present the proof of erasure conservation for confidentiality (Lemma~\ref{lemma:erasecons}):  evaluation step preserves erasure with respect to the confidentiality projections of an observer. An analogous lemma holds for integrity.

\erasecons*
\begin{proof}
By induction on the evaluation of $e$, using the definition of $\observe$ (Figure~\ref{fig:observe}) and delegation invariance (Lemma~\ref{lemma:di}).
  \begin{description}
  \item[Case \ruleref{E-App*}:] We have $\observefc{\outproj{\lamc{x}{τ}{\pc'}{e'}~w}{1}}{Π}{ℓ^{→}} =\observefc{\outproj{\lamc{x}{τ}{\pc'}{e'}~w}{2}}{Π}{ℓ^{→}}$. We have to show that
      $\observefc{\outproj{\octx {pc'}{\subst{e'}{x}{w}}}{1}}{Π}{ℓ^{→}} =\observefc{\outproj{\octx {pc'}{\subst{e'}{x}{w}}}{2}}{Π}{ℓ^{→}}$.
    Depending on the observability of $\pc'$, we have two subcases.
    \begin{description}
      \item[Case  $\rafjudge{Π}{ℓ^{→}}{\pc'^{→}}$:]
      We have
      $\observefc{\outproj{e'}{1}}{Π}{ℓ^{→}} =\observefc{\outproj{e'}{2}}{Π}{ℓ^{→}}$ and
      $\observefc{\outproj{w}{1}}{Π}{ℓ^{→}} =\observefc{\outproj{w}{2}}{Π}{ℓ^{→}}$
      Since substitution under context preserves observation function (Lemma~\ref{lemma:substerase}),
      we thus have $\observefc{\outproj{\octx {pc'}{\subst{e'}{x}{w}}}{1}}{Π}{ℓ^{→}} =\observefc{\outproj{\octx {pc'}{\subst{e'}{x}{w}}}{2}}{Π}{ℓ^{→}}$.
      \item[Case  $\notrafjudge{Π}{ℓ^{→}}{\pc'^{→}}$:]
       We have to show that if $\observefc{\outproj{\lamc{x}{τ}{\pc'}{e'}}{i}}{Π}{ℓ^{→}} = \circ$
        then $\observefc{\outproj{\octx {pc'}{\subst{e'}{x}{w}}}{i}}{Π}{ℓ^{→}} = \circ$.  This holds by the definition of $\observe$ since 
        $\observefc{\outproj{\lamc{x}{τ}{\pc'}{e'}}{i}}{Π}{ℓ^{→}} = \circ$ implies that $\notrafjudge{Π}{ℓ^{→}}{\pc'^{→}}$.
        \end{description}
 \item[Case \ruleref{E-TApp*}:] Similar to the \ruleref{E-App*} case but using Lemma~\ref{lemma:typesubsterase}.
 \item[Case \ruleref{E-BindM*}:]
   Given  $\bind{x}{\returng{ℓ'}{w}}{e'} \stepsone \octx{\ell'}{\subst{e'}{x}{w}}$.
   Also given, 
   $$\observefc{\outproj{\bind{x}{\returng{ℓ'}{w}}{e'}}{1}}{Π}{ℓ^{→}} =\observefc{\outproj{\bind{x}{\returng{ℓ'}{w}}{e'}}{2}}{Π}{ℓ^{→}}$$
   \begin{description}
   \item[Case $\rafjudge{Π}{ℓ^{→}}{ℓ'^{→}}$:]
     We have
     $\observefc{\outproj{e'}{1}}{Π}{ℓ^{→}} =\observefc{\outproj{e'}{2}}{Π}{ℓ^{→}}$ and
     $\observefc{\outproj{w}{1}}{Π}{ℓ^{→}} =\observefc{\outproj{w}{2}}{Π}{ℓ^{→}}$.
                Since substitution under context preserves observation function (Lemma~\ref{lemma:substerase}),
   we thus have $\observefc{\outproj{\octx{\ell'}{\subst{e'}{x}{w}}}{1}}{Π}{ℓ^{→}} =\observefc{\outproj{\octx{\ell'}{\subst{e'}{x}{w}}}{2}}{Π}{ℓ^{→}}$.
   \item[Case $\notrafjudge{Π}{ℓ^{→}}{ℓ'^{→}}$:]
      It remains to show that if $\observefc{\outproj{\returng{ℓ'}{w}}{i}}{Π}{ℓ^{→}} = \circ$ then $\observefc{\outproj{\octx {\ell'}{\subst{e'}{x}{w}}}{i}}{Π}{ℓ^{→}} = \circ$. 
      This holds by the definition of $\observe$ since 
      $\observefc{\outproj{\returng{ℓ'}{w}}{i}}{Π}{ℓ^{→}} = \circ$ implies that $\notrafjudge{Π}{ℓ^{→}}{ℓ'^{→}}$.
   \end{description}
   
 \item[Case \ruleref{O-Ctx}:]  Given $\octx{ℓ'}{w} \stepsone w$.
   Also given $\observefc{\outproj{\octx{ℓ'}{w}}{1}}{Π}{ℓ^{→}} = \observefc{\outproj{\octx{ℓ'}{w}}{2}}{Π}{ℓ^{→}}$. We have to prove that
   \[
   \observefc{\outproj{w}{1}}{Π}{ℓ^{→}} = \observefc{\outproj{w}{2}}{Π}{ℓ^{→}}
   \]
   We have two cases: either $\observefc{\outproj{\octx{ℓ'}{w}}{i}}{Π}{ℓ^{→}} \ne \circ$ or $\observefc{\outproj{\octx{ℓ'}{w}}{i}}{Π}{ℓ^{→}} = \circ$.
      If $\observefc{\outproj{\octx{ℓ'}{w}}{i}}{Π}{ℓ^{→}} \ne \circ$, then the proof is straightforward (since the \observe function is invoked on same term $w$).
      If $\observefc{\outproj{\octx{ℓ'}{w}}{i}}{Π}{ℓ^{→}} = \circ$,
      we have to show that $$\observefc{\outproj{w}{1}}{Π}{ℓ^{→}} = \observefc{\outproj{w}{2}}{Π}{ℓ^{→}}$$

   By the typing rule \ruleref{Ctx}, we know \TValGpc{w}{τ}. 
   If $\outproj{\octx {ℓ'}{w}}{1}  ≠\outproj{\octx {ℓ'}{w}}{2}$, then $w$ contains a bracket subexpression. That is either  $w = \bracket{\where{w_1}{v_1}}{\where{w_2}{v_2}}$ or $ w = \where{\bracket{w_1}{w_2}}{v}$. Note that $v$ cannot be bracketed as per \ruleref{B-Assume}.
   Consider the former case where $w = \bracket{\where{w_1}{v_1}}{\where{w_2}{v_2}}$. From \ruleref{Bracket-Values}
   \begin{align}
           \protjudge{\Pi}{H^\rightarrow}{\tau^\rightarrow} \label{eq:ec0}\\
           \TVal{\Pi;\G;\pc}{\where{w_i}{v_i}}{\tau} \label{eq:ec1} \\
           \TVal{\Pi;\G;\pc}{v_i}{\delexp{r}{t}} 
   \end{align}

        From  \eqref{eq:ec0}, we have that $\tau$ is protected.
        Invoking Lemma~\ref{lemma:correctobserve} (erasure on protected expressions) on $w$, we thus have
       $\observefc{\outproj{w}{1}}{Π}{ℓ^{→}} =  \observefc{\outproj{w}{2}}{Π}{ℓ^{→}}$.
        \medskip
        
        In the latter case, we have that $ w = \where{\bracket{w_1}{w_2}}{v}$. From the well-typedness of \ruleref{Where}, we have
        \begin{align}
        \TVal{\Pi,v;\G;\pc}{\bracket{w_1}{w_2}}{\tau} \label{eq:ec00}
        \end{align}
        From \ruleref{Bracket-Values}
        \begin{align}
                \protjudge{\Pi,v}{H^\rightarrow}{\tau^\rightarrow} \label{eq:ec10}\\
                \TVal{\Pi,v;\G;\pc}{w_i}{\tau} \label{eq:ec11}
        \end{align}
        From the soundness of the bracketed semantics (Lemma~\ref{lemma:sound}), we have that     
        $\octx{ℓ'}{w} \stepsone w$ implies  $\outproj{\octx{ℓ'}{w}}{i} \stepsone \outproj{w}{i}$ such that $\outproj{w}{i} = \where{w_i}{v}$.
        Invoking  delegation compartmentalization (Lemma~\ref{lemma:dc}) on the step $\outproj{\octx{ℓ'}{w}}{i} \stepsone$ $\where{w_i}{v}$,   we have two cases:

        \begin{description}
        \item[Subcase  \rafjudge{\Pi}{\pc}{\voice{t}}:] Invoking delegation invariance (Lemma~\ref{lemma:di}), we have that \notrafjudge{\Pi}{\ell^\rightarrow}{H^\rightarrow} and
          \notrafjudge{\Pi}{\pc}{\voice{\ell^\rightarrow -  H^\rightarrow}} implies the following. \footnote{Contrapositive of Lemma~\ref{lemma:di}.}
          \[
          \notrafjudge{\Pi,v}{\ell^\rightarrow}{H^\rightarrow}
          \]
          Invoking Lemma~\ref{lemma:correctobserve} (Correctness of \observe) on \eqref{eq:ec11} and the above judgment, we have that    $\observefc{w_1}{\{Π,v\}}{ℓ^{→}} =  \observefc{w_2}{\{Π,v\}}{ℓ^{→}}$.
          Invoking Lemma~\ref{lemma:observefew} that preserves $\observe$  under a restricted $\Pi$, we  have the required  $\observefc{w_1}{Π}{ℓ^{→}} =  \observefc{w_2}{Π}{ℓ^{→}}$. Hence proved.
          
          \item[Subcase $\protjudge{\Pi}{\confid{H}}{\tau^\rightarrow}$:]
            We have that  $\protjudge{\Pi}{\confid{H}}{\tau^\rightarrow}$.
            Using $\protjudge{\Pi}{\confid{H}}{\tau^\rightarrow}$ and \notrflowjudge{\Pi}{H^\rightarrow}{\ell^\rightarrow},   invoke Lemma~\ref{lemma:correctobserve} (correctness of \observe)
            on \eqref{eq:ec11} (with $\Pi' = \Pi, v$), yields the required $\observefc{w_1}{Π}{ℓ^{→}} =  \observefc{w_2}{Π}{ℓ^{→}}$.
             Hence proved.
        \end{description}
  \item[Case \ruleref{Assume}:]
    Given $\assume{\delexp{p}{q}}{e} \stepsone \where{e}{\delexp{p}{q}}$. Also,
    given that $\observefc{\outproj{\assume{\delexp{p}{q}}{e}}{1}}{Π}{ℓ^{→}} = \observefc{\outproj{\assume{\delexp{p}{q}}{e}}{2}}{Π}{ℓ^{→}}$. We have to prove that
   \[
   \observefc{\outproj{\where{e}{\delexp{p}{q}}}{1}}{Π}{ℓ^{→}} =    \observefc{\outproj{\where{e}{\delexp{p}{q}}}{2}}{Π}{ℓ^{→}}
   \]
   From the given conditions and definition of \observe, we already have
   $\observefc{\outproj{e}{1}}{Π}{ℓ^{→}} =    \observefc{\outproj{e}{2}}{Π}{ℓ^{→}}$.
   Hence     $\observefc{\outproj{\where{e}{\delexp{p}{q}}}{1}}{Π}{ℓ^{→}} =    \observefc{\outproj{\where{e}{\delexp{p}{q}}}{2}}{Π}{ℓ^{→}}$.
   
    \item[Case \ruleref{W-App}:] %
    Given $(\where{w}{v}){e} \stepsone \where{w~e}{v}$. Also
    given that $\observefc{\outproj{(\where{w}{v}){e}}{1}}{Π}{ℓ^{→}} = \observefc{\outproj{(\where{w}{v}){e}}{2}}{Π}{ℓ^{→}}$. We have to prove that
   \[
   \observefc{\outproj{\where{w~e}{v}}{1}}{Π}{ℓ^{→}} =    \observefc{\outproj{\where{w~e}{v}}{2}}{Π}{ℓ^{→}}
   \]
  Note that by virtue of \ruleref{B-Assume}, $v$ cannot be a bracket value. Thus
   $\outproj{v}{1} =   \outproj{v}{2}$.
  From the given conditions, we have
  \begin{align}
   \observefc{\outproj{\where{w}{v}}{1}}{Π}{ℓ^{→}} &=    \observefc{\outproj{\where{w}{v}}{2}}{Π}{ℓ^{→}} \label{eq:erwapp1}\\
   \observefc{\outproj{e}{1}}{Π}{ℓ^{→}} &=    \observefc{\outproj{e}{2}}{Π}{ℓ^{→}} \label{eq:erwapp2}
  \end{align}
   Expanding the \observe function in \eqref{eq:erwapp1}, we have
   \begin{equation}
   \where{\observefc{\outproj{w}{1}}{Π}{ℓ^{→}}}{\outproj{v}{1}} = \where{\observefc{\outproj{w}{2}}{Π}{ℓ^{→}}}{\outproj{v}{2}} \label{eq:erwapp3}
   \end{equation}

   Combining \eqref{eq:erwapp2} and  \eqref{eq:erwapp3} we have 
   \[
   \where{\observefc{\outproj{w~e}{1}}{Π}{ℓ^{→}}}{\outproj{v}{1}} = \where{\observefc{\outproj{w~e}{2}}{Π}{ℓ^{→}}}{\outproj{v}{2}}
   \]
   From the \observe function this is equal to
   \[
   \observefc{\outproj{\where{w~e}{v}}{1}}{Π}{ℓ^{→}} =    \observefc{\outproj{\where{w~e}{v}}{2}}{Π}{ℓ^{→}}
   \]
   Hence proved.
 \item[Case \ruleref{W-TApp}:]  Similar to \ruleref{W-App} case. 
 \item[Case \ruleref{W-UnPair}:]  Similar to \ruleref{W-App} case. 
 \item[Case \ruleref{W-Case}:]  Similar to \ruleref{W-App} case. 
 \item[Case \ruleref{W-BindM}:]  Similar to \ruleref{W-App} case. 
 \item[Case \ruleref{W-Assume}:]  Similar to \ruleref{W-App} case. 
 \item[Case \ruleref{B-Step}:] Straightforward application of induction hypothesis.
 \item[Case \textsc{B-*}:] Observe that, with the exception of \ruleref{B-Step}, all \textsc{B-*} rules step from $e \stepsto e'$ 
                                         such that $\outproj{e}{i}=\outproj{e'}{i}$.  Therefore the lemma trivially holds. 
     
 \item[E-Unpair:] Trivial. 
 \item[E-Case:]  Trivial.
 \item[E-UnitM:] Trivial.
 \item[E-Eval:]
   Given the following:
   \begin{align}
   &  E[e] \stepsone E[e']  \label{eq:eceval1}\\
   &  \TValGpc{E[e]}{\tau}  \label{eq:eceval2} \\
   &  \observefc{\outproj{E[e]}{1}}{Π}{ℓ^{→}} = \observefc{\outproj{E[e]}{2}}{Π}{ℓ^{→}} \label{eq:eceval3}
   \end{align}
   We need to show
   \begin{equation}
     \observefc{\outproj{E[e']}{1}}{Π}{ℓ^{→}} = \observefc{\outproj{E[e']}{2}}{Π}{ℓ^{→}} \label{eq:eceval4}
   \end{equation}
   
   Applying   Lemma~\ref{lemma:contextsubsterase}  to \eqref{eq:eceval3}, we have 
    \begin{equation}
        \observef{\outproj{e}{1}}{Π}{ℓ^{→}} = \observef{\outproj{e}{2}}{Π}{ℓ^{→}} \label{eq:eceval5}
    \end{equation}
    Since we have \eqref{eq:eceval3}, we can apply induction hypothesis to the premise of \ruleref{E-Eval} yielding
    \begin{equation}
        \observef{\outproj{e'}{1}}{Π}{ℓ^{→}} = \observef{\outproj{e'}{2}}{Π}{ℓ^{→}} \label{eq:eceval6}
    \end{equation}
   Applying   Lemma~\eqref{lemma:contextsubsterase}  to \eqref{eq:eceval6}, we thus have \eqref{eq:eceval4}
  \end{description}
\end{proof}

\section{Proofs for Robust Declassification (Lemma~\ref{thm:robdecl})}
\label{app:robdecl}

\begin{lemma}[Substitution preserves $\observe$] \label{lemma:valueequiv}
 Let \TValP{\G, x:\tau', \G';\pc}{e}{\tau} and \TValP{\G;\pc}{v_i}{\tau'} for $i \in \{1, 2 \}$.
 If        $\subst{e}{x}{v_1} ≈^{Π}_{p^{\pi}}  \subst{e}{x}{v_2}$ then
   either $e$ does not have free occurrence of $x$ or $v_1 ≈^{Π}_{p^{\pi}} v_2$.
\end{lemma}
\begin{proof}
 Induction on the structure of $e$.
\end{proof}

\begin{lemma}[Voice of Integrity Principal] \label{lemma:voiceOfInteg} 
  \reqjudge{\Pi}{∇(ℓ^{←})}{ℓ^{←}}
\end{lemma}
\begin{proof}
  Proven in Coq~\cite{SAFproof}.
\end{proof}

\begin{lemma}[Duality of Voice and View] \label{lemma:voiceviewdual} 
  For all $\ell$, \view{\voice{\confid{\ell}}} $\equiv \confid{\ell}$.
\end{lemma}
\begin{proof}
  We have the following:
  \begin{align}
    \voice{\confid{\ell}} = \integ{\ell} \\
    \view{\voice{\confid{\ell}}} = \view{\integ{\ell}} = \confid{\ell} \label{eq:voiceviewdual1}
  \end{align}
       From \eqref{eq:voiceviewdual1}, we thus have \view{\voice{\confid{\ell}}} $\equiv \confid{\ell}$.
  \end{proof}

\robdecl*
\begin{proof}
  By induction on the evaluation of $e$. A note on the notation: whenever the context is clear, we elide the type annotation on the hole. 
  \begin{description}
  \item[Case \ruleref{E-App*}]
    Since $\vec{a_1}$ and $\vec{a_2}$ are $\Pi_{H}$-fair attacks, we have that the attacks are  well-typed.
    That is, for each $a_{rj} \in \vec{a_j}$,
       \begin{equation}
         \TVal{\Pi_H; Γ_{rj};\pc}{a_{rj}}{τ_{rj}} \text{ for } j \in \{1, 2 \}  \label{eq:rob22}
       \end{equation}
       Without loss of generality, assume that  $ e[\vec{a}_j] = (e'[a_{kj}] e''[a_{lj}])$.
       Given $e[\vec{a}_j][x \mapsto v_i] \stepsone  e'_{ij}$. That is, $(e'[a_{kj}] e''[a_{lj}])[x \mapsto v_i] \stepsone e'_{ij}$.
    Inverting the last condition, we  have the following.
    \begin{equation}
    \begin{array}{l  l  l}
      e'[\vec{a}_{kj}]   &\triangleq (\lamc{y}{τ''}{pc'}{e_{\lambda j}}) & \text{ for some } y, \tau'', pc' \text{ and } e_{\lambda j} \text{ corresponding to attack } j \in \{1, 2\}  \\
      e''[\vec{a}_{lj}] & \triangleq  v_{aj}  & \text{ for some value } v_{aj} \text{ corresponding to attack } j \in \{1, 2\}  \\
    \end{array} \label{eq:rob0}
    \end{equation}
     At a high-level, \eqref{eq:rob0} says that the attack is either an argument to the function or the function itself or both.
      We are also given that the terms are well-typed. From \ruleref{T-App}, we have
      \begin{align}
        \TValP{Γ,x\ty τ',Γ';\pc}{e'[\vec{a}_{k1}]~e''[\vec{a}_{l1}]}{τ} \label{eq:rob26} \\
        \TValP{Γ,x\ty τ',Γ';\pc}{e'[\vec{a}_{k2}]~e''[\vec{a}_{l2}]}{τ} \label{eq:rob27}
      \end{align}

      Since we are given that  $\TValP{Γ,x\ty τ',Γ';\pc}{e[\vec{•}^{\vec{τ}^{*}}]}{τ}$, we have
      \begin{align}
      \TValP{Γ,x\ty τ',Γ';\pc}{e'[\vec{\bullet}]}{\tau} \label{eq:rob20}\\
      \TValP{Γ,x\ty τ',Γ';\pc}{e''[\vec{\bullet}]}{\tau''}  \label{eq:rob21}
      \end{align}
      Note that we have elided the type annotations on holes to enhance readability.
      Substituting  holes in \eqref{eq:rob20} and \eqref{eq:rob21} with $\vec{a}_1$, we have
      \begin{align}
        \TValP{Γ,x\ty τ',Γ';\pc}{(\lamc{y}{τ''}{\pc'}{e_{\lambda 1}})~v_{a1}}{τ} \label{eq:rob6} \\
        \TValP{Γ,x\ty τ',Γ';\pc}{(\lamc{y}{τ''}{\pc'}{e_{\lambda 2}})~v_{a2} }{τ} \label{eq:rob7}
      \end{align}

      From \ruleref{E-App*}, we have the following step.
      \begin{align}
       \subst{\big ( (\lamc{y}{τ''}{\pc'}{e_{\lambda j}})~v_{aj}\big )}{x}{v_i} \stepsone \subst{\octx{pc'}{e_{\lambda j}[ y ↦ v_{aj}]}}{x}{vᵢ} \label{eq:rob2}
      \end{align}
      Using the above relations, we are now ready to prove the "if" direction.
      Assume that the left hand holds.
      $$
      \subst{e[\vec{a}_{1}]}{x}{v₁} · e'_{11}   ≈^{Π}_{Δ(H^{←})} \subst{e[\vec{a}_{1}]}{x}{v₂} ·  e'_{21}
      $$
      Substituting $e'_{11}$ and $e'_{21}$ with the result of step in \eqref{eq:rob2}  we have,
      \begin{align}
        & \subst{\big ( (\lamc{y}{τ''}{\pc'}{e_{\lambda 1}})~v_{a1} \big )}{x}{v_1}  \cdot  \subst{\octx{pc'}{\subst{e_{\lambda 1}}{y}{v_{a1}}}}{x}{v_1}  \nonumber \\
        & \qquad \qquad \qquad \qquad \qquad \qquad ≈^{Π}_{Δ(H^{←})} \nonumber  \\
        &  \subst{\big ( (\lamc{y}{τ''}{\pc'}{e_{\lambda 1}})~v_{a1}\big )}{x}{v_2} \cdot  \subst{\octx{pc'}{\subst{e_{\lambda 1}}{y}{v_{a1}}}}{x}{v_2} \label{eq:rob3}
      \end{align}
      Unfolding the definition of trace equivalence, we have
      \begin{align}
          \subst{\big ( (\lamc{y}{τ''}{\pc'}{e_{\lambda 1}})~v_{a1} \big )}{x}{v_1} &≈^{Π}_{Δ(H^{←})}  \subst{\big ( (\lamc{y}{τ''}{\pc'}{e_{\lambda 1}})~v_{a1}\big )}{x}{v_2}   \label{eq:rob4} \\
        \subst{\octx{pc'}{\subst{e_{\lambda 1}}{y}{v_{a1}}}}{x}{v_1}   &≈^{Π}_{Δ(H^{←})} \subst{\octx{pc'}{\subst{e_{\lambda 1}}{y}{v_{a1}}}}{x}{v_2}   \label{eq:rob5}
      \end{align}
      Without loss of generality, let  $\vec{a}_2 \triangleq \vec{a}_{k2} \cdot \vec{a}_{l2} $.
       Following \eqref{eq:rob0},  we have the following when the application is filled with second attack.
       \begin{align}
        e'[\vec{a}_{k2}]  &\triangleq \lamc{y}{τ''}{\pc'}{e_{\lambda 2}} \\
        e''[\vec{a}_{l2}] &\triangleq v_{a2} 
       \end{align}
       
       We have to prove the following:
       \begin{align}
          \subst{\big ( (\lamc{y}{τ''}{\pc'}{e_{\lambda 2}})~v_{a2} \big )}{x}{v_1} &≈^{Π}_{Δ(H^{←})}  \subst{\big ( (\lamc{y}{τ''}{\pc'}{e_{\lambda 2}})~v_{a2}\big )}{x}{v_2}     \label{eq:rob8}
      \end{align}

       From Proposition~\ref{prop:niatk}, we have that $\vec{a}_{k1}$ and $\vec{a}_{l1}$ do not have any free occurrences of the variable $x$. Thus $x$ only occurs in $e'[\bullet]~e''[\bullet]$.
             Let  $x$ be some arbitrary occurrence in $e'[\vec{a}_{k1}]~e''[\vec{a}_{l1}]$     such that     $\TVal{\Pi; \G', x\ty \tau', \G''; \pc}{e'[\vec{a}_{k1}]~e''[\vec{a}_{l1}]}{\tau}$. Since $x$ is substituted with $v_1$ and $v_2$ and substitution preserves observation equivalence (Lemma~\ref{lemma:valueequiv}), from \eqref{eq:rob4} we have
      \begin{equation}
      \observef{v_1}{\Pi}{Δ(H^{←})} = \observef{v_2}{\Pi}{Δ(H^{←})} \label{eq:rob29}
      \end{equation}
      The crucial insight here is that though $x$ may occur in a sub term that is well-typed under different delegation context $\Pi'$, however, the \observe function still uses the original $\Pi$ under which the program is well-typed. Hence \eqref{eq:rob29} uses the original delegation context $\Pi$, rather than $\Pi'$.

      \medskip
      Again, from Proposition~\ref{prop:niatk}, we have that $\vec{a}_2$ cannot contain $x$. That is,  $\vec{a}_{k2}$ and $\vec{a}_{l2}$ do not have free occurrences of $x$. Thus $x$ only occurs in $e'[\bullet]~e''[\bullet]$.
      Similar to the argument under $\vec{a}_1$, let  $x$ be some arbitrary occurrence in $e'[\vec{a}_{k2}]~e''[\vec{a}_{l2}]$     such that     $\TVal{\Pi; \G_2^{'}, x\ty \tau', \G_{2}^{''}; \pc''}{x}{\tau'}$. Again, $x$ is substituted with $v_1$ and $v_2$, to prove \eqref{eq:rob8}, it suffices to prove
      \begin{equation}
      \observef{v_1}{\Pi}{Δ(H^{←})} = \observef{v_2}{\Pi}{Δ(H^{←})} \label{eq:rob30}
      \end{equation}
      This follows trivially  from \eqref{eq:rob29}.
      \medskip
      We still have to prove that the step preserves erasure of terms:
           \[
      \subst{\octx{pc'}{\subst{e_{\lambda 2}}{y}{v_{a2}}}}{x}{v_1}  ≈^{Π}_{Δ(H^{←})}  \subst{\octx{pc'}{\subst{e_{\lambda 2}}{y}{v_{a2}}}}{x}{v_2}
      \]
      That is,
      \begin{equation}
      \observefc{\subst{\octx{pc'}{\subst{e_{\lambda 2}}{y}{v_{a2}}}}{x}{v_1}}{\Pi}{Δ(H^{←})} = \observefc{\subst{\octx{pc'}{\subst{e_{\lambda 2}}{y}{v_{a2}}}}{x}{v_2}}{\Pi}{Δ(H^{←})} \label{eq:rob35}
      \end{equation}
      
Applying type preservation (Lemma~\ref{lemma:subjred}) to \eqref{eq:rob6} and \eqref{eq:rob7}, we have that
      \begin{align}
        \TValP{Γ,x\ty τ',Γ';\pc}{\octx{pc'}{e_{\lambda 1} [y \mapsto v_{a1}] }}{τ} \label{eq:rob8} \\
        \TValP{Γ,x\ty τ',Γ';\pc}{\octx{pc'}{e_{\lambda 2}[ y \mapsto v_{a2}] }}{τ} \label{eq:rob9}
      \end{align}
      From variable substitution under a context (Lemma~\ref{lemma:ctxtvarsubst}), \eqref{eq:rob8} and \eqref{eq:rob9} lead to
      \begin{align}
        \TValP{Γ,Γ';\pc}{\octx{pc'}{\subst{e_{\lambda 2}[ y \mapsto v_{a2}] }{x}{v_1}}}{τ} \label{eq:rob37} \\
        \TValP{Γ,Γ';\pc}{\octx{pc'}{\subst{e_{\lambda 2}[ y \mapsto v_{a2}] }{x}{v_2}}}{τ} \label{eq:rob38}
      \end{align}
      It suffices to prove the following
      \begin{equation}
      \observef{v_1}{\Pi}{Δ(H^{←})} = \observef{v_2}{\Pi}{Δ(H^{←})} 
      \end{equation}
      This is already proven  in \eqref{eq:rob29}. Hence proved.
      \medskip
      
      The other "if" direction follows using analogous argument.
     
     \item[Case \ruleref{E-TApp*}:] Similar to \ruleref{E-App*}.
     \item[Case \ruleref{O-Ctx}:] Given $\octx{ℓ}{w[\vec{a}_{1}]}[x \mapsto v_1] \stepsone w_1[\vec{a}_{1}] $ and  $\octx{ℓ}{\vec{w[a_{1}}]}[x \mapsto v_2] \stepsone w_2[\vec{a}_{1}] $. We have to prove that
       \begin{align*}
         & \octx{\ell}{w[\vec{a}_{1}]}[x \mapsto v_1] \cdot w_1[\vec{a}_{1}]   ≈^{Π}_{Δ(H^{←})} \octx{\ell}{w[\vec{a}_{1}]}[x \mapsto v_2] \cdot w_2[\vec{a}_{1}] \\
         & \iff \\
         & \octx{\ell}{w[\vec{a}_{2}]}[x \mapsto v_1] \cdot w_1[\vec{a}_{2}]  ≈^{Π}_{Δ(H^{←})} \octx{\ell}{w[\vec{a}_{2}]}[x \mapsto v_2] \cdot w_2[\vec{a}_{2}]
       \end{align*}
       We will prove one direction of the implication. The other direction is analagous.
       We are given that  $\TValP{Γ,x\ty τ',Γ';\pc}{e[\vec{•}^{\vec{τ}^{*}}]}{τ}$, we 
      \begin{align}
      \TValP{Γ,x\ty τ',Γ';\pc}{\octx{\ell}{w[\vec{\bullet}]}}{\tau} \label{eq:octx1}
      \end{align}
      Substituting the hole in \eqref{eq:octx1}  with $\vec{a_1}$, we have
      \begin{align}
        \TValP{Γ,x\ty τ',Γ';\pc}{\octx{\ell}{w[\vec{a}_{1}]}}{τ} \label{eq:octx2}
      \end{align}

      Assume
      \[
      \octx{\ell}{w[\vec{a}_{1}]}[x \mapsto v_1] \cdot w_1[\vec{a}_{1}]   ≈^{Π}_{Δ(H^{←})} \octx{\ell}{w[\vec{a}_{1}]}[x \mapsto v_2] \cdot w_2[\vec{a}_{1}]
      \]
      From the erasure definition, this implies
      \begin{align}
        \octx{\ell}{w[\vec{a}_{1}]}[x \mapsto v_1]  &≈^{Π}_{Δ(H^{←})}   \octx{\ell}{w[\vec{a}_{1}]}[x \mapsto v_2]   \label{eq:octx3} \\
        w_1[\vec{a}_{1}] &≈^{Π}_{Δ(H^{←})}  w_2[\vec{a}_{1}]  \label{eq:octx4}
      \end{align}

      From Proposition~\ref{prop:niatk}, we have that $\vec{a}_2$ cannot contain $x$.
      Thus $x$ only occurs in $\octx{\ell}{w[\vec{•}]}$.
      Let  $x$ be some arbitrary occurrence in $\octx{\ell}{w[\vec{•}]}$     such that     $\TVal{\Pi; \G', x\ty \tau', \G''; \pc'}{\octx{\ell}{w[\vec{a}_2]}}{\tau'}$. Since $x$ is substituted with $v_1$ and $v_2$ and substitution preserves equivalence (Lemma~\ref{lemma:valueequiv}), from \eqref{eq:octx3} we have
      \begin{equation}
      \observef{v_1}{\Pi}{Δ(H^{←})} = \observef{v_2}{\Pi}{Δ(H^{←})} \label{eq:octx5}
      \end{equation}
      Again, from Proposition~\ref{prop:niatk}, we have that $\vec{a}_2$ cannot contain $x$.
      Thus $x$ only occurs in $\octx{\ell}{w[\vec{•}]}$.

       We need to prove,
      \[
      \octx{\ell}{w[\vec{a}_{2}]}[x \mapsto v_1] \cdot w_1[\vec{a}_{2}]   ≈^{Π}_{Δ(H^{←})} \octx{\ell}{w[\vec{a}_{2}]}[x \mapsto v_2] \cdot w_2[\vec{a}_{2}]
      \]
      That is, we need to prove the following.
      \begin{align}
        \octx{\ell}{w[\vec{a}_{2}]}[x \mapsto v_1]  &≈^{Π}_{Δ(H^{←})}   \octx{\ell}{w[\vec{a}_{2}]}[x \mapsto v_2]   \label{eq:octx6} \\
        w_1[\vec{a}_{2}] &≈^{Π}_{Δ(H^{←})}  w_2[\vec{a}_{2}]  \label{eq:octx7}
      \end{align}
      
      Similar to the argument under $\vec{a}_1$, let  $x$ be some arbitrary occurrence in $w[\vec{a}_{2}]$.
      Again, $x$ is substituted with $v_1$ and $v_2$, to prove \eqref{eq:octx6}, it suffices  to prove
      \begin{equation}
      \observef{v_1}{\Pi}{Δ(H^{←})} = \observef{v_2}{\Pi}{Δ(H^{←})}
      \end{equation}
      And, judgment \eqref{eq:octx5} already gives this.
      \medskip

      We will now focus on proving \eqref{eq:octx7}. Note that $w_1[\vec{\bullet}] = w[\vec{\bullet}][x \mapsto v_1]$ and $w_2[\vec{\bullet}] = w[\vec{\bullet}][x \mapsto v_2]$.
      From \eqref{eq:octx6} and the definition of erasure function for contexts, we have that 
      \[
              w[\vec{a}_{2}][x \mapsto v_1]  ≈^{Π}_{Δ(H^{←})} w[\vec{a}_{2}][x \mapsto v_2]
      \]
      Thus \eqref{eq:octx7} follows trivially from \eqref{eq:octx6}. Hence proved.

    \item[Case \ruleref{E-Case}:] Given $(\casexp{\injk{w}}{y}{e_1}{e_2})[\vec{a}_j][x \mapsto v_i] \stepsone e_k[\vec{a}_j][y \mapsto w]$ for $i,j,k \in \{1,2\}$. We have to prove that
      \begin{align*}
        & (\casexp{\injk{w}}{y}{e_1}{e_2})[\vec{a}_1][x \mapsto v_1] \cdot e_k[\vec{a}_j][x \mapsto v_1][y \mapsto w] \\
        & \qquad \qquad \qquad ≈^{Π}_{Δ(H^{←})} \\
        & (\casexp{\injk{w}}{y}{e_1}{e_2})[\vec{a}_1][x \mapsto v_2] \cdot e_k[\vec{a}_j][x \mapsto v_2][y \mapsto w] \\
        & \iff \\
        & (\casexp{\injk{w}}{y}{e_1}{e_2})[\vec{a}_2][x \mapsto v_1] \cdot e_k[\vec{a}_1][x \mapsto v_1][y \mapsto w] \\
        & \qquad \qquad \qquad ≈^{Π}_{Δ(H^{←})} \\
        & (\casexp{\injk{w}}{y}{e_1}{e_2})[\vec{a}_2][x \mapsto v_2] \cdot e_k[\vec{a}_2][x \mapsto v_2][y \mapsto w]
      \end{align*}
      We will prove only one direction. Assume 
      \begin{align*}
        & (\casexp{\injk{w}}{y}{e_1}{e_2})[\vec{a}_1][x \mapsto v_1] \cdot e_k[\vec{a}_1][x \mapsto v_1][y \mapsto w] \\
        & \qquad \qquad \qquad ≈^{Π}_{Δ(H^{←})} \\
        & (\casexp{\injk{w}}{y}{e_1}{e_2})[\vec{a}_2][x \mapsto v_2] \cdot e_k[\vec{a}_2][x \mapsto v_2][y \mapsto w]
      \end{align*}
        
        That implies
      \begin{align*}
        (\casexp{\injk{w}}{y}{e_1}{e_2})[\vec{a}_1][x \mapsto v_1] & \\
         ≈^{Π}_{Δ(H^{←})} & \\
        (\casexp{\injk{w}}{y}{e_1}{e_2})[\vec{a}_1][x \mapsto v_2] &
      \end{align*}
      and
      \begin{align*}
        e_k[\vec{a}_2][x \mapsto v_1][y \mapsto w] ≈^{Π}_{Δ(H^{←})} e_k[\vec{a}_2][x \mapsto v_2][y \mapsto w]
      \end{align*}

      We need to prove
      \begin{align}
        (\casexp{\injk{w}}{y}{e_1}{e_2})[\vec{a}_2][x \mapsto v_1]  \nonumber \\
        ≈^{Π}_{Δ(H^{←})} \nonumber \\ 
        (\casexp{\injk{w}}{y}{e_1}{e_2})[\vec{a}_2][x \mapsto v_2]  \label{eq:rdcase1}
      \end{align}
      and
      \begin{align}
        e_k[\vec{a}_2][x \mapsto v_1][y \mapsto w] &≈^{Π}_{Δ(H^{←})} e_k[\vec{a}_2][x \mapsto v_2][y \mapsto w]  \label{eq:rdcase2}
      \end{align}
      Let  $x$ be some arbitrary occurrence in $\casexp{\injk{w}}{y}{e_1}{e_2}[\vec{a}_1]$     such that     $\TVal{\Pi; \G', x\ty \tau', \G''; \pc'}{\casexp{\injk{w}}{y}{e_1}{e_2}[\vec{a}_1]}{\tau'}$. Since $x$ is substituted with $v_1$ and $v_2$ and substitution preserves observational equivalence (Lemma~\ref{lemma:valueequiv}), we have
      \begin{equation}
      \observef{v_1}{\Pi}{Δ(H^{←})} = \observef{v_2}{\Pi}{Δ(H^{←})} \label{eq:rdcase3}
      \end{equation}
      
      Again, from Proposition~\ref{prop:niatk}, we have that $\vec{a}_2$ cannot contain $x$.
      Thus $x$ only occurs in

      $\casexp{\injk{w}}{y}{e_1}{e_2})[\vec{\bullet}]$.
      Similar to the argument under $\vec{a}_1$, let  $x$ be some arbitrary occurrence in $\casexp{\injk{w}}{y}{e_1}{e_2}[\vec{a}_{2}]$     such that     $\TVal{\Pi; \G_2^{'}, x\ty \tau', \G_2^{''}; \pc}{\casexp{\injk{w}}{y}{e_1}{e_2}[\vec{a}_{2}]}{\tau'}$. Again, $x$ is substituted with $v_1$ and $v_2$.     
      Following an argument similar to previous cases, it suffice to  prove
      \begin{equation}
      \observef{v_1}{\Pi}{Δ(H^{←})} = \observef{v_2}{\Pi}{Δ(H^{←})}
      \end{equation}
      And, judgment \eqref{eq:rdcase3} already gives us the required proof.
      \medskip

      We now focus on proving \eqref{eq:rdcase2}. From \eqref{eq:rdcase1}, we have
      \begin{align}
        \injk{w}[\vec{a}_2][x \mapsto v_1] &≈^{Π}_{Δ(H^{←})}  \injk{w}[\vec{a}_2][x \mapsto v_2] \label{eq:rdcase5} \\
        e_1[\vec{a}_2][x \mapsto v_1] &≈^{Π}_{Δ(H^{←})}  e_1[\vec{a}_2][x \mapsto v_2] \label{eq:rdcase6} \\
        e_2[\vec{a}_2][x \mapsto v_1] &≈^{Π}_{Δ(H^{←})}  e_2[\vec{a}_2][x \mapsto v_2] \label{eq:rdcase7}
      \end{align}
      We are given that
      \[
      \TValP{\Gamma, x:\tau', \Gamma'; \pc}{(\casexp{\injk{w}}{y}{e_1}{e_2})[\vec{a}_1]}{\tau}
      \]
      Inverting the rule \ruleref{Case}, we have
      \[
      \TValP{\Gamma, x:\tau', \Gamma', y:\tau_1 + \tau_2; \pc}{(e_k)[\vec{a}_1]}{\tau}
      \]
      Since variable substitution is preserved under contexts (Lemma~\ref{lemma:ctxtvarsubst}), from \eqref{eq:rdcase5}, \eqref{eq:rdcase6}, we have \eqref{eq:rdcase2}. Hence proved.

    \item[Case \ruleref{E-BindM*}:]
      Given $(\bind{y}{\returng{\ell}{w}}{e_2})[\vec{a}_j][x \mapsto v_i]      \stepsone e_2[\vec{a}_j][y \mapsto w]$ for $i,j \in \{1,2\}$. We have to prove that
                 \begin{align*}
        & (\bind{y}{\returng{\ell}{w}}{e_2})[\vec{a}_1][x \mapsto v_1] \cdot e_2[\vec{a}_j][x \mapsto v_1][y \mapsto w] \\
        & \qquad \qquad \qquad ≈^{Π}_{Δ(H^{←})} \\
        & (\bind{y}{\returng{\ell}{w}}{e_2})[\vec{a}_1][x \mapsto v_2] \cdot e_2[\vec{a}_j][x \mapsto v_2][y \mapsto w] \\
        & \iff \\
        & (\bind{y}{\returng{\ell}{w}}{e_2})[\vec{a}_2][x \mapsto v_1] \cdot e_2[\vec{a}_j][x \mapsto v_1][y \mapsto w] \\
        & \qquad \qquad \qquad ≈^{Π}_{Δ(H^{←})} \\
        & (\bind{y}{\returng{\ell}{w}}{e_2})[\vec{a}_2][x \mapsto v_2] \cdot e_2[\vec{a}_j][x \mapsto v_2][y \mapsto w]
      \end{align*}
     The argument is similar to the above case.
     \item[Case \ruleref{E-Assume}:]
      Given $(\assume{\delexp{p}{q}}{e_2})[\vec{a}_j][x \mapsto v_i]      \stepsone (\where{e_2}{\delexp{p}{q}})[\vec{a}_j][x \mapsto v_i]$ for $i,j \in \{1,2\}$. We have to prove that
                 \begin{align*}
        & (\assume{\delexp{p}{q}}{e_2})[\vec{a}_1][x \mapsto v_1] \cdot (\where{e_2}{\delexp{p}{q}})[\vec{a}_j][x \mapsto v_1] \\
        & \qquad \qquad \qquad ≈^{Π}_{Δ(H^{←})} \\
        & (\assume{\delexp{p}{q}}{e_2})[\vec{a}_1][x \mapsto v_2] \cdot (\where{e_2}{\delexp{p}{q}})[\vec{a}_j][x \mapsto v_2] \\
        & \iff \\
        & (\assume{\delexp{p}{q}}{e_2})[\vec{a}_2][x \mapsto v_1] \cdot (\where{e_2}{\delexp{p}{q}})[\vec{a}_j][x \mapsto v_1] \\
        & \qquad \qquad \qquad ≈^{Π}_{Δ(H^{←})} \\
        & (\assume{\delexp{p}{q}}{e_2})[\vec{a}_2][x \mapsto v_2] \cdot (\where{e_2}{\delexp{p}{q}})[\vec{a}_j][x \mapsto v_2]
      \end{align*}
      The argument is similar to the previous cases. To prove in one direction, we assume the following.
        \begin{align*}
        & (\assume{\delexp{p}{q}}{e_2})[\vec{a}_1][x \mapsto v_1] \cdot (\where{e_2}{\delexp{p}{q}})[\vec{a}_1][x \mapsto v_1] \\
        & \qquad \qquad \qquad ≈^{Π}_{Δ(H^{←})} \\
          & (\assume{\delexp{p}{q}}{e_2})[\vec{a}_1][x \mapsto v_2] \cdot (\where{e_2}{\delexp{p}{q}})[\vec{a}_1][x \mapsto v_2]
        \end{align*}
        This gives us
        \begin{align}
         (\assume{\delexp{p}{q}}{e_2})[\vec{a}_1][x \mapsto v_1]  &≈^{Π}_{Δ(H^{←})} (\assume{\delexp{p}{q}}{e_2})[\vec{a}_1][x \mapsto v_2] \label{eq:rdass1} \\
          (\where{e_2}{\delexp{p}{q}})[\vec{a}_1][x \mapsto v_1]   &≈^{Π}_{Δ(H^{←})} (\where{e_2}{\delexp{p}{q}})[\vec{a}_1][x \mapsto v_2] \label{eq:rdass2}
        \end{align}
        Consider \eqref{eq:rdass1}.
      Let  $x$ be some arbitrary occurrence in $(\assume{\delexp{p}{q}}{e_2})[\vec{a}_1]$     such that     $\TVal{\Pi; \G', x\ty \tau', \G''; \pc'}{(\assume{\delexp{p}{q}}{e_2})[\vec{a}_1]}{\tau'}$. Since $x$ is substituted with $v_1$ and $v_2$ and substitution preserves observational equivalence (Lemma~\ref{lemma:valueequiv}), we have
      \begin{equation}
      \observef{v_1}{\Pi}{Δ(H^{←})} = \observef{v_2}{\Pi}{Δ(H^{←})} \label{eq:rdass3}
      \end{equation}
      
      Again, from Proposition~\ref{prop:niatk}, we have that $\vec{a}_2$ cannot contain $x$.
      Thus $x$ only occurs in $ (\assume{\delexp{p}{q}}{e_2})[\vec{\bullet}]$.
      Similar to the argument under $\vec{a}_1$, let  $x$ be some arbitrary occurrence in $(\assume{\delexp{p}{q}}{e_2})[\vec{a}_{2}]$     such that     $\TVal{\Pi; \G_2^{'}, x\ty \tau', \G_2^{''}; \pc''}{(\assume{\delexp{p}{q}}{e_2})[\vec{a}_2]}{\tau'}$. Again, $x$ is substituted with $v_1$ and $v_2$.     
      Following an argument similar to one in \ruleref{E-App*} case, it suffices prove the following
      \begin{equation}
      \observef{v_1}{\Pi}{Δ(H^{←})} = \observef{v_2}{\Pi}{Δ(H^{←})} \label{eq:rdass4}
      \end{equation}
      Judgment \eqref{eq:rdass3} already gives us the above judgment.
      Hence
      \[
         (\assume{\delexp{p}{q}}{e_2})[\vec{a}_2][x \mapsto v_1]  ≈^{Π}_{Δ(H^{←})} (\assume{\delexp{p}{q}}{e_2})[\vec{a}_2][x \mapsto v_2] 
      \]
      A similar argument applied to \eqref{eq:rdass2} gives us
      \[
                (\where{e_2}{\delexp{p}{q}})[\vec{a}_2][x \mapsto v_1]   ≈^{Π}_{Δ(H^{←})} (\where{e_2}{\delexp{p}{q}})[\vec{a}_2][x \mapsto v_2] 
      \]
      Hence proved.
      \item[Case \ruleref{W-App}:] Similar to the "assume" case.
    \item[Case \ruleref{W-TApp}:] Similar to the "assume" case.
    \item[Case \ruleref{W-UnPair}:] Trivial.
    \item[Case \ruleref{W-Case}:] Similar to the "assume" case.
    \item[Case \ruleref{W-BindM}:] Similar to the "assume" case.
    \item[Case \ruleref{W-Assume}:] Similar to the "assume" case.
   \end{description}
\end{proof}

\end{document}